\newif\ifsubmission

\ifsubmission
\documentclass{llncs}

\else
\documentclass[11pt]{article}
\usepackage{fullpage}
\fi

\ifsubmission \else
\usepackage{palatino}
\fi
\usepackage[normalem]{ulem}
\usepackage{framed}
\usepackage[utf8]{inputenc} 
\usepackage[pdfstartview=FitH,pdfpagemode=UseNone,colorlinks,linkcolor=blue,filecolor=blue,citecolor=Violet,urlcolor=red]{hyperref}
\usepackage{cmap} 
\usepackage[T1]{fontenc} 
\usepackage{multirow}
\usepackage{float}
\usepackage[section]{placeins}         
\usepackage{mathtools}  
\usepackage[dvipsnames]{xcolor}
\usepackage{amsmath,amsthm,amssymb,algpseudocode,algorithm,cryptocode}
\usepackage{mathrsfs} 
\usepackage[capitalize]{cleveref}
\usepackage{dashbox} 
\usepackage{braket}
\usepackage{physics}
\usepackage{xspace}  
\usepackage{braket}
\usepackage[normalem]{ulem}
\usepackage{soul,xcolor}

\usepackage{tikz}

\usepackage{soul}
\usepackage{enumitem}
\usepackage{breakcites}

\newif\ifnotes
\notestrue
\newcommand{\authnote}[3]{\textcolor{#3}{[{\footnotesize {\bf #1:} { {#2}}}]}}
\newcommand{\james}[1]{\ifnotes \authnote{J}{#1}{blue} \fi}
\newcommand{\dakshita}[1]{\ifnotes \authnote{D}{#1}{orange} \fi}
\newcommand{\akshay}[1]{\ifnotes \authnote{A}{#1}{red} \fi}

\newif\ifforlater
\forlatertrue

\ifsubmission

\newtheorem{axiom}[theorem]{Axiom}
\newtheorem{physicsaxiom}[theorem]{Physics Axiom}
\newtheorem{importedtheorem}[theorem]{Imported Theorem}
\newtheorem{importedlemma}[theorem]{Imported Lemma}
\newtheorem{informaltheorem}[theorem]{Informal Theorem}
\newtheorem{physicstheorem}[theorem]{Physical Theorem}

\newtheorem{claim}[theorem]{Claim}
\newtheorem{subclaim}[theorem]{SubClaim}

\newtheorem{fact}[theorem]{Fact}

\newtheorem{construction}[theorem]{Construction}
\Crefname{importedtheorem}{Imported Theorem}{Imported Theorems}
\Crefname{importedlemma}{Imported Lemma}{Imported Lemma}
\Crefname{theorem}{Theorem}{Theorems}
\Crefname{proposition}{Proposition}{Propositions}
\Crefname{claim}{Claim}{Claims}
\Crefname{subclaim}{SubClaim}{SubClaims}
\Crefname{subsubclaim}{SubSubClaim}{SubSubClaims}
\Crefname{lemma}{Lemma}{Lemmas}
\Crefname{conjecture}{Conjecture}{Conjectures}
\Crefname{corollary}{Corollary}{Corollaries}
\Crefname{construction}{Construction}{Constructions}
\Crefname{property}{Property}{Properties}

\theoremstyle{definition}

\newtheorem{assumption}[theorem]{Assumption}
\newtheorem{notation}[theorem]{Notation}
\Crefname{definition}{Definition}{Definitions}
\Crefname{assumption}{Assumption}{Assumptions}
\Crefname{notation}{Notation}{Notations}

\theoremstyle{remark}

\newtheorem{comment}[theorem]{Comment}

\Crefname{question}{Question}{Questions}
\Crefname{remark}{Remark}{Remarks}
\Crefname{comment}{Comment}{Comments}
\Crefname{fact}{Fact}{Facts}
\Crefname{step}{Step}{Steps}

\else

\newtheorem{theorem}{Theorem}[section]

\newtheorem{importedtheorem}[theorem]{Imported Theorem}

\newtheorem{informaltheorem}[theorem]{Informal Theorem}

\newtheorem{claim}[theorem]{Claim}

\newtheorem{lemma}[theorem]{Lemma}

\newtheorem{corollary}[theorem]{Corollary}

\newtheorem{definition}[theorem]{Definition}
\newtheorem{remark}[theorem]{Remark}
\Crefname{importedtheorem}{Imported Theorem}{Imported Theorems}
\Crefname{importedlemma}{Imported Lemma}{Imported Lemmas}
\Crefname{theorem}{Theorem}{Theorems}
\Crefname{proposition}{Proposition}{Propositions}
\Crefname{claim}{Claim}{Claims}
\Crefname{subclaim}{SubClaim}{SubClaims}
\Crefname{subsubclaim}{SubSubClaim}{SubSubClaims}
\Crefname{lemma}{Lemma}{Lemmas}
\Crefname{conjecture}{Conjecture}{Conjectures}
\Crefname{corollary}{Corollary}{Corollaries}
\Crefname{construction}{Construction}{Constructions}
\Crefname{property}{Property}{Properties}

\theoremstyle{definition}

\Crefname{definition}{Definition}{Definitions}
\Crefname{assumption}{Assumption}{Assumptions}
\Crefname{notation}{Notation}{Notations}

\theoremstyle{remark}

\Crefname{question}{Question}{Questions}
\Crefname{comment}{Comment}{Comments}
\Crefname{fact}{Fact}{Facts}
\Crefname{step}{Step}{Steps}

\fi

\newcommand{\secp}{\lambda}

\def\cA{{\cal A}}
\def\cB{{\cal B}}

\def\cD{{\cal D}}

\def\cF{{\cal F}}

\def\cH{{\cal H}}

\def\cL{{\cal L}}

\def\cP{{\cal P}}

\def\cR{{\cal R}}
\def\cS{{\cal S}}

\def\cV{{\cal V}}

\def\cX{{\cal X}}
\def\cY{{\cal Y}}

\def\regA{{\cal A}}

\def\regD{{\cal D}}

\def\regK{{\cal K}}

\def\regP{{\cal P}}

\def\regR{{\cal R}}
\def\regS{{\cal S}}

\def\regX{{\cal X}}
\def\regY{{\cal Y}}

\def\bbC{{\mathbb C}}

\def\bbI{{\mathbb I}}

\def\bbN{{\mathbb N}}

\def\poly{{\rm poly}}

\def\negl{{\rm negl}}

\newcommand{\Setup}{\mathsf{Setup}}

\newcommand{\Com}{\mathsf{Com}}

\newcommand{\Sim}{\mathsf{Sim}}

\newcommand{\Verify}{\mathsf{Verify}}

\DeclareMathOperator*{\expectation}{\mathbb{E}}
\newcommand{\E}{\expectation}

\newcommand{\intersect}{\cap}

\newcommand{\ek}{\mathsf{ek}}

\newenvironment{boxfig}[2]{\begin{figure}[#1]\fbox{
    \begin{minipage}{\linewidth}
    \vspace{0.2em}\makebox[0.025\linewidth]{}    \begin{minipage}{0.95\linewidth}{{#2 }}
    \end{minipage}\vspace{0.2em}\end{minipage}}}{\end{figure}}

\newcommand{\pprotocol}[4]{
\begin{boxfig}{h!}{
\begin{center}
\textbf{#1}
\end{center}
    {\footnotesize#4}
\vspace{0.2em} } \caption{\label{#3} #2}
\end{boxfig}
}

\newcommand{\protocol}[4]{
\pprotocol{#1}{#2}{#3}{#4} }

\newcommand{\crs}{\mathsf{crs}}

\newcommand{\state}{\mathsf{st}}

\renewcommand{\partial}{\mathsf{partial}}

\newcommand{\abort}{\mathsf{abort}}

\newcommand{\Ext}{\mathsf{Ext}}

\newcommand{\TD}{\mathsf{TD}}

\newcommand{\Exp}{\mathsf{Exp}}

\newcommand{\Adv}{\mathsf{Adv}}

\newcommand{\OT}{\mathsf{OT}}

\newcommand{\maj}{\mathsf{maj}}

\newcommand{\sA}{\mathsf{A}}
\newcommand{\sB}{\mathsf{B}}

\newcommand{\sD}{\mathsf{D}}

\newcommand{\msg}{\mathsf{msg}}
\newcommand{\ctl}{\mathsf{ctl}}

\newcommand{\NIZK}{\mathsf{NIZK}}
\newcommand{\Prove}{\mathsf{Prove}}
\newcommand{\Ver}{\mathsf{Ver}}
\newcommand{\PRG}{\mathsf{PRG}}
\newcommand{\hk}{\mathsf{hk}}

\newcommand{\cm}{\mathsf{cm}}

\newcommand{\ck}{\mathsf{ck}}

\newcommand{\ExtGen}{\mathsf{ExtGen}}
\newcommand{\Samp}{\mathsf{Samp}}
\newcommand{\ROT}{\mathsf{ROT}}
\newcommand{\CI}{\mathsf{CI}}
\newcommand{\CL}{\mathsf{CL}}
\newcommand{\QU}{\mathsf{QU}}

\newcommand{\hw}{\mathsf{hw}}
\newcommand{\nonnegl}{\mathsf{non}\text{-}\mathsf{negl}}
\newcommand{\ots}{\mathsf{ots}}
\newcommand{\C}{\mathsf{C}}
\newcommand{\NC}{\mathsf{NC}}

\newcommand{\EQ}{\mathsf{EQ}}
\newcommand{\QMA}{\mathsf{QMA}}
\newcommand{\pre}{\mathsf{pre}}
\newcommand{\chck}{\mathsf{check}}
\newcommand{\post}{\mathsf{post}}
\newcommand{\NAND}{\mathsf{NAND}}
\begin{document}
\setstcolor{red}
\ifsubmission
\title{Secure Computation with Shared EPR Pairs\\ \Large{(Or: How to Teleport in Zero-Knowledge)}}
\author{}

\institute{}
\else
\title{Secure Computation with Shared EPR Pairs\\ \Large{(Or: How to Teleport in Zero-Knowledge)}}
\author{James Bartusek\thanks{UC Berkeley. Email: \texttt{bartusek.james@gmail.com}} \
\and Dakshita Khurana\thanks{UIUC. Email: \texttt{dakshita@illinois.edu}} \ \and Akshayaram Srinivasan\thanks{Tata Institute of Fundamental Research. Email: \texttt{akshayaram.srinivasan@tifr.res.in}}}
\date{}
\fi
\maketitle


\begin{abstract}
    
    Can a sender non-interactively transmit one of two strings to a receiver without knowing which string was received? Does there exist minimally-interactive secure multiparty computation that only makes (black-box) use of symmetric-key primitives? We provide affirmative answers to these questions in a model where parties have access to shared EPR pairs, thus demonstrating the cryptographic power of this resource.

    
    
    \begin{itemize}
    
    \item First, we construct a one-shot (i.e., single message) string oblivious transfer (OT) protocol with random receiver bit in the shared EPR pairs model, assuming the (sub-exponential) hardness of LWE. 
    
    Building on this, we show that {\em secure teleportation through quantum channels} is possible. Specifically, given the description of any quantum operation $Q$, a sender with (quantum) input $\rho$ can send a single classical message that securely transmits $Q(\rho)$ to a receiver. That is, we realize an ideal quantum channel that takes input $\rho$ from the sender and provably delivers $Q(\rho)$ to the receiver without revealing any other information. 

    This immediately gives a number of applications in the shared EPR pairs model: (1) non-interactive secure computation of unidirectional \emph{classical} randomized functionalities, (2) NIZK for QMA from standard (sub-exponential) hardness assumptions, and (3) a non-interactive \emph{zero-knowledge} state synthesis protocol.
    

    \item Next, we construct a two-round (round-optimal) secure multiparty computation protocol for classical functionalities in the shared EPR pairs model that is \emph{unconditionally-secure} in the (quantum-accessible) random oracle model. 
    
    \end{itemize}
    
    Classically, both of these results cannot be obtained without some form of correlated randomness shared between the parties, and the only known approach is to have a trusted dealer set up random (string) OT correlations. In the quantum world, we show that shared EPR pairs (which are simple and can be deterministically generated) are sufficient. At the heart of our work are novel techniques for making use of entangling operations to generate string OT correlations, and for instantiating the Fiat-Shamir transform using correlation-intractability in the quantum setting.

    
\end{abstract}

\ifsubmission
\else
\fi
\newpage
\section{Introduction}

Understanding the nature of shared entanglement is one of the most prominent goals of quantum information science, and its study has repeatedly unearthed surprisingly strong properties. A remarkable example of this is the quantum teleportation protocol of \cite{PhysRevLett.70.1895}, which demonstrated that shared EPR pairs \cite{PhysRev.47.777}, the most basic entangled resource, are ``complete'' for quantum communication using classical channels. That is, if Alice and Bob share EPR pairs a priori, then Alice can communicate an arbitrary state $\rho$ to Bob by sending just a single classical message. In particular, this result positions shared EPR pairs at the center of proposals for building a quantum internet. 

\subsection{Our contributions}

In this work, we investigate the \emph{cryptographic} power of shared EPR pairs. 

\paragraph{Secure Teleportation through a Quantum Channel.} First, we revisit the setting of quantum teleportation, which shows that shared EPR pairs and one-way classical communication give rise to a quantum channel implementing the identity map $\rho \to \rho$. We ask: what if Alice would instead like to send her state $\rho$ to Bob through some arbitrary quantum map $\rho \rightarrow Q(\rho)$?\footnote{We will also allow for preserving entanglement that $\rho$ may have with its environment, so technically we consider $Q$ to map a state on Alice's input register $\cA$ to a state on Bob's output register $\cB$.} 

Note that this is trivial given quantum teleportation if we allow either Alice or Bob to compute the map $\rho \to Q(\rho)$ for themselves. However, we are interested in guaranteeing that the effect of the protocol would be (computationally) ``no different'' than the effect of Alice inputting $\rho$ to an ``ideal'' channel $Q$, and Bob receiving $Q(\rho)$ on the other side, \emph{even if} Alice or Bob attempt to save extra information from or deviate from the protocol. In particular, we require each of the following three properties to hold against arbitrarily malicious adversaries: (1) Alice would not learn any side information created during the computation of $Q(\rho)$, (2) Bob would learn nothing about $\rho$ beyond $Q(\rho)$, and (3) Bob would be convinced that the state he received was actually computed as the output of the map $Q$ (on some input $\rho$). We show that this is possible under the sub-exponential hardness of learning with errors (LWE), a  standard post-quantum security assumption.

\begin{informaltheorem}
For any efficient quantum map $Q$, there exists a protocol for ``secure teleportation through $Q$'' in the shared EPR pairs model assuming the sub-exponential hardness of LWE. That is, there exists a one-shot\footnote{We use one-shot, one-message, and non-interactive interchangeably to refer to a protocol that consists of a single message from a sender to a receiver.} protocol in the shared EPR pairs model that computes the ideal functionality $\rho \to Q(\rho)$.
\end{informaltheorem}

\paragraph{Building Block: One-shot String OT.} The main building block for this protocol, and the key technical contribution of this paper, is a one-shot protocol for (random receiver bit) \emph{string oblivious transfer} (OT) in the shared EPR pairs model, which realizes an ideal funtionality that takes two strings $m_0,m_1$ from a sender Alice, and delivers $(b,m_b)$ to Bob for a uniformly random bit $b$.\footnote{Note that it is impossible to obtain a one-shot protocol for fixed receiver bit OT, since Bob does not send any message.}

\begin{informaltheorem}
\label{infthm:stringOT}
Assuming the sub-exponential hardness of LWE, there exists a simulation-secure one-shot protocol for (random receiver bit) string OT in the shared EPR pairs model. 
\end{informaltheorem}

Given such an OT protocol, we rely on two key previous results to obtain our final implication to secure teleportation through quantum channels: (1) \cite{C:GIKOS15} showed how to construct a one-message protocol for secure computation of any unidirectional \emph{classical} randomized functionality $f$ that maps $x \to f(x;r)$ given a one-message protocol for string OT, and (2) \cite{BCKM} (building on the work of \cite{10.1145/3519935.3520073}) showed (implicitly) how to construct a one-message protocol for secure computation of any unidirectional \emph{quantum} functionality given a one-message protocol for unidirectional classical functionalities.

\paragraph{Correlation Interactability.} There have been many recent works that show how to instantiate random oracles with a concrete hash function family and base the security of (classical) primitives such as NIZKs and SNARGs on standard cryptographic assumptions~\cite{CCHLRRW,C:PeiShi19,BKM20,JJ21,JKKZ21,HLR,CJJ21b,KalaiVZ21,CJJ21,EC:HJKS22,EC:part1,EC:part2}. These works proceed by constructing a special hash function family that satisfies the cryptographic notion of \emph{correlation-intractability}~\cite{JACM:CanGolHal04}. Ours is the first to apply correlation-intractability to a setting that involves \emph{quantum communication}, addressing technical barriers along the way. In fact, we obtain our one-message string OT protocol (refer to Informal Theorem~\ref{infthm:stringOT}) by utilizing correletion-intractability, which we discuss further in \cref{sec:tech-overview}.

\paragraph{The Multiparty Setting.} Next, we consider the \emph{multiparty} setting, where all pairs of parties have access to shared EPR pairs. If each party has their own private input $x_i$, and their goal is to compute $C(x_1,\dots,x_n)$ for some (classical) circuit $C$, they will have to use at least two rounds of interaction as single round protocols are susceptible to resetting attacks~\cite{C:HalLinPin11}. 

Classically, two rounds are known to suffice for secure multiparty computation, under the (minimal) assumption that two-round (chosen-input) oblivious transfer \cite{EC:GarSri18a,EC:BenLin18} protocols exist.\footnote{In chosen-input OT, the receiver specifies their input bit $b$, and they receive the message $m_b$. We contrast this with the notion of OT discussed above, where the receiver's bit $b$ is chosen uniformly at random.} In the classical setting, OT is a ``public-key-style'' primitive that provably cannot be built from ``minicrypt-style'' primitives, including hash functions modeled as a random oracle \cite{C:ImpRud88}. On the other hand, a line of work beginning with \cite{FOCS:CreKil88} and culminating with \cite{10.1007/978-3-030-84242-0_17,10.1007/978-3-030-77886-6_18} established that with quantum communication, it is possible to obtain oblivious transfer, and thus multiparty computation, from one-way functions or potentially even weaker assumptions~\cite{C:JiLiuSon18,AQY22,MY}. However, these protocols require many rounds, and the possibility of achieving \emph{round-optimal} (two-round) secure computation without public-key primitives was left open.

In this work, we show that round-optimal secure computation that makes black-box use of symmetric-key primitives (specifically, a random oracle) can be obtained in the shared EPR pairs model. 

\begin{informaltheorem}
There exists a two-round secure multiparty computation protocol in the shared EPR pairs model with either of the following properties.
\begin{itemize}
    \item Unconditional security in the quantum-accesible random oracle model (QROM).
    \item Computational security assuming (the black-box use of) non-interactive extractable commitments and hash functions that are correlation-intractable for efficient functions.
\end{itemize}
\end{informaltheorem}



\paragraph{Discussion: Towards Weaker Correlated Randomness.} In the classical world, it can be shown that without any form of correlated randomness shared between the parties, it is impossible to obtain either one-shot OT or two-round MPC (even with public-key primitives). Furthermore, we show in \cref{subsec:impossibility} that one-shot (random receiver bit) string OT is impossible in the classical common reference string model, even when parties can compute and communicate quantumly.
On the other hand, we remark that both our results can be obtained (even in the classical world) with an ``OT correlations'' setup, which assumes that a trusted dealer has sampled random strings $x_0,x_1$ and bit $b$ and delivered $x_0,x_1$ to the sender and $b,x_b$ to the receiver. For the case of string OT, this consequence is immediate and for the case of two-round MPC, this result follows from the work of Garg et al.~\cite{TCC:GarIshSri18}. 

Our results state that in the quantum world, shared EPR pairs are sufficient to obtain (i) one-shot (random receiver bit) string OT and (ii) two-round MPC from symmetric-key primitives. As noted in \cite{ABKK}, shared EPR pairs are a fundamentally different resource than OT correlations. Indeed, OT correlations are \emph{specific} to OT, while, as indicated above, shared EPR pairs are known to be broadly useful and have been widely studied independent of OT.  Moreover, an OT correlations setup requires private (hidden) randomness, while generating EPR pairs is a fully deterministic quantum process.\footnote{In particular, any (even semi-honest) dealer that sets up OT correlations can learn the parties’ private inputs by observing the resulting transcript of communication, while this is not necesarily true of an EPR pair setup, by monogamy of entanglement. We also remark that obtaining OT correlations from \emph{any} deterministically generated shared quantum state is non-trivial. In particular, if the parties shared a (deterministically generated) superposition over classical OT correlations, the receiver could simply decide not to measure the register holding their choice bit, and obtain a superposition over the sender's strings, which violates the security of OT. } Our work can thus be viewed as a step towards realizing secure computation protocols using weaker forms of correlated randomness. Finally, we remark that, unlike the case of one-shot OT, it may be possible to achieve two-round MPC from symmetric-key primitives in the classical common reference string model (i.e., without shared EPR pairs), and we leave this as an intriguing open question for future study.

\subsection{Applications} We now discuss several applications of our one-shot string OT construction and secure teleportation through quantum channel protocol. 

\paragraph{Non-Interactive Computation of Unidirectional Functionalities.} The study of non-interactive protocols for unidirectional classical functionalities was initiated by \cite{C:GIKOS15}. Such functionalities are defined by a classical circuit $f$, take an input $x$ from the sender, (potentially) sample some random coins $r$, and deliver $f(x;r)$ to the receiver. They showed the possibility (or impossibility) of achieving them in a model where the sender and the receiver have access to an \emph{one-way communication channel}. In particular, they showed that ideal string OT channel suffices to build non-interactive secure computation of unidirectional classical functionalities. On the other hand, the work of Agrawal et al.~\cite{AC:AIKNPPR20} showed that bit OT channels provably do \emph{not} suffice for non-interactive secure computation.\footnote{A followup work of \cite{C:AIKNPPR21} showed that, assuming \emph{ideal obfuscation}, there exists a protocol over a bit OT channels with (non-standard) $1/\poly(\secp)$ security.}

Using our one-shot string OT construction, we can instantiate the results of Garg et al.~\cite{C:GIKOS15} and obtain non-interactive secure computation of unidirectional functionalities in the shared EPR pairs model, assuming sub-exponential LWE. 


The works of \cite{C:GIKOS15,C:AIKNPPR21} also discuss several applications of non-interactive secure computation of unidirectional classical functionalities, and we mention one intriguing application here. The modern internet relies on a public-key infrastructure, where certificate authorities validate public keys by signing them under their own signing key.\footnote{Note that despite the existence of quantum key distribution \cite{BenBra84}, public-key infrastructure would still likely be required for the quantum internet, since QKD requires authenticated classical channels.} A single message protocol for unidirectional classical functionalities would enable key authorities to \emph{non-interactively} generate and send freshly sampled and signed public key secret key pairs to clients, \emph{without} learning the client's secret key. Moreover, the client would not learn the secret signing key of the authority who sent their fresh pair. Thus, we show that there is a truly non-interactive solution to this widespread key certification functionality in a world where nodes are connected by shared EPR pairs.\footnote{We do stress that our model assumes the EPR pairs are generated honestly, for example by an honest network administrator. Otherwise, such secure one-message protocols would be impossible to achieve.}

\paragraph{NIZKs for QMA.} Our secure telportation through quantum channels immediately gives a non-interactive zero-knowlede (NIZK) for QMA in the shared EPR pairs model, by letting the channel $Q$ compute a QMA verification circuit and output the resulting bit to the receiver. The only previous NIZK for QMA in the shared EPR pairs model is due to \cite{MY22}, who argued security in the quantum random oracle model.\footnote{We also remark that \cite{BM22} achieve NIZK for QMA in the (incomparable) common reference string model, but they argue security using classical oracles, or alternatively assuming indistinguishability obfuscation and the non-black-box use of a hash function modeled as a random oracle.} Thus, we obtain the first such protocol from a standard (sub-exponential) hardness assumption.

\paragraph{Non-Interactive Zero-Knowledge State Synthesis.} Many recent works consider the problem of quantum \emph{state synthesis} \cite{https://doi.org/10.48550/arxiv.1607.05256,Rosenthal2022InteractivePF,10.4230/LIPIcs.CCC.2022.5}, which studies the efficiency of preparing a complex quantum state with the help of an oracle or untrusted powerful prover. That is, given the implicit description of a quantum circuit $Q$, can a verifier prepare $\ket{\psi} = Q\ket{0^n}$ with the help of a prover, and be convinced that they end up with the correct state? 

In fact, \cite{Rosenthal2022InteractivePF} asked whether there is any meaningful notion of \emph{zero-knowledge} state synthesis. In this work, we propose one way to define zero-knowledge state synthesis. Roughly, we consider any \emph{family} of circuits $\{Q_w\}_w$ parameterized by a potentially secret witness $w$, and require that a prover help the verifier prepare $\ket{\psi_w} = Q_w\ket{0^n}$ without leaking the witness $w$. We formalize our definition in \cref{subsec:applications} and show that our secure teleportation protocol immediately gives a \emph{one-message} solution to this task in the shared EPR pairs model. We stress that there may be other meaningful ways to define zero-knowledge state synthesis, and we leave a more thorough exploration of definitions and applications of zero-knowledge state synthesis to future work.

\paragraph{Non-Interactive Quantum Cryptography.} Finally, we observe that the full power of non-interactive secure computation of unidirectional quantum functionalities gives rise to quantum analogues of the classical applications mentioned above. For example, a certificate authority could non-interactively prepare and send signed key pairs for encryption schemes with \emph{uncloneable} or \emph{revocable} decryption keys \cite{cryptoeprint:2020/877,10.1007/978-3-030-84242-0_20,AKNYY,cryptoeprint:2023/265,cryptoeprint:2023/325}, where decryption keys are quantum states that can either provably not be distributed or verifiably be destroyed. The novel guarantee is that even the certificiate authority itself will not learn the (description of) the decrpytion key.\footnote{In this setting, \emph{publicly-verifiable} revocation \cite{cryptoeprint:2023/265} seems crucial to ensure that no one need know the classical description of the secret key.} As another example, a bank could non-interactively distribute signed \emph{quantum money} states (technically, the serial number would be signed), without ever learning the classical description of the state. In particular, while valid money states could be provably generated and distributed non-interactively, no one (not even the bank) would ever learn a classical description that would enable cloning. 

\subsection{Related Works}

This work continues a long line of research that studies the power of shared entanglement as a resource. We show that shared EPR pairs, which already have a long history of study in communication \cite{PhysRevLett.69.2881,PhysRevLett.70.1895}, cryptography \cite{PhysRevLett.67.661,10.1007/978-3-540-24587-2_20,https://doi.org/10.48550/arxiv.2204.02265,ABKK}, and error-correction \cite{BDH}, can be leveraged to obtain perhaps surprisingly powerful secure computation tasks.

We also compare our results with the prior work of \cite{ABKK}, which also studies oblivious transfer in the shared EPR pairs model. They achieve a one-message protocol for \emph{bit} OT, where the sender's inputs are one bit each, and explicitly leave open the problem of building \emph{string} OT, which we address in this work. We note that bit OT is not known to be complete for one-message secure computation~\cite{C:GIKOS15,AC:AIKNPPR20}. Moreover, security of the protocols in \cite{ABKK} are all argued in the quantum random oracle model, while we argue security \emph{without} random oracles, and based on concrete properies of hash functions instead. 

\paragraph{Concurrent Work.} Finally, we mention a concurrent and independent work~\cite{Leoetal} that was posted recently to the arXiv. Their results and techniques are orthogonal to ours: in particular, they obtain two-message OT in the CRS model assuming NIZK (and an assumption on hash functions), whereas we obtain one-message OT from sub-exponential LWE, as well as unconditional two-round MPC in the QROM, both in the shared EPR pairs model. We do not believe that (a simple modification of) either work's results or techniques immediately subsumes or improves results in the other. We also remark that both our work and~\cite{Leoetal} leave open the intriguing question of obtaining minimally-interactive (two-round) MPC in the CRS model without the use of public-key primitives. 
\section{Technical Overview}\label{sec:tech-overview}
We give an overview of the key techniques used to obtain our results.
\subsection{One-shot string OT}

In this subsection, we focus on our key technical contribution, which is a construction of one-shot string OT in the shared EPR pairs model. Throughout this section, we define \emph{one-shot string OT} as a one-message protocol that takes two strings $m_0,m_1$ from the sender, and delivers $m_b$ to the receiver for a random bit $b \gets \{0,1\}$. For more discussion on our applications, we refer the reader to \cref{subsec:applications}. 


\paragraph{A string OT skeleton.} As mentioned earlier, \cite{ABKK} constructed a one-shot \emph{bit} OT protocol in the shared EPR pairs model (where the sender's inputs are one bit each). However, their techniques don't appear to extend easily to the setting of one-shot \emph{string} OT, for arbitrary length strings. In fact, \cite{C:GIKOS15,AC:AIKNPPR20} showed that in the non-interactive setting, it is impossible to obtain string OT from bit OT. We additionally observe that prior quantum OT templates \cite{FOCS:CreKil88,ABKK} only obtain ``bitwise'' correlations by sending unentangled BB84 states or by immediately measuring each EPR pair independently. 

To get around this barrier, our idea is to directly obtain string correlations from shared entanglement. This can be done by first \emph{entangling} the separate EPR pairs in a special way before performing measurements. 

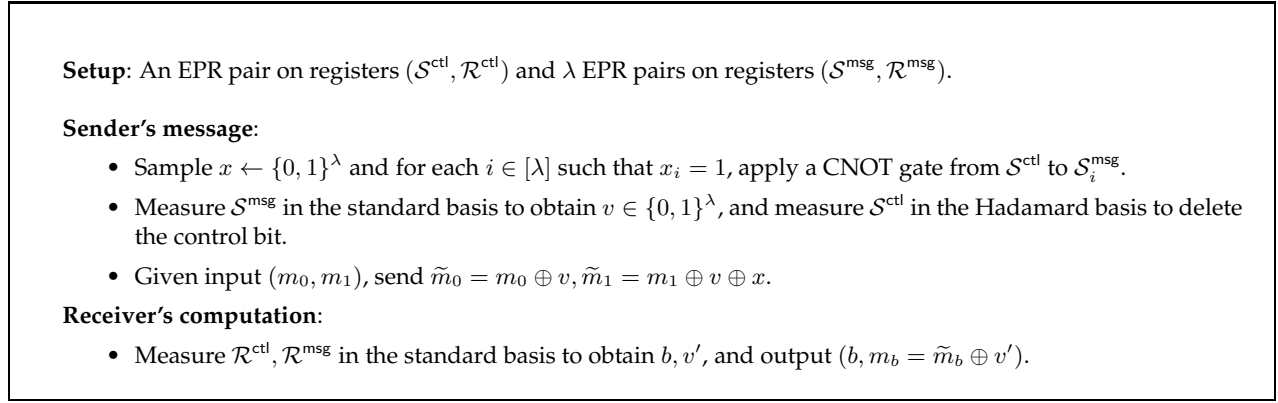
\protocol{}{An (insecure) skeleton for one-shot string OT}{fig:skeleton}{
\textbf{Setup}: An EPR pair on registers $(\cS^\ctl,\cR^\ctl)$ and $\secp$ EPR pairs on registers $(\cS^\msg,\cR^\msg)$.\\

\textbf{Sender's message}: 
\begin{itemize}
    \item Sample $x \gets \{0,1\}^\secp$ and for each $i \in [\secp]$ such that $x_i = 1$, apply a CNOT gate from $\cS^\ctl$ to $\cS_i^\msg$.
    \item Measure $\cS^\msg$ in the standard basis to obtain $v \in \{0,1\}^\secp$, and measure $\cS^\ctl$ in the Hadamard basis to delete the control bit.
    \item Given input $(m_0,m_1)$, send $\widetilde{m}_0 = m_0 \oplus v, \widetilde{m}_1 = m_1 \oplus v \oplus x$.
\end{itemize}

\textbf{Receiver's computation}:
\begin{itemize}
    \item Measure $\cR^\ctl,\cR^\msg$ in the standard basis to obtain $b,v'$, and output $(b,m_b = \widetilde{m}_b \oplus v')$.
\end{itemize}
}

Our approach is illustrated in \cref{fig:skeleton}. Note that after the sender applies the random CNOT gates and measures $\cS^\msg$ to obtain $v$, the remaining state of the system is 

\[\frac{1}{\sqrt{2}}\ket{0}_{\cS^\ctl}\ket{0}_{\cR^\ctl}\ket{v}_{\cR^\msg} + \frac{1}{\sqrt{2}}\ket{1}_{\cS^\ctl}\ket{1}_{\cR^\ctl}\ket{v \oplus x}_{\cR^\msg}.\] Thus, tracing out $\cS^\ctl$, we see that the receiver has a uniform mixture over $\ket{0,v}$ and $\ket{1,v \oplus x}$, where $v,v\oplus x$ are uniformly random strings from their perspective, exactly as desired. Unfortunately, since the sender's control register is entangled with the receiver's, the sender could know exactly which bit $b$ the receiver obtains by measuring $\cS^\ctl$ in the standard basis. Thus, we instead ask that the sender ``delete'' their control bit by measuring it in the \emph{Hadamard} basis. Of course, a malicious (or even specious) sender may not follow these instructions, rendering this protocol insecure. However, this protocol serves as the foundation for our eventual secure realization of one-shot string OT.

\paragraph{Measurement check.} Next, we add a mechansim for ``forcing'' the sender to delete their control bit. We build on the commitment-based cut-and-choose approach~\cite{FOCS:CreKil88,C:BouFeh10,ABKK} as follows. Suppose the sender really did behave honestly, and measured $\cS^\ctl$ in the Hadamard basis to obtain a bit $h$. Then, the state on the receiver's side will be \[\ket{\psi_{v,x,h}} \coloneqq \frac{1}{\sqrt{2}}\left(\ket{0,v} + (-1)^h\ket{1,v \oplus x}\right).\] So if the receiver was given $(v,x,h)$, they could measure $(\cR^\ctl,\cR^{\msg})$ in the \[\left\{\dyad{\psi_{v,x,h}}{\psi_{v,x,h}}, \bbI - \dyad{\psi_{v,x,h}}{\psi_{v,x,h}}\right\}\] basis and accept if the first outcome is observed. Of course, sending $(v,x,h)$ to the receiver would render the protocol insecure because the receiver could now obtain both $v$ and $v \oplus x$. Instead, we apply a variant of the Fiat-Shamir-based non-interactive measurement check subprotocol of \cite{ABKK}, using a non-interactive commitment scheme $\Com$ and a hash function $H$:
\begin{itemize}
    \item Repeat the skeleton protocol $\ell$ times in parallel, and have the sender commit to all descriptions $\cm_1 = \Com(v_1,x_1,h_1),\dots,
    \cm_\ell = \Com(v_\ell,x_\ell,h_\ell)$.
    \item Hash $T = H(\cm_1,\dots,\cm_\ell)$ to obtain a subset $T \subset [\ell]$ of commitments.
    \item The sender sends $(\cm_1,\dots,\cm_\ell)$ along with openings to $\{\cm_i\}_{i \in T}$.
    \item For each $i \in T$, the receiver measures registers $\cR_i^\ctl,\cR_i^\msg$ in basis \[\left\{\dyad{\psi_{v_i,x_i,h_i}}{\psi_{v_i,x_i,h_i}}, \bbI - \dyad{\psi_{v_i,x_i,h_i}}{\psi_{v_i,x_i,h_i}}\right\}\] and aborts if any of these measurements reject. Otherwise, the parties continue the protocol using indices $i \in \overline{T}$.
\end{itemize}

Now, assuming $H$ behaves as a random oracle, we should be able to claim that conditioned on the receiver not aborting, their states on registers $\{\cR_i^\ctl,\cR_i^\msg\}_{i \in \overline{T}}$ should be ``close'' to the honest states $\{\ket{\psi_{v_i,x_i,h_i}}\}_{i \in \overline{T}}$. We can make this precise by arguing that after an appropriate change of basis, the states $\{\cR_i^\ctl\}_{i \in \overline{T}}$ are in a superposition of Hadamard basis states that are close in Hamming distance to the honest state $H^{\otimes|\overline{T}|}\ket{h_{\overline{T}}}$, where $h_{\overline{T}}$ are the bits $\{h_i\}_{i \in \overline{T}}$. If this is the case, then by the ``XOR extractor'' lemma of \cite{ABKK}, measuring these bits in the standard basis and XORing the results together would produce a bit $b$ that is truly uniformly random and independent of the sender's view. Thus, we should be able to extract a perfectly random receiver's bit by \emph{combining} correlations obtained from multiple instances $i \in \overline{T}$ of the skeleton protocol.

\paragraph{Defining two sender strings.} Unfortunately, if we XOR together the correlations from all $i \in \overline{T}$, it is no longer clear how to define the two sender strings. Indeed, the receiver will obtain one out of two of each pair $\{(v_i,v_i \oplus x_i)\}_{i \in \overline{T}}$, which means one out of $2^{|T|}$ possible sets of strings! However, note that if the sender had used the same offset $x$ for each repetition, then if the receiver XORs together one out of two of each $\{(v_i,v_i \oplus x)\}_{i \in \overline{T}}$, they obtain either $\bigoplus_{i \in \overline{T}}v_i$ or $x \oplus \bigoplus_{i \in \overline{T}}v_i$ depending on the parity of their choice bits. Of course, since we are opening the commitments on indices $i \in T$, the receiver would learn $x$, rendering this approach insecure. 

Our solution is to make use of this ``common offset'' approach in a less direct manner. In addition to the $\ell$ repetitions of the skeleton protocol described above, the sender will sample an independent collection of strings $t_1,\dots,t_\ell,\Delta$ and include commitments \[\widehat{\cm}_{1,0} = \Com(t_1),\widehat{\cm}_{1,1} = \Com(t_1 \oplus \Delta),\dots,\widehat{\cm}_{\ell,0} = \Com(t_\ell),\widehat{\cm}_{\ell,1} = \Com(t_\ell \oplus \Delta)\] in their message. Then, the sender will use the random strings $(v_1,v_1 \oplus x_1),\dots,(v_\ell,v_\ell \oplus x_\ell)$ to mask the \emph{openings} for the commitments $(\widehat{\cm}_{1,0},\widehat{\cm}_{1,1}),\dots,(\widehat{\cm}_{\ell,0},\widehat{\cm}_{\ell,1})$. The effect of this is that the receiver will be able to open one out of two of each pair of commitments $\{\widehat{\cm}_{i,0},\widehat{\cm}_{i,1}\}_{i \in \overline{T}}$, obtaining either $\bigoplus_{i \in \overline{T}} t_i$ or $\Delta \oplus \bigoplus_{i \in \overline{T}} t_i$.

Finally, to maintain security, we require that the sender computes a non-interactive zero-knowledge (NIZK) argument that they sampled $\{\widehat{\cm}_{i,b}\}_{i \in [\ell],b \in \{0,1\}}$ as commitments to pairs of strings that all share the same offset $\Delta$.

\paragraph{Using correlation-intractability.} This nearly completes the description of our protocol. Turning to the security proof, our goal is to reduce to a standard cryptographic assumption. Fortunately, the flavors of commitments and zero-knowledge we require are known from LWE. However, we also need some security from the Fiat-Shamir hash function $H$. In \cite{ABKK} this hash was modeled as a random oracle, and it was left open whether one could obtain security in the plain model.

Classically, a recent exciting line of work has shown how to securely instantiate the Fiat-Shamir transform from standard cryptographic assumptions in many settings \cite{CCHLRRW,C:PeiShi19,BKM20,JJ21,JKKZ21,HLR,CJJ21b,KalaiVZ21,CJJ21,EC:HJKS22,EC:part1,EC:part2}. These works rely on the notion of \emph{correlation-intractability} (CI), which is a property of the hash function $H$ requiring that for some relation $R$ over inputs and outputs, the adversary can't find any input $x$ such that $(x,H(x)) \in R$. In particular, it is known how to obtain CI for efficiently computable \emph{functions} from LWE \cite{CCHLRRW,C:PeiShi19}. Moreover, \cite{HLR} showed to extend this result to CI for efficiently verifiable product relations $R$, where the range of $H$ is the $t$-wise cartesian product of a set $Y$, and each input $x$ is associated with sets $S_{x,1},\dots,S_{x,t} \subset Y$ such that $(x,(y_1,\dots,y_t)) \in R$ iff each $y_i \in S_{x,i}$. The property of efficient verifiability states that there is an efficient (classical) algorithm that, given $(x,i,y_i)$, determines whether $y_i \in S_{x,i}$.

Recall that in our protocol, we apply $H$ to a set of $\ell$ commitments in order to obtain the description of a subset $T \subset [\ell]$ of commitments to open. Intuitively, we want it to be difficult for the sender to find a set of commitments $(\cm_1,\dots,\cm_\ell)$ to strings $(v_1,x_1,h_1),\dots,(v_\ell,x_\ell,h_\ell)$ such that $T = H(\cm_1,\dots,\cm_\ell)$ is a ``bad'' set, meaning that the receiver's registers $\{(\cR_i^\ctl,\cR_i^\msg)\}_{i \in T}$ are ``close'' to the states $\{\ket{\psi_{v_i,x_i,h_i}}\}_{i \in T}$ (so the receiver won't abort) but the registers $\{(\cR_i^\ctl,\cR_i^\msg)\}_{i \in \overline{T}}$ are ``far'' from the states $\{\ket{\psi_{v_i,x_i,h_i}}\}_{i \in \overline{T}}$. Thus, given an input $(\cm_1,\dots,\cm_\ell)$, it appears that determining whether or not a potential output $T$ is ``bad'' requires (at least) applying some \emph{quantum measurement} to the receiver's registers. Unfortunately, all prior work has used CI in a purely classical setting, and extending the notion of efficiently verifiable relation to handle \emph{quantum} verification algorithms appears to be beyond the reach of current techniques (though this may be an interesting direction for future research).

Instead, we take a different approach. Suppose that the sender's choices of $x_1,\dots,x_\ell$ were fixed before the protocol begins. Then, we could pre-measure the receiver's registers even before initializing the malicious sender to obtain $(v_1,h_1),\dots,(v_\ell,h_\ell)$. That is, we could first apply CNOTs from $\cR_i^\ctl$ to each of the qubits in $\cR_i^\msg$ controlled on $x_i$, and then measure $\cR_i^\ctl$ in the Hadamard basis to obtain $h_i$ and measure $\cR_i^\msg$ in the standard basis to obtain $v_i$. Then given just this classical data, we can distinguish between honest commitments $\cm_i$ to $(v_i,x_i,h_i)$ and dishonest commitments $\cm_i$ to some other string (as long as the commitment is efficiently extractable). If we split $\ell$ into $t$ disjoint groups and parse $T$ as $t$ different subsets of $[\ell / t]$, then we can formulate a classically efficiently verifiable product relation $R$ where $((\cm_1,\dots,\cm_\ell),T) \in R$ iff all $\{\cm_i\}_{i \in T}$ are honest and ``many'' $\{\cm_i\}_{i \in \overline{T}}$ are dishonest. 

Now, while we cannot guarantee that a malicious sender will sample any fixed $(x_1,\dots,x_\ell)$, we can \emph{guess} beforehand which $x_1,\dots,x_\ell$ they will use, and simply give up on reducing to CI if the guess is wrong. Using complexity leveraging (and setting the security parameter of the CI hash function large enough), we can hope that this is enough to still break \emph{sub-exponentially-secure} CI. It turns out that this strategy can only be made to work if our guessing loss depends only on the security parameter $\secp$, and \emph{not} on the number of repetitions $\ell$ (which must depend on the level of security required from the CI hash). Thus, we make one final tweak to the protocol. The sender will be required to sample $x_1,\dots,x_\ell$ as the output of a pseudorandom generator with seed $s$ of length $\{0,1\}^\secp$, and prove using the NIZK that they have done so honestly. Then, in the reduction to CI, it suffices to guess a $\secp$-bit string $s$ rather than a $\secp\ell$-bit string $(x_1,\dots,x_\ell)$. This allows us to eventually reduce security to the sub-exponential hardness of LWE.

\paragraph{Unconditional Protocols in the QROM.}
We remark that it appears plausible to obtain more efficient and unconditionally secure variants of our non-interactive protocols in the (quantum) random oracle model. In particular, following \cite{ABKK}, we expect that the measure-and-reprogram technique~\cite{C:DFMS19} in the quantum random oracle model can be used in place of correlation intractability, which would remove the need for sampling $x_1, \ldots, x_\ell$ as the output of a PRG, and remove complexity leveraging in the approach outlined above. It also may be possible to rely on {\em black-box} commit-and-prove sigma protocols (e.g., variants of the protocol in~\cite{TCC:KhuOstSri18}) to prove that commitments to pairs of strings share a common offset, thereby making our protocol black-box and unconditionally secure in the QROM. We leave a formalization and detailed analysis of this approach, and more generally an exploration of one-message protocols in the QROM, to future work.

\subsection{Two-round MPC}

In this section, we give a brief overview of our approach to building two-round MPC in the shared EPR model, which is presented in \cref{sec:MPC}. Our starting point is a three-round chosen-input string OT protocol from \cite{ABKK}, which can be viewed as a two-round protocol in the shared EPR model. In order to use this protocol to build two-round MPC, we take the following steps.

\begin{enumerate}
    \item Show that the protocol is ``black-box friendly''. That is, we split the protocol into an \emph{input-independent} phase that uses both quantum measurements and cryptographic operations, and an \emph{input-dependent} phase that is fully classical and information-theoretic.
    \item Appeal to existing compilers (e.g. \cite{C:CreVanTap95,C:IshPraSah08}) to obtain a ``black-box friendly'' MPC protocol in the shared EPR pair model. Again, we have (1) an input-independent phase at the beginning where every party performs a measurement on their halves of EPR pairs, broadcasts a message, and performs some crytographic checks, and (2) an input-dependent multi-round phase that is entirely classical and information-theoretic.
    \item Use the \cite{EC:GarSri18a} round-compressing compiler and two-round OT in the shared EPR pair model to compress this black-box-friendly protocol into a two-round MPC in the shared EPR pair model. Crucially, the compiler only has to operate on the second (multi-round input-dependent) phase, and thus we obtain a final protocol that makes black-box use of cryptography. 
\end{enumerate}

We stress that to make the above compiler work, we need to start with an OT protocol in which all \emph{cryptographic operations} and \emph{quantum computations} are performed \emph{indepedently} of the parties' inputs and \emph{before} the second message. That is, it does not follow from any two-round quantum OT protocol.


If we start with the protocol from \cite{ABKK} that was proven secure in the quantum random oracle model, then we obtain a final MPC protocol in the quantum random oracle model. In addition, we prove that a slight variant of the \cite{ABKK} protocol is secure \emph{without} random oracles, assuming non-interactive extractable commitments and correlation-intractability for efficient functions. Interestingly, while we use a similar approach as described above, we do not have to resort to sub-exponential assumptions here. Roughly, this is because the \cite{ABKK} protocol is built from ``bitwise'' rather than ``stringwise'' correlations, and it suffices for the reduction to correctly guess a random subset of the adversary's bitwise measurements.

\section{Preliminaries}

Let $\secp$ denote the security parameter. We write $\negl(\cdot)$ to denote any \emph{negligible} function, which is a function $f$ such that for every constant $c \in \mathbb{N}$ there exists $N \in \mathbb{N}$ such that for all $n > N$, $f(n) < n^{-c}$. We write $\nonnegl(\cdot)$ to denote any function $f$ that is not negligible. That is, there exists a constant $c$ such that for infinitely many $n$, $f(n) \geq n^{-c}$. 

\subsection{Quantum information}

A register $\cX$ is a named Hilbert space $\bbC^{2^n}$. A pure quantum state on register $\cX$ is a unit vector $\ket{\psi}^{\cX} \in \bbC^{2^n}$, and we say that $\ket{\psi}^{\cX}$ consists of $n$ qubits. A mixed state on register $\cX$ is described by a density matrix $\rho^{\cX} \in \bbC^{2^n \times 2^n}$, which is a positive semi-definite Hermitian operator with trace 1. 

A \emph{quantum operation} (also referred to as quantum map or quantum channel) $Q$ is a completely-positive trace-preserving (CPTP) map from a register $\cX$ to a register $\cY$, which in general may have different dimensions. That is, on input a density matrix $\rho^{\cX}$, the operation $Q$ produces $\tau^{\cY} \gets Q(\rho^{\cX})$ a mixed state on register $\cY$. We will sometimes write a quantum operation $Q$ applied to a state on register $\cX$ and resulting in a state on register $\cY$ as $\cY \gets Q(\cX)$. Note that we have left the actual mixed states on these registers implicit in this notation, and just work with the names of the registers themselves.

A \emph{unitary} $U: \cX \to \cX$ is a special case of a quantum operation that satisfies $U^\dagger U = U U^\dagger = \bbI^{\cX}$, where $\bbI^{\cX}$ is the identity matrix on register $\cX$. A \emph{projector} $\Pi$ is a Hermitian operator such that $\Pi^2 = \Pi$, and a \emph{projective measurement} is a collection of projectors $\{\Pi_i\}_i$ such that $\sum_i \Pi_i = \bbI$.

Let $\Tr$ denote the trace operator. For registers $\cX,\cY$, the \emph{partial trace} $\Tr^{\cY}$ is the unique operation from $\cX,\cY$ to $\cX$ such that for all $(\rho,\tau)^{\cX,\cY}$, $\Tr^{\cY}(\rho,\tau) = \Tr(\tau)\rho$. The \emph{trace distance} between states $\rho,\tau$, denoted $\TD(\rho,\tau)$ is defined as \[\TD(\rho,\tau) \coloneqq \frac{1}{2}\|\rho-\tau\|_1 \coloneqq \frac{1}{2}\Tr\left(\sqrt{(\rho-\tau)^\dagger(\rho-\tau)}\right).\] The trace distance between two states $\rho$ and $\tau$ is an upper bound on the probability that any (unbounded) algorithm can distinguish $\rho$ and $\tau$. When clear from context, we will write $\TD(\cX,\cY)$ to refer to the trace distance between a state on register $\cX$ and a state on register $\cY$.

\begin{lemma}[Gentle measurement \cite{DBLP:journals/tit/Winter99}]\label{lemma:gentle-measurement}
Let $\rho^{\cX}$ be a quantum state and let $(\Pi,\bbI-\Pi)$ be a projective measurement on $\cX$ such that $\Tr(\Pi\rho) \geq 1-\delta$. Let \[\rho' = \frac{\Pi\rho\Pi}{\Tr(\Pi\rho)}\] be the state after applying $(\Pi,\bbI-\Pi)$ to $\rho$ and post-selecting on obtaining the first outcome. Then, $\TD(\rho,\rho') \leq 2\sqrt{\delta}$.
\end{lemma}

A non-uniform quantum polynomial-time (QPT) machine $\{\Adv_\secp,\ket{\psi}_\secp\}_{\secp \in \bbN}$ is a family of polynomial-size quantum machines $\Adv_\secp$, where each is initialized with a polynomial-size advice state $\ket{\psi_\secp}$. Each $\Adv_\secp$ is in general described by a CPTP map. Similar to above, when we write $\cY \gets \Adv(\cX)$, we mean that the machine $\Adv$ takes as input a state on register $\cX$ and produces as output a state on register $\cY$, and we leave the actual descripions of these states implicit. Finally, a quantum \emph{interactive} machine is simply a sequence of quantum operations, with designated input, output, and work registers.

Finally we will often use $\approx_c$ as a shorthard to denote \emph{computational} indistinguishability between two families of distributions (over quantum states), and $\approx_s$ as a shorthard to denote \emph{statistical} indistinguishability (or negligible closeness in trace distance) between two families of distributions.

\subsection{Correlation intractability}\label{subsec:CI}

\begin{definition}[Correlation intractable hash function]\label{def:CI}
Let $\{\cX_\secp,\cY_\secp\}_{\secp \in \bbN}$ be families of finite sets. An efficiently computable keyed hash function family $\{H_\secp: \{0,1\}^{k(\secp)} \times \cX_\secp \to \cY_\secp\}_{\secp \in \bbN}$ with keys of length $k(\secp)$ is $\epsilon(\secp)$-\emph{correlation intractable} for a relation ensemble $\{R_\secp \subseteq \cX_\secp \times \cY_\secp\}_{\secp \in \bbN}$ if for any QPT adversary $\{\Adv_\secp\}_{\secp \in \bbN}$, 

\[\Pr\left[(x,H_\secp(\hk,x)) \in R_\secp : \begin{array}{r} \hk \gets \{0,1\}^{k(\secp)} \\ x \gets \Adv_\secp(\hk)\end{array}\right] \leq \epsilon(\secp).\]

We say that $\{H_\secp\}_{\secp \in \bbN}$ is \emph{sub-exponentially} correlation intractable for $\{R_\secp\}_{\secp \in \bbN}$ if it is $2^{-\secp^\delta}$-correlation intractable for some constant $\delta > 0$.
\end{definition}

\begin{definition}[Sparse, efficiently verifiable, approximate product relations \cite{HLR}]\label{def:product-relations}
A relation $R \subseteq \cX \times \cY^t$ is an efficiently verifiable $\alpha$-approximate product relation with sparsity $\rho$ if the following hold.
\begin{itemize}
    \item \textbf{Approximate product.} For every $x$, the set $R_x \coloneqq \{y : (x,y) \in R\}$ consists of $y = (y_1,\dots,y_t) \in \cY^t$ such that  \[|\{i \in [t] : y_i \in S_i\}| \geq \alpha t\] for some sets $S_{1,x},\dots,S_{t,x} \subseteq \cY$ that may depend on $x$.
    
    \item \textbf{Efficiently verifiable.} There is a polynomial-size circuit $C$ such that for every $x$, the sets $S_{1,x},\dots,S_{t,x}$ are such that for any $i,y_i \in S_{i,x}$ if and only if $C(x,y_i,i) = 1$.
    \item \textbf{Sparse.} For every $x$, the sets $S_{1,x},\dots,S_{t,x}$ are such that for all $i$, $|S_{i,x}| \leq \rho|\cY|$.
\end{itemize}
\end{definition}

\begin{importedtheorem}[\cite{HLR}]\label{impthm:CI}
Assuming the existence of an efficiently computable keyed hash function family that is $\epsilon(\secp)$-correlation intractable for any efficient function, there exists an efficiently computable keyed hash function family $\{H_\secp : \{0,1\}^{k(\secp)} \times \cX_\secp \to \cY_\secp^{t(\secp)}\}_{\secp \in \bbN}$ that is $\epsilon(\secp)$-correlation intractable for any efficiently verifiable $\alpha$-approximate product relation ensemble $\{R_\secp \subseteq \cX_\secp \times \cY_\secp^{t(\secp)}\}_{\secp \in \bbN}$ with sparsity $\rho$, as long as $\rho < \alpha$ and $t(\secp) \geq \secp / (\alpha-\rho)^3$. 
\end{importedtheorem}

\begin{importedtheorem}[\cite{CCHLRRW,C:PeiShi19}]\label{impthm:CI2}
Assuming the $\epsilon(\secp)$-hardness of LWE, there exists an efficiently computable keyed hash function family that is $\epsilon(\secp)$-correlation intractable for any efficient function.
\end{importedtheorem}

\begin{definition}[Programmability]\label{def:programmability}
A hash function family $\{H_\secp: \{0,1\}^{k(\secp)} \times \cX_\secp \to \cY_\secp\}_{\secp \in \bbN}$ is \emph{programmable} if for any $\secp, x \in \cX_\secp$, and $y \in \cY_\secp$, \[\Pr_{\hk \gets \{0,1\}^{k(\secp)}}[H_\secp(\hk,x) = y] = \frac{1}{2^{m(\secp)}},\] and there exists a PPT sampling algorithm $\Samp(1^\secp,x,y)$ that samples from the conditional distribution \[\hk : H_\secp(\hk,x) = y.\]
\end{definition}

\begin{remark}
\cite{CCHLRRW} show a simple transformation that generically adds the above notion of programmability to natural correlation intractable hash functions.
\end{remark}

\ifsubmission
In \cref{sec:crypto-prelim}, we present additional preliminaries covering commitments, zero-knowledge, and quantum leftover hashing.
\else
\ifsubmission \section{Additional Preliminaries}\label{sec:crypto-prelim}\else\fi

\subsection{Commitments}\label{subsec:commitments}

A non-interactive commitment in the common random string model is parameterized by polynomials $h(\secp),n(\secp)$, and consists of the following algorithm. 
\begin{itemize}
    \item $\Com(\ck,m) \to c$: Take as input a commitment key $\ck \in \{0,1\}^{h(\secp)}$ and a message $m \in \{0,1\}^{n(\secp)}$, and output a commitment string $c$.
\end{itemize}

\begin{definition}[Hiding]\label{def:hiding}
A non-interactive commitment scheme $\Com$ is \emph{hiding} if for any QPT adversary $\{\Adv_\secp\}_{\secp \in \bbN}$ and messages $\{m_{0,\secp},m_{1,\secp}\}_{\secp \in \bbN}$,

\[\bigg| \Pr\left[1 \gets \Adv_\secp(\ck,c) : \begin{array}{r}\ck \gets \{0,1\}^{h(\secp)} \\ c \gets \Com(\ck,m_{0,\secp})\end{array}\right] - \Pr\left[1 \gets \Adv_\secp(\ck,c) : \begin{array}{r}\ck \gets \{0,1\}^{h(\secp)} \\ c \gets \Com(\ck,m_{1,\secp})\end{array}\right] \bigg| = \negl(\secp).\]
\end{definition}

\begin{definition}[Extractability]\label{def:extractability}
A non-interactive commitment scheme $\Com$ is \emph{extractable} if there exist PPT algorithms $(\ExtGen,\Ext)$ such that for any QPT adversary $\{\Adv_\secp\}_{\secp \in \bbN}$,

\[\left|\Pr[1 \gets \Adv_\secp(\ck) : \ck \gets \{0,1\}^{h(\secp)}] - \Pr[1 \gets \Adv_\secp(\ck) : (\ck,\ek) \gets \ExtGen(1^\secp)]\right| = \negl(\secp),\]

and 

\[\Pr\left[\exists m' \neq m,r \text{ s.t. } \Com(\ck,m';r) = c : \begin{array}{r} (\ck,\ek) \gets \ExtGen(1^\secp) \\ c \gets \Adv_\secp(\ck) \\ m \gets \Ext(\ek,c)
\end{array}\right] = \negl(\secp),\] where $\Ext$ either outputs an $n(\secp)$-bit message or $\bot$.

\end{definition}

\begin{remark}\label{def:com-lwe}
Commitments satisfying these properties are known from LWE, for example via dual-Regev encryption.
\end{remark}



\subsection{Non-interactive zero-knowledge}\label{subsec:ZK}

Let $\cL$ be an NP language and let $\cR_\cL$ be the associated binary relation, where a statement $x \in \cL$ if and only if there exists a witness $w$ such that $(x,w) \in \cR_\cL$. A non-interactive argument system for $\cL$ in the common random string model consists of the following algorithms.

\begin{itemize}
    \item $\Prove(\crs,x,w) \to \pi$: The prover algorithm takes as input a common random string $\crs \in \{0,1\}^{n(\secp)}$, a statememt $x$, and a witness $w$, and outputs a proof $\pi$.
    \item $\Ver(\crs,x,\pi) \to \{\top,\bot\}$: The verify algorithm takes as input a common random string $\crs \in \{0,1\}^{n(\secp)}$, a statement $x$, and a proof $\pi$, and outputs either $\top$ or $\bot$.
\end{itemize}

\begin{definition}[Completeness]\label{def:ZK-completeness}
The non-interactive argument system $(\Prove,\Ver)$ satisfies \emph{completeness} if for any $(x,w) \in \cR_\cL$, \begin{align*}&\Pr\left[\Ver(\crs,x,\pi) = 1 : \begin{array}{r}\crs \gets \{0,1\}^{n(\secp)} \\ \pi \gets \Prove(\crs,x,w)\end{array}\right] = 1-\negl(\secp).\end{align*}
\end{definition}

\begin{definition}[Soundness]\label{def:ZK-soundness}
The non-interactive argument system $(\Prove,\Ver)$ satisfies \emph{soundness} if for any QPT adversary $\{\Adv_\secp\}_{\secp \in \bbN}$,

\[\Pr\left[x \notin \cL \wedge \Ver(\crs,x,\pi) = \top : \begin{array}{r}\crs \gets \{0,1\}^{n(\secp)} \\ (x,\pi) \gets \Adv_\secp(\crs)\end{array}\right] = \negl(\secp).\]
\end{definition}




\begin{definition}[Zero-knowledge]\label{def:ZK-ZK}
The argument system $(\Prove,\Ver)$ satisfies \emph{zero-knowledge} if there exists a PPT simulator $\Sim$ such that for any $(x,w) \in \cR_\cL$,

\begin{align*}&\bigg|\Pr\left[1 \gets \Adv_\secp(\crs,x,\pi) : \begin{array}{r} \crs \gets \{0,1\}^{n(\secp)}, \\ \pi \gets \Prove(\crs,x,w)\end{array}\right] - \Pr[1 \gets \Adv_\secp(\crs,x,\pi) : (\crs,\pi) \gets \Sim(1^\secp,x)]\bigg| = \negl(\secp).\end{align*}
\end{definition}

\begin{importedtheorem}[\cite{CCHLRRW,C:PeiShi19}]\label{impthm:nizk}
There exists a NIZK argument for NP assuming LWE.
\end{importedtheorem}

\subsection{Quantum entropy and leftover hashing}
\label{subsec: quantum min entropy}

\paragraph{Quantum conditional min-entropy.} 

Let $\rho^{\regX,\regY}$ denote a bipartite quantum state on registers $\regX,\regY$. Following~\cite{Renner08,KonRenSch09}, the conditional min-entropy of $\rho^{\regX,\regY}$ given $\regY$ is defined to be 
\[\mathbf{H}_\infty\left(\rho^{\regX,\regY} \mid \regY\right) \coloneqq \sup_\tau \max \left\{h \in \mathbb{R} : 2^{-h} \cdot \bbI^\regX \otimes \tau^\regY - \rho^{\regX,\regY} \geq 0\right\}.\]

In this work, we will exclusively consider the case where $\rho^{\regX,\regY}$ can be written as 
\[ \sum_{x \in X} p_x\dyad{x}{x}^\regX \otimes \tau^\regY \] for some finite set $X$ and probability distribution $\{p_x\}_{x \in X}$. We refer to such $\rho^{\regX,\regY}$ as a classical-quantum state. In this case, quantum conditional min-entropy exactly corresponds to the maximum probability of guessing $x$ given the state on register $\regY$.

\begin{importedtheorem}[\cite{KonRenSch09}]\label{impthm:conditional-min-entropy}
Let $\rho^{\regX,\regY}$ be a classical-quantum state, and let $p_{\mathsf{guess}}(\rho^{\regX,\regY} | \regY)$ be the maximum probability that any quantum operation can output the $x$ on register $\regX$, given the state on register $\regY$. Then \[p_{\mathsf{guess}}(\rho^{\regX,\regY} | \regY) = 2^{-\mathbf{H}_\infty(\rho^{\regX,\regY} | \regY)}.\]

\end{importedtheorem}

\paragraph{Leftover hash lemma with quantum side information.} We now state a generalization of the leftover hash lemma to the setting of quantum side information. 
\begin{importedtheorem}[\cite{TCC:RenKon05}]\label{impthm:privacy-amplification}
Let $\mathcal{H}$ be a family of universal hash functions from $X$ to $\{0,1\}^\secp$, i.e. for any $x \neq x'$, $\Pr_{h \leftarrow \mathcal{H}}[h(x) = h(x')] = 2^{-\secp}$. Let $\rho^{\regX,\regY}$ be any classical-quantum state. Let $\regK$ be a register that holds $h \gets \cH$, let $\regR$ be a register that holds $h(x)$ where $x$ is from register $\regX$, and define $\rho^{\regX,\regY,\regK,\regR}$ to be the entire system. Then, it holds that
\[\TD\left(\rho^{\regY,\regK,\regR},\rho^{\regY,\regK} \otimes \frac{1}{2^\secp}\sum_{r \in \{0,1\}^\secp}\dyad{r}{r}^\regR\right) \leq \frac{1}{2^{1+\frac{1}{2}(\mathbf{H}_\infty(\rho^{\regX,\regY}|\regY)-\secp)}}.\]

\end{importedtheorem}

\paragraph{Small superposition of terms.} We will also make use of the following lemma from \cite{C:BouFeh10}.

\begin{importedtheorem}(\cite{C:BouFeh10})\label{impthm:small-superposition}
Let $\regX,\regY$ be registers of arbitrary size, and let $\{\ket{i}\}_{i \in I}$ and $\{\ket{w}\}_{w \in W}$ be orthonormal bases of $\cX$. Let $\ket{\psi}^{\regX,\regY}$ and $\rho^{\regX,\regY}$ be of the form \[\ket{\psi} = \sum_{i \in J}\alpha_i\ket{i}^\regX\ket{\psi_i}^\regY  \text{   and     } \rho = \sum_{i \in J}|\alpha_i|^2\dyad{i}{i}^\regX \otimes \dyad{\psi_i}{\psi_i}^\regY\] for some subset $J \subseteq I$. Furthermore, let $\widehat{\rho}^{\regX,\regY}$ and $\widehat{\rho}_{\mathsf{mix}}^{\regX,\regY}$ be the classical-quantum states obtained by measuring register $\cX$ of $\ket{\psi}$ and $\rho$, respectively, in basis $\{\ket{w}\}_{w \in W}$ to observe outcome $w$. Then, \[\mathbf{H}_\infty(\widehat{\rho}^{\regX,\regY} | \regY) \geq \mathbf{H}_\infty(\widehat{\rho}_{\mathsf{mix}}^{\regX,\regY} | \regY) - \log|J|.\]
\end{importedtheorem}

\fi

\subsection{Secure computation}\label{subsec:secure-computation}

An ideal functionality $\cF$ is an interactive (classical or quantum) machine specifying some distributed computation. In this work, we will specifically focus on \emph{two-party} functionalities between party $A$ and party $B$. In some cases, party $B$ will have a random input, or no input. The ideal functionalities we will consider in this work are specified in \cref{fig:ideal-functionalities}.

\protocol{Ideal functionalities}{Ideal functionalities considered in this work.}{fig:ideal-functionalities}{

Setup: Parties $A$ and $B$, security parameter $\secp$.\\

$\underline{\cF_\OT}$
\begin{itemize}
    \item $\cF_\OT$ receives input $m_0,m_1 \in \{0,1\}^\secp$ from $A$ and $b \in \{0,1\}$ from $B$.
    \item $\cF_\OT$ delivers $m_b$ to $B$.
\end{itemize}

$\underline{\cF_{\ROT}}$
\begin{itemize}
    \item $\cF_\ROT$ receives input $m_0,m_1 \in \{0,1\}^\secp$ from $A$.
    \item $\cF_\ROT$ samples a bit $b \gets \{0,1\}$ and delivers $(b,m_b)$ to $B$.
\end{itemize}

$\underline{\cF_\CL[C]}$
\begin{itemize}
    \item $C$ is a classical circuit with two inputs, one of length $n_1 = n_1(\secp)$ and one of length $n_2 = n_2(\secp)$.
    \item $\cF_\CL[C]$ receives input $x \in \{0,1\}^{n_1}$ from $A$.
    \item $\cF_\CL[C]$ samples a string $r \gets \{0,1\}^{n_2}$ and delivers $C(x,r)$ to $B$.
\end{itemize}


$\underline{\cF_\QU[Q]}$ 
\begin{itemize}
    \item $Q$ is a quantum operation that takes as input a state on register $\regX$ of $n = n(\secp)$ qubits and outputs a state on register $\regY$. 
    \item $\cF_\QU[Q]$ receives as input a state on register $\regX$ from $A$.
    \item $\cF_\QU[Q]$ computes $Q(\regX) = \regY$ and delivers $\regY$ to $B$.
\end{itemize}
}

\paragraph{Security with abort.} In what follows, we will by default consider the notion of security with abort, where the ideal functionality $\cF$ is always modified to (1) know the identity of the corrupt party (if one exists) and (2) be slightly reactive: after the parties have provided input, the functionality computes outputs and sends output to the corrupt party only (if it expects output). Then the functionality awaits either a “deliver” or “abort” command from the corrupted party. Upon receiving “deliver”, the functionality delivers the honest party output. Upon receiving “abort”, the functionality instead delivers an abort message $\bot$ to the honest party. In the case where the corrupt party does not expect output, the functionality $\cF$ still awaits a ``deliver'' or ``abort'' from the corrupt party before delivering output (or $\bot$) to the honest party.



%

\paragraph{The real-ideal paradigm.} A two-party protocol $\Pi_\cF$ for computing the functionality $\cF$ consists of two families of quantum interactive machines $\{\sA_\secp\}_{\secp \in \bbN},\{\sB_\secp\}_{\secp \in \bbN}$. An adversary intending to attack the protocol by corrupting one of the parties can be described by a family of quantum interactive machines  $\{\Adv_\secp\}_{\secp \in \bbN}$ and a family of initial quantum states $\{\ket{\psi_\secp}^{\regX,\regA,\regD}\}_{\secp \in \bbN}$ on registers $(\regX,\regA,\regD)$, where $\regX$ is the honest party's input register, $\regA$ is the adversary's input register, and $\regD$ is given directly to the distinguisher. That is, the honest party takes as input the state on register $\regX$, $\Adv_\secp$ takes as input the state on register $\regA$,  and they interact in the protocol $\Pi_\cF$. Then, the honest party outputs a state on register $\regX'$, $\Adv_\secp$ outputs a state on register $\regA'$, and we define the random variable $\Pi_\cF[\Adv_\secp,\ket{\psi_\secp}]$ to consist of the resulting state on registers $(\regX',\regA',\regD)$, which will be given to a distinguisher. In the case where the honest party has no input, we don't include a register $\regX$, and just consider families $\{\ket{\psi_\secp}^{\regA,\regD}\}_{\secp \in \bbN}$ on registers $\regA$ and $\regD$. In the case where the honest party has a classical input, we assume that $\regX$ is in a standard basis state. In other words, we consider families $\{(x_\secp,\ket{\psi_\secp}^{\regA,\regD})\}_{\secp \in \bbN}$, where each $x_\secp$ is a classical string.

An \emph{ideal-world} protocol $\widetilde{\Pi}_\cF$ for functionality $\cF$ consists of ``dummy'' parties $\widetilde{A}$ and $\widetilde{B}$ that have access to an additional ``trusted'' party that implements $\cF$. That is, $\widetilde{A}$ and $\widetilde{B}$ only interact directly with $\cF$, providing inputs and receiving outputs, and do not interact with each other. We consider the execution of ideal-world protocols in the presence of a simulator, described by a family of quantum interactive machines $\{\Sim_\secp\}_{\secp \in \bbN}$ that controls either party $\widetilde{A}$ or $\widetilde{B}$. The execution of the protocol in the presence of the simulator also begins with a family of states $\{\ket{\psi_\secp}^{\regX,\regA,\regD}\}_{\secp \in \bbN}$ on registers $(\regX,\regA,\regD)$ as described above, and we define the analogous random variable $\widetilde{\Pi}_\cF[\Sim_\secp,\ket{\psi_\secp}]$.

\paragraph{Secure realization.} We define what it means for a protocol to securely realize an ideal functionality.

\begin{definition}[Secure realization]\label{def:secure-realization}
A protocol $\Pi_\cF$ \emph{securely realizes} the functionality $\cF$ if for any QPT adversary $\{\Adv_\secp\}_{\secp \in \bbN}$ corrupting party $M \in \{A,B\}$, there exists a QPT simulator $\{\Sim_\secp\}_{\secp \in \bbN}$ corrupting party $M$ such that for any QPT distinguisher $\{\sD_\secp\}_{\secp \in \bbN}$ and polynomial-size family of states $\{\ket{\psi_\secp}^{\regX,\regA,\regD}\}_{\secp \in \bbN}$, 
\[\bigg| \Pr[1 \gets \sD_\secp(\Pi_\cF[\Adv_\secp,\ket{\psi_\secp}])] - \Pr[1 \gets \sD_\secp(\widetilde{\Pi}_\cF[\Sim_\secp,\ket{\psi_\secp}])]\bigg| = \negl(\secp).\]
\end{definition}

\subsection{The XOR extractor}

\begin{importedtheorem}[\cite{ABKK}]\label{thm:XOR-extractor}
Let $\regX$ be an $n$-qubit register, and consider any quantum state $\ket{\gamma}^{\regA,\regX}$ that can be written as \[\ket{\gamma}^{\regA,\regX} = \sum_{u: \hw(u) < n/2} \ket{\psi_u}^{\regA} \ket{u}^{\regX},\] where $\hw(\cdot)$ denotes the Hamming weight. Let $\rho^{\regA,\regP}$ be the mixed state that results from measuring $\regX$ in the Hadamard basis to produce a string $x \in \{0,1\}^n$, and writing $\bigoplus_{i \in [n]}x_i$ into a single qubit register $\regP$. Then it holds that \[\rho^{\regA,\regP} = \Tr^{\regX}(\dyad{\gamma}{\gamma}) \otimes \left(\frac{1}{2}\dyad{0}{0} + \frac{1}{2}\dyad{1}{1}\right)^{\regP}.\] 
\end{importedtheorem}

\section{One-Shot String Oblivious Transfer}

\subsection{Impossibility in the CRS model}\label{subsec:impossibility}

First, we show that a classical shared random string is not sufficient to achieve one-shot (random receiver bit) string OT, even when parties can compute and communicate quantumly. The intuition is that the sender has the same view of the receiver right before it sends its message, and can thus run the receiver's honest computation on its message in order to learn the receiver's choice bit. This is easy to formalize in the case that the sender and receiver are running classical computations, but takes a little more care in the quantum setting.

We will rely on recent observations about the quantum equivalence between mapping and distinguishing. In particular, we will show that, EITHER the receiver cannot tell that the sender measured their choice bit, which violates security against a malicious sender, OR their exists an efficient adversarial receiver that can map between the sender's strings $m_0,m_1$ and can recover them both, violating security against a malicious receiver.

\begin{theorem}
There does not exist a one-message protocol that securely realizes the functionality $\cF_\ROT$ in the common reference string model, even if parties can compute and communicate quantumly.
\end{theorem}

\begin{proof}
We will use the following imported theorem, which is a special case of \cite[Claim 3.5]{cryptoeprint:2022/786}.

\begin{importedtheorem}[\cite{cryptoeprint:2022/786}]\label{impthm:DS}
Let $\sD$ be a projector, $\Pi_0,\Pi_1$ be orthogonal projectors, and $\ket{\psi} \in \mathsf{Im}\left(\Pi_0+\Pi_1\right)$. Then,

\[\|\Pi_1\sD\Pi_0\ket{\psi}\|^2 + \|\Pi_0\sD\Pi_1\ket{\psi}\|^2 \geq \frac{1}{2}\left(\|\sD\ket{\psi}\|^2 - \left(\|\sD\Pi_0\ket{\psi}\|^2 + \|\sD\Pi_1\ket{\psi}\|^2 \right)\right)^2.\]
\end{importedtheorem}

Now consider any one-message protocol for $\cF_\ROT$, where the common reference string is sampled as $\crs \gets \cR$ from some distribution $\cR$ over classical strings. Fix any pair of sender inputs $(m_0,m_1)$, and let $\ket{\psi}$ be the message sent by the honest sender on input $m_0,m_1$ and $\crs$ (note that $\ket{\psi}$ is a sample from a distribution). Without loss of generality, we can consider the honest receiver strategy to apply a unitary $U_R$ on $\ket{\psi}\ket{0}$ (where the auxiliary register is initialized with sufficiently many qubits) and then measure an output register in the standard basis to obtain $(b,m_b)$. We define

\[\Pi_0 \coloneqq U_R^\dagger\dyad{0,m_0}{0,m_0}U_R, \ \ \, \Pi_1 \coloneqq U_R^\dagger\dyad{1,m_1}{1,m_1}U_R.\] 

Suppose that the protocol satisfies perfect correctness,\footnote{If it satisfies $1-\negl(\secp)$ correctness, we can apply Gentle Measurement (\cref{lemma:gentle-measurement}) to $\ket{\psi}$, which will only affect our conclusions by a $\negl(\secp)$ amount.} which implies that $\ket{\psi} \in \mathsf{Im}(\Pi_0 + \Pi_1)$. Now, suppose that an adversarial sender applies the measurement $\{\Pi_0,\Pi_1\}$ to $\ket{\psi}$ and then sends the resulting mixture over $\Pi_0\ket{\psi}$ and $\Pi_1\ket{\psi}$ to the receiver. There are two cases.

In the first case, there does not exist a QPT procedure that distinguishes $\ket{\psi}$ from the mixture over $\Pi_0\ket{\psi}$ and $\Pi_1\ket{\psi}$. In this case, the sender learns the bit $b$ obtained by the receiver from the outcome of measurement $\{\Pi_0,\Pi_1\}$, which violates security.

In the second case, there exists a QPT distinguisher $\sD$ (written as a projective measurement) that distinguishes with $\nonnegl(\secp)$ probability. Then by \cref{impthm:DS}, we have that \[\E_{\crs,\ket{\psi}}\left[\|\Pi_1\sD\Pi_0\ket{\psi}\|^2 + \|\Pi_0\sD\Pi_1\ket{\psi}\|^2\right] = \nonnegl(\secp),\] where the expectation is over the sampling of $\crs \gets \cR$ and the sampling of the sender's message $\ket{\psi}$. However, this means that the following receiver strategy will return $\{m_0,m_1\}$ with $\nonnegl(\secp)$ probability, which violates security. Begin with $\ket{\psi}\ket{0}$. Apply $U_R$ and measure the output register in the computational basis to obtain $(b,m_b)$. Then, apply $U_R \sD U_R^\dagger$ and measure the output register in the computational basis to obtain $(1-b,m_{1-b})$. This completes the proof.

\end{proof}

\subsection{Construction in the shared EPR pairs model}

In this section, we give our construction of one-shot (random receiver bit) string oblivious transfer in the shared EPR pairs model.

\paragraph{\bf Ingredients}
\begin{itemize}
    \item Non-interactive extractable commitment $(\Com,\ExtGen,\Ext)$ in the common random string model (\cref{subsec:commitments}). This is known from LWE (\cref{def:com-lwe}).
    \item A programmable hash function family $\{H_\secp\}_{\secp \in \bbN}$ that is sub-exponentially correlation intractable for efficiently verifiable approximate product relations with constant sparsity (\cref{subsec:CI}). This is known from the sub-exponential hardness of LWE (\cref{impthm:CI,impthm:CI2}).
    \item Non-interactive zero-knowledge argument $(\NIZK.\Prove,\allowbreak\NIZK.\Ver,\allowbreak\NIZK.\Sim)$ in the common random string model (\cref{subsec:ZK}). This is known from LWE (\cref{impthm:nizk}).
    \item Pseudorandom generator $\PRG$.
\end{itemize}



\paragraph{\bf Parameters}
\begin{itemize}
    \item Security parameter $\secp$.
    \item Correlation intractable hash security parameter $\secp_\CI \coloneqq \secp^{1/\delta}$, where $\delta > 0$ is the constant such that $\{H_{\secp_\CI}\}_{\secp_\CI \in \bbN}$ is $2^{-\secp_\CI^\delta}$-correlation intractable.
    \item Size of commitment key $h = h(\secp)$.
    \item Size of NIZK crs $n = n(\secp)$.
    \item Size of hash key $k = k(\secp_\CI)$.
    \item Approximation parameter $\alpha = 1/120$.
    \item Number of repetitions in each group $c = 480$.
    \item Sparsity $\rho = \frac{\binom{(1-\alpha)c}{(1/2)c}}{2^c} < \alpha$.
    \item Product parameter $t = t(\secp_\CI) = 180^3\secp_\CI \geq \secp_\CI / (\alpha-\rho)^3$.
    \item Total number of repetitions $\ell = \ell(\secp) = c \cdot t = \poly(\secp)$.
    \item $\PRG$ range $\{0,1\}^{2\secp\ell}$.
    \item CI hash range $\cY^t$, where $\cY$ is the set of subsets of $[c]$ of size $c/2$. We will also parse $T \in \cY^t$ as a subset of $[\ell]$ of size $\ell / 2$.
\end{itemize}

We remark that we have not tried to fully optimize the constants in the parameters above.

\paragraph{\bf Setup}
\begin{itemize}
    \item $\ell$ collections of EPR pairs indexed by $i \in [\ell]$. Each collection consists of one ``control'' pair $\{\regS_i^\ctl,\regR_i^\ctl\}$ and $2\secp$ ``message'' pairs on registers $\{\regS_{i,j}^\msg,\regR_{i,j}^\msg\}_{j \in [2\secp]}$. For each $i \in [\ell]$, we define $\regS_i \coloneqq (\regS_i^\ctl,\regS_{i,1}^\msg,\dots,\regS_{i,2\secp}^\msg)$ and $\regR_i \coloneqq (\regR_i^\ctl,\regR_{i,1}^\msg,\dots,\regR_{i,2\secp}^\msg)$.
    \item Commitment key $\ck \gets \{0,1\}^h$.
    \item NIZK common random string $\crs \gets \{0,1\}^n$.
    \item Correlation intractable hash key $\hk \gets \{0,1\}^k$.
    
\end{itemize}

Note that a shared uniformly random string can be obtained by measuring shared EPR pairs in the same basis, and thus this entire Setup can be obtained with just shared EPR pairs. 

Finally, given a commitment key $\ck$ for $\Com$ and a set $\overline{T} \subset [\ell]$, we define the NP language $\cL_{\ck,\overline{T}}$ of instance-witness pairs as follows. 

\[\left(\left(\{\widehat{\cm}_{i,0},\widehat{\cm}_{i,1}\}_{i \in \overline{T}},\{\cm_i\}_{i \in [\ell]}\right),\left(\{t_i\}_{i \in \overline{T}},\Delta,s\right)\right) \in \cL_{\ck,\overline{T}}\] if and only if\footnote{Technically, the random coins used to compute the commitments must also be included in the witness.}

\begin{align*}
    &\forall i \in \overline{T}, \widehat{\cm}_{i,0} \in \Com(\ck,t_i) \, \wedge \, \widehat{\cm}_{i,1} \in \Com(\ck,t_i \oplus \Delta), \text{ and}\\
    &\forall i \in [\ell], \cm_i \in \Com(\ck,(\cdot,x_i,\cdot)), \text{where } (x_1,\dots,x_\ell) \coloneqq \PRG(s).
\end{align*}



Now, our protocol is described in \cref{fig:NIOT}.

\protocol{One-shot protocol for $\cF_\ROT$}{A protocol for one-shot random string OT in the shared EPR pair model.}{fig:NIOT}{
{\underline{Sender message.} ~~ Input strings $m_0,m_1 \in \{0,1\}^\secp$.} 
    \begin{enumerate}
        \item Sample a PRG seed $s \gets \{0,1\}^\secp$ and set $(x_1,\dots,x_\ell) \coloneqq \PRG(s)$, where each $x_i \in \{0,1\}^{2\secp}$.
        \item For each $i \in [\ell]$:
        \begin{itemize}
            \item For each $j \in [2\secp]$ such that $x_{i,j} = 1$, apply a CNOT gate from register $\regS_i^\ctl$ to register $\regS_{i,j}^\msg$.
            \item Measure $\{\cS_{i,j}^\msg\}_{j \in [2\secp]}$ in the standard basis to obtain $v_i \in \{0,1\}^{2\secp}$ and measure $\cS_i^\ctl$ in the Hadamard basis to obtain $h_i \in \{0,1\}$. 
            \item Compute $\cm_i \coloneqq \Com(\ck,(v_i,x_i,h_i);r_i)$, where $r_i \gets \{0,1\}^\secp$ are the random coins used for commitment.
        \end{itemize}
        \item Compute $T = H_\secp(\hk,(\cm_1,\dots,\cm_\ell)) \subset [\ell]$ and let $\overline{T} \coloneqq [\ell]\setminus T$.
        \item Sample $\Delta \gets \{0,1\}^\secp$ and for each $i \in \overline{T}$: 
        \begin{itemize}
            \item Sample $t_i \gets \{0,1\}^\secp$ and compute $\widehat{\cm}_{i,0} \coloneqq \Com(\ck,t_i;r_{i,0})$ and $\widehat{\cm}_{i,1} \coloneqq \Com(\ck,t_i \oplus \Delta;r_{i,1})$ where $r_{i,0},r_{i,1} \gets \{0,1\}^\secp$ are the random coins used for commitment.
            \item Define $z_{i,0} = (t_i,r_{i,0}) \oplus v_i, ~~ z_{i,1} = (t_i \oplus \Delta,r_{i,1})  \oplus v_i \oplus x_i.$
        \end{itemize}
        \item Define \[\widetilde{m}_0 \coloneqq m_0 \oplus \bigoplus_{i \in \overline{T}} t_i, ~~ \widetilde{m}_1 \coloneqq m_1 \oplus \Delta \oplus \bigoplus_{i \in \overline{T}} t_i.\]
        \item Compute $\pi \gets \NIZK.\Prove\left(\crs,\left(\{\widehat{\cm}_{i,0},\widehat{\cm}_{i,1}\}_{i \in \overline{T}},\{\cm_i\}_{i \in [\ell]}\right),\left(\{t_i\}_{i \in \overline{T}}, \Delta, s\right)\right)$ for the language $\cL_{\ck,\overline{T}}$.

        \item Send $\left( \{\cm_i\}_{i \in [\ell]},\{v_i,x_i,h_i,r_i\}_{i \in T}, \{\widehat{\cm}_{i,0},\widehat{\cm}_{i,1},z_{i,0},z_{i,1}\}_{i \in \overline{T}}, \pi, \widetilde{m}_0,\widetilde{m}_1\right)$ to the receiver. 
    \end{enumerate}
    
\medskip
{\underline{Receiver computation.}} ~~ In what follows, abort and output $\bot$ if any check fails.
    \begin{enumerate}
        \item Compute $T = H_\secp(\hk,(\cm_1,\dots,\cm_\ell))$ and check that for all $i \in T$, $\cm_i = \Com(\ck,(v_i,x_i,h_i);r_i)$.
        \item For each $i \in T$, define $\ket{\psi_{v_i,x_i,h_i}} \coloneqq \frac{1}{\sqrt{2}}\left({\ket{0,v_i} + (-1)^{h_i}\ket{1,v_i \oplus x_i}}\right)$, and measure register $\regR_i$ in the basis $\left\{\dyad{\psi_{v_i,x_i,h_i}}{\psi_{v_i,x_i,h_i}}, \bbI - \dyad{\psi_{v_i,x_i,h_i}}{\psi_{v_i,x_i,h_i}}\right\}$. Check that for all $i \in T$, the first outcome is observed.
        \item Check that $\NIZK.\Ver\left(\crs,\left(\{\widehat{\cm}_{i,0},\widehat{\cm}_{i,1}\}_{i \in \overline{T}},\{\cm_i\}_{i \in [\ell]}\right),\pi\right) = \top$.
        \item For each $i \in \overline{T}$, measure register $\regR_i$ in the standard basis to obtain $b_i \in \{0,1\}$ and  $v_i' \in \{0,1\}^{2\secp}$, compute $(t_i',r_i') = z_{i,b_i} \oplus v_i'$, and check that for each $i \in \overline{T}$, $\widehat{\cm}_{i,b_i} = \Com(\ck,t_i';r_i')$.
        \item Output \[b \coloneqq \bigoplus_{i \in \overline{T}} b_i, ~~ m_b \coloneqq \widetilde{m}_b \oplus \bigoplus_{i \in \overline{T}}t'_i.\]
    \end{enumerate}
}

\subsection{Security}

\begin{theorem}\label{thm:one-shot-security}
The protocol in \cref{fig:NIOT} securely realizes (\cref{def:secure-realization}) the functionality $\cF_\ROT$. Thus, assuming the sub-exponential hardness of LWE, there exists a one-message protocol for $\cF_\ROT$ in the shared EPR pair model.
\end{theorem}

The proof of this theorem follows from \cref{lemma:ROT-sender-security} and \cref{lemma:ROT-receiver-security} below.

\begin{lemma}\label{lemma:ROT-sender-security}
The protocol in \cref{fig:NIOT} is secure against a malicious sender.
\end{lemma}

\begin{proof}
Let $\{\Adv_\secp\}_{\secp \in \bbN}$ be a QPT adversary corrupting the sender, which takes as input register $\regA$ of $\{\ket{\psi_\secp}^{\regA,\regD}\}_{\secp \in \bbN}$. Note that we don't consider a register $\regX$ holding the honest party's input, since an honest receiver has no input. We will define a sequence of hybrids, beginning with the real distribution $\Pi_{\cF_\ROT}[\Adv_\secp,\ket{\psi_\secp}]$ and ending with the distribution $\widetilde{\Pi}_{\cF_\ROT}[\Sim_\secp,\ket{\psi_\secp}]$ defined by a simulator $\{\Sim_\secp\}_{\secp \in \bbN}$. Each hybrid is a distribution described by applying an operation to the input register $\cA$, and a QPT distinguisher will obtain the output of this distribution along with the register $\cD$. We drop the dependence of the hybrids on $\secp$ for convenience.\\

\noindent\underline{$\cH_0(\regA)$}
\begin{itemize}
    \item Prepare $\ell$ collections of EPR pairs on registers $\{\regS_i,\regR_i\}_{i \in [\ell]}$, and sample $\ck \gets \{0,1\}^h$, $\crs \gets \{0,1\}^n$, and $\hk \gets \{0,1\}^k$.
    \item Run $\Adv_\secp$ on input $\regA,\{\regS_i\}_{i \in [\ell]},\ck,\crs,\hk$ until it outputs a message
    \[\left( \{\cm_i\}_{i \in [\ell]},\{v_i,x_i,h_i,r_i\}_{i \in T}, \{\widehat{\cm}_{i,0},\widehat{\cm}_{i,1},z_{i,0},z_{i,1}\}_{i \in \overline{T}}, \pi, \widetilde{m}_0,\widetilde{m}_1\right)\] and a state on register $\regA'$.
    \item Run the Receiver's honest computation on the sender's message to obtain an output $(b,m_b)$ or $\bot$. Output either $(\regA',(b,m_b))$ or $(\regA',\bot)$.
\end{itemize}

\noindent\underline{$\cH_1(\regA)$}
\begin{itemize}
    \item Prepare $\ell$ collections of EPR pairs on registers $\{\regS_i,\regR_i\}_{i \in [\ell]}$, and sample \textcolor{red}{$(\ck,\ek) \gets \ExtGen(1^\secp)$}, $\crs \gets \{0,1\}^n$, and $\hk \gets \{0,1\}^k$.
    \item Run $\Adv_\secp$ on input $\regA,\{\regS_i\}_{i \in [\ell]},\ck,\crs,\hk$ until it outputs a message
    \[\left( \{\cm_i\}_{i \in [\ell]},\{v_i,x_i,h_i,r_i\}_{i \in T}, \{\widehat{\cm}_{i,0},\widehat{\cm}_{i,1},z_{i,0},z_{i,1}\}_{i \in \overline{T}}, \pi, \widetilde{m}_0,\widetilde{m}_1\right)\] and a state on register $\regA'$.
    \item Run the Receiver's honest computation on the sender's message to obtain an output $(b,m_b)$ or $\bot$. Output either $(\regA',(b,m_b))$ or $(\regA',\bot)$.
\end{itemize}

\noindent\underline{$\cH_2(\regA)$}
\begin{itemize}
    \item Prepare $\ell$ collections of EPR pairs on registers $\{\regS_i,\regR_i\}_{i \in [\ell]}$, and sample $(\ck,\ek) \gets \ExtGen(1^\secp)$, $\crs \gets \{0,1\}^n$, and $\hk \gets \{0,1\}^k$.
    \item Run $\Adv_\secp$ on input $\regA,\{\regS_i\}_{i \in [\ell]},\ck,\crs,\hk$ until it outputs a message
    \[\left( \{\cm_i\}_{i \in [\ell]},\{v_i,x_i,h_i,r_i\}_{i \in T}, \{\widehat{\cm}_{i,0},\widehat{\cm}_{i,1},z_{i,0},z_{i,1}\}_{i \in \overline{T}}, \pi, \widetilde{m}_0,\widetilde{m}_1\right)\] and a state on register $\regA'$.
    \item Run Steps 1-3 of the Receiver's honest computation on the sender's message.
    \item We will now \emph{coherently} apply the check described in Step 4 to the registers $\{\regR_i\}_{i \in \overline{T}}$. First we introduce some notation. For commitment key $\ck$, commitment $\widehat{\cm}$, and two strings $z_0,z_1 \in \{0,1\}^{2\secp}$, let $\Pi[\ck,\widehat{\cm},z_0,z_1]$ be a projection onto strings $(b,v') \in \{0,1\}^{1+2\secp}$ such that $\widehat{\cm} = \Com(\ck,t;r)$, where $(t,r) \coloneqq z_b \oplus v'$. 
    
    \textcolor{red}{Attempt to project registers $\{\regR_i\}_{i \in \overline{T}}$ onto \[\bigotimes_{i \in \overline{T}}\Pi[\ck,\widehat{\cm}_{i,0},z_{i,0},z_{i,1}]^{\regR_i},\] and aborts if the projection fails.}
    \item If there was an abort, output $(\regA',\bot)$. Otherwise, for each $i \in \overline{T}$, measure register $\cR_i$ in the standard basis to obtain $b_i \in \{0,1\}$ and $v_i' \in \{0,1\}^{2\secp}$, and compute $(t_i',r_i') = z_{i,b_i} \oplus v_i'$. Then, define \[b \coloneqq \bigoplus_{i \in \overline{T}} b_i, ~~ m_b \coloneqq \widetilde{m}_b \oplus \bigoplus_{i \in \overline{T}}t'_i,\] and output $(\regA',(b,m_b))$.
\end{itemize}

\noindent\underline{$\cH_3(\regA)$}
\begin{itemize}
    \item Prepare $\ell$ collections of EPR pairs on registers $\{\regS_i,\regR_i\}_{i \in [\ell]}$, and sample $(\ck,\ek) \gets \ExtGen(1^\secp)$, $\crs \gets \{0,1\}^n$, and $\hk \gets \{0,1\}^k$.
    \item Run $\Adv_\secp$ on input $\regA,\{\regS_i\}_{i \in [\ell]},\ck,\crs,\hk$ until it outputs a message
    \[\left( \{\cm_i\}_{i \in [\ell]},\{v_i,x_i,h_i,r_i\}_{i \in T}, \{\widehat{\cm}_{i,0},\widehat{\cm}_{i,1},z_{i,0},z_{i,1}\}_{i \in \overline{T}}, \pi, \widetilde{m}_0,\widetilde{m}_1\right)\] and a state on register $\regA'$.
    \item Run Steps 1-3 of the Receiver's honest computation on the sender's message.
    \item Attempt to project registers $\{\regR_i\}_{i \in \overline{T}}$ onto \[\bigotimes_{i \in \overline{T}}\Pi[\ck,\widehat{\cm}_{i,0},z_{i,0},z_{i,1}]^{\regR_i},\] and abort if the projection fails.
    \item \textcolor{red}{For each $i \in \overline{T},b \in \{0,1\}$, compute $t_{i,b} \gets \Ext(\ek,\widehat{\cm}_{i,b})$. Abort if any $t_{i,b} = \bot$ or if there does not exist $\Delta$ such that $t_{i,1} = \Delta \oplus t_{i,0}$ for all $i \in \overline{T}$.}
    \item If there was an abort, output $(\regA',\bot)$. Otherwise, for each $i \in \overline{T}$, measure register $\cR_i$ in the standard basis to obtain $b_i \in \{0,1\}$ and $v_i' \in \{0,1\}^{2\secp}$, and compute $(t_i',r_i') = z_{i,b_i} \oplus v_i'$. Then, define \[b \coloneqq \bigoplus_{i \in \overline{T}} b_i, ~~ m_b \coloneqq \widetilde{m}_b \oplus \bigoplus_{i \in \overline{T}}t'_i,\] and output $(\regA',(b,m_b))$.
\end{itemize}

\noindent\underline{$\cH_4(\regA)$}
\begin{itemize}
    \item Prepare $\ell$ collections of EPR pairs on registers $\{\regS_i,\regR_i\}_{i \in [\ell]}$, and sample $(\ck,\ek) \gets \ExtGen(1^\secp)$, $\crs \gets \{0,1\}^n$, and $\hk \gets \{0,1\}^k$.
    \item Run $\Adv_\secp$ on input $\regA,\{\regS_i\}_{i \in [\ell]},\ck,\crs,\hk$ until it outputs a message
    \[\left( \{\cm_i\}_{i \in [\ell]},\{v_i,x_i,h_i,r_i\}_{i \in T}, \{\widehat{\cm}_{i,0},\widehat{\cm}_{i,1},z_{i,0},z_{i,1}\}_{i \in \overline{T}}, \pi, \widetilde{m}_0,\widetilde{m}_1\right)\] and a state on register $\regA'$.
    \item Run Steps 1-3 of the Receiver's honest computation on the sender's message.
    \item Attempt to project registers $\{\regR_i\}_{i \in \overline{T}}$ onto \[\bigotimes_{i \in \overline{T}}\Pi[\ck,\widehat{\cm}_{i,0},z_{i,0},z_{i,1}]^{\regR_i},\] and abort if the projection fails.
    \item For each $i \in \overline{T},b \in \{0,1\}$, compute $t_{i,b} \gets \Ext(\ek,\widehat{\cm}_{i,b})$. Abort if any $t_{i,b} = \bot$ or if there does not exist $\Delta$ such that $t_{i,1} = \Delta \oplus t_{i,0}$ for all $i \in \overline{T}$.
    \item If there was an abort, output $(\regA',\bot)$. Otherwise, for each $i \in \overline{T}$, measure register $\textcolor{red}{\cR_i^{\ctl}}$ in the standard basis to obtain $b_i \in \{0,1\}$. Then, define \[b \coloneqq \bigoplus_{i \in \overline{T}} b_i, ~~ \textcolor{red}{m_0 \coloneqq \bigoplus_{i \in \overline{T}}t_{i,0}, ~~ m_1 \coloneqq \widetilde{m}_1 \oplus \Delta \oplus \bigoplus_{i \in \overline{T}}t_{i,0}},\] and output $(\regA',(b,m_b))$. 
\end{itemize}

\noindent\underline{$\cH_5(\regA)$}
\begin{itemize}
    \item Prepare $\ell$ collections of EPR pairs on registers $\{\regS_i,\regR_i\}_{i \in [\ell]}$, and sample $(\ck,\ek) \gets \ExtGen(1^\secp)$, $\crs \gets \{0,1\}^n$, and $\hk \gets \{0,1\}^k$.
    \item Run $\Adv_\secp$ on input $\regA,\{\regS_i\}_{i \in [\ell]},\ck,\crs,\hk$ until it outputs a message
    \[\left( \{\cm_i\}_{i \in [\ell]},\{v_i,x_i,h_i,r_i\}_{i \in T}, \{\widehat{\cm}_{i,0},\widehat{\cm}_{i,1},z_{i,0},z_{i,1}\}_{i \in \overline{T}}, \pi, \widetilde{m}_0,\widetilde{m}_1\right)\] and a state on register $\regA'$.
    \item Run Steps 1-3 of the Receiver's honest computation on the sender's message.
    \item We will insert a measurement on the registers $\{\cR_i\}_{i \in \overline{T}}$. Before specifying this measurement, we introduce some notation.
    \begin{itemize}
        \item For $\{(v_i,x_i,h_i)\}_{i \in \overline{T}}$ and a string $e \in \{0,1\}^{|\overline{T}|}$, define
        \[\Pi[e,\{(v_i,x_i,h_i)\}_{i \in \overline{T}}]^{\{\regR_i\}_{i \in \overline{T}}} \coloneqq \bigotimes_{i : e_i = 0}\dyad{\psi_{v_i,x_i,h_i}}{\psi_{v_i,x_i,h_i}}^{\regR_i} \otimes \bigotimes_{i : e_i = 1} \bbI - \dyad{\psi_{v_i,x_i,h_i}}{\psi_{v_i,x_i,h_i}}^{\regR_i}.\]
        \item For $\{(v_i,x_i,h_i)\}_{i \in \overline{T}}$ and a constant $\gamma \in [0,1]$, define
        \[\Pi[\gamma,\{(v_i,x_i,h_i)\}_{i \in \overline{T}}]^{\{\regR_i\}_{i \in \overline{T}}} \coloneqq \sum_{e \in \{0,1\}^{|S|} : \hw(e) < \gamma |\overline{T}|} \Pi[e,\{(v_i,x_i,h_i)\}_{i \in \overline{T}}]^{\{\regR_i\}_{i \in \overline{T}}}.\]
    \end{itemize}
    
    \textcolor{red}{Compute $(v_i,x_i,h_i) \gets \Ext(\ek,\cm_i)$ for each $i \in \overline{T}$.  Attempt to project registers $\{\regR_i\}_{i \in \overline{T}}$ onto \[\Pi\left[1/30, \{(v_i,x_i,h_i)\}_{i \in \overline{T}}\right],\] and abort if this projection fails.}
    \item Attempt to project registers $\{\regR_i\}_{i \in \overline{T}}$ onto \[\bigotimes_{i \in \overline{T}}\Pi[\ck,\widehat{\cm}_{i,0},z_{i,0},z_{i,1}]^{\regR_i},\] and abort if the projection fails.
    \item For each $i \in \overline{T},b \in \{0,1\}$, compute $t_{i,b} \gets \Ext(\ek,\widehat{\cm}_{i,b})$. Abort if any $t_{i,b} = \bot$ or if there does not exist $\Delta$ such that $t_{i,1} = \Delta \oplus t_{i,0}$ for all $i \in \overline{T}$.
    \item If there was an abort, output $(\regA',\bot)$. Otherwise, for each $i \in \overline{T}$, measure register $\cR_i^{\ctl}$ in the standard basis to obtain $b_i \in \{0,1\}$. Then, define \[b \coloneqq \bigoplus_{i \in \overline{T}} b_i, ~~ m_0 \coloneqq \bigoplus_{i \in \overline{T}}t_{i,0}, ~~ m_1 \coloneqq \widetilde{m}_1 \oplus \Delta \oplus \bigoplus_{i \in \overline{T}}t_{i,0},\] and output $(\regA',(b,m_b))$. 
\end{itemize}

\noindent\underline{$\cH_6(\regA)$}
\begin{itemize}
    \item Prepare $\ell$ collections of EPR pairs on registers $\{\regS_i,\regR_i\}_{i \in [\ell]}$, and sample $(\ck,\ek) \gets \ExtGen(1^\secp)$, $\crs \gets \{0,1\}^n$, and $\hk \gets \{0,1\}^k$.
    \item Run $\Adv_\secp$ on input $\regA,\{\regS_i\}_{i \in [\ell]},\ck,\crs,\hk$ until it outputs a message
    \[\left( \{\cm_i\}_{i \in [\ell]},\{v_i,x_i,h_i,r_i\}_{i \in T}, \{\widehat{\cm}_{i,0},\widehat{\cm}_{i,1},z_{i,0},z_{i,1}\}_{i \in \overline{T}}, \pi, \widetilde{m}_0,\widetilde{m}_1\right)\] and a state on register $\regA'$.
    \item Run Steps 1-3 of the Receiver's honest computation on the sender's message.
    \item Compute $(v_i,x_i,h_i) \gets \Ext(\ek,\cm_i)$ for each $i \in \overline{T}$.  Attempt to project registers $\{\regR_i\}_{i \in \overline{T}}$ onto \[\Pi\left[1/30, \{(v_i,x_i,h_i)\}_{i \in \overline{T}}\right],\] and abort if this projection fails.
    \item Attempt to project registers $\{\regR_i\}_{i \in \overline{T}}$ onto \[\bigotimes_{i \in \overline{T}}\Pi[\ck,\widehat{\cm}_{i,0},z_{i,0},z_{i,1}]^{\regR_i},\] and abort if the projection fails.
    \item \textcolor{red}{Attempt to project registers $\{\regR_i\}_{i \in \overline{T}}$ onto \[\Pi\left[1/2, \{(v_i,x_i,h_i)\}_{i \in \overline{T}}\right],\] and abort if this projection fails.}
    \item For each $i \in \overline{T},b \in \{0,1\}$, compute $t_{i,b} \gets \Ext(\ek,\widehat{\cm}_{i,b})$. Abort if any $t_{i,b} = \bot$ or if there does not exist $\Delta$ such that $t_{i,1} = \Delta \oplus t_{i,0}$ for all $i \in \overline{T}$.
    \item If there was an abort, output $(\regA',\bot)$. Otherwise, for each $i \in \overline{T}$, measure register $\cR_i^{\ctl}$ in the standard basis to obtain $b_i \in \{0,1\}$. Then, define \[b \coloneqq \bigoplus_{i \in \overline{T}} b_i, ~~ m_0 \coloneqq \bigoplus_{i \in \overline{T}}t_{i,0}, ~~ m_1 \coloneqq \widetilde{m}_1 \oplus \Delta \oplus \bigoplus_{i \in \overline{T}}t_{i,0},\] and output $(\regA',(b,m_b))$. 
\end{itemize}

\noindent\underline{$\cH_7(\regA)$}
\begin{itemize}
    \item Prepare $\ell$ collections of EPR pairs on registers $\{\regS_i,\regR_i\}_{i \in [\ell]}$, and sample $(\ck,\ek) \gets \ExtGen(1^\secp)$, $\crs \gets \{0,1\}^n$, and $\hk \gets \{0,1\}^k$.
    \item Run $\Adv_\secp$ on input $\regA,\{\regS_i\}_{i \in [\ell]},\ck,\crs,\hk$ until it outputs a message
    \[\left( \{\cm_i\}_{i \in [\ell]},\{v_i,x_i,h_i,r_i\}_{i \in T}, \{\widehat{\cm}_{i,0},\widehat{\cm}_{i,1},z_{i,0},z_{i,1}\}_{i \in \overline{T}}, \pi, \widetilde{m}_0,\widetilde{m}_1\right)\] and a state on register $\regA'$.
    \item Run Steps 1-3 of the Receiver's honest computation on the sender's message.
    \item Compute $(v_i,x_i,h_i) \gets \Ext(\ek,\cm_i)$ for each $i \in \overline{T}$.  Attempt to project registers $\{\regR_i\}_{i \in \overline{T}}$ onto \[\Pi\left[1/30, \{(v_i,x_i,h_i)\}_{i \in \overline{T}}\right],\] and abort if this projection fails.
    \item Attempt to project registers $\{\regR_i\}_{i \in \overline{T}}$ onto \[\bigotimes_{i \in \overline{T}}\Pi[\ck,\widehat{\cm}_{i,0},z_{i,0},z_{i,1}]^{\regR_i},\] and abort if the projection fails.
    \item Attempt to project registers $\{\regR_i\}_{i \in \overline{T}}$ onto \[\Pi\left[1/2, \{(v_i,x_i,h_i)\}_{i \in \overline{T}}\right],\] and abort if this projection fails.
    \item For each $i \in \overline{T},b \in \{0,1\}$, compute $t_{i,b} \gets \Ext(\ek,\widehat{\cm}_{i,b})$. Abort if any $t_{i,b} = \bot$ or if there does not exist $\Delta$ such that $t_{i,1} = \Delta \oplus t_{i,0}$ for all $i \in \overline{T}$.
    \item If there was an abort, output $(\regA',\bot)$. Otherwise, \textcolor{red}{sample $b \gets \{0,1\}$}. Then, define \[m_0 \coloneqq \bigoplus_{i \in \overline{T}}t_{i,0}, ~~ m_1 \coloneqq \widetilde{m}_1 \oplus \Delta \oplus \bigoplus_{i \in \overline{T}}t_{i,0},\] and output $(\regA',(b,m_b))$. 
\end{itemize}

\noindent\underline{$\cH_8(\regA)$}
\begin{itemize}
    \item Prepare $\ell$ collections of EPR pairs on registers $\{\regS_i,\regR_i\}_{i \in [\ell]}$, and sample $(\ck,\ek) \gets \ExtGen(1^\secp)$, $\crs \gets \{0,1\}^n$, and $\hk \gets \{0,1\}^k$.
    \item Run $\Adv_\secp$ on input $\regA,\{\regS_i\}_{i \in [\ell]},\ck,\crs,\hk$ until it outputs a message
    \[\left( \{\cm_i\}_{i \in [\ell]},\{v_i,x_i,h_i,r_i\}_{i \in T}, \{\widehat{\cm}_{i,0},\widehat{\cm}_{i,1},z_{i,0},z_{i,1}\}_{i \in \overline{T}}, \pi, \widetilde{m}_0,\widetilde{m}_1\right)\] and a state on register $\regA'$.
    \item Run Steps 1-3 of the Receiver's honest computation on the sender's message.
    \item Attempt to project registers $\{\regR_i\}_{i \in \overline{T}}$ onto \[\bigotimes_{i \in \overline{T}}\Pi[\ck,\widehat{\cm}_{i,0},z_{i,0},z_{i,1}]^{\regR_i},\] and abort if the projection fails.
    \item For each $i \in \overline{T},b \in \{0,1\}$, compute $t_{i,b} \gets \Ext(\ek,\widehat{\cm}_{i,b})$. Abort if any $t_{i,b} = \bot$ or if there does not exist $\Delta$ such that $t_{i,1} = \Delta \oplus t_{i,0}$ for all $i \in \overline{T}$.
    \item If there was an abort, output $(\regA',\bot)$. Otherwise, sample $b \gets \{0,1\}$. Then, define \[m_0 \coloneqq \bigoplus_{i \in \overline{T}}t_{i,0}, ~~ m_1 \coloneqq \widetilde{m}_1 \oplus \Delta \oplus \bigoplus_{i \in \overline{T}}t_{i,0},\] and output $(\regA',(b,m_b))$. 
\end{itemize}

\noindent\underline{$\cH_9(\regA)$ / $\Sim(\regA)$}
\begin{itemize}
    \item Prepare $\ell$ collections of EPR pairs on registers $\{\regS_i,\regR_i\}_{i \in [\ell]}$, and sample $(\ck,\ek) \gets \ExtGen(1^\secp)$, $\crs \gets \{0,1\}^n$, and $\hk \gets \{0,1\}^k$.
    \item Run $\Adv_\secp$ on input $\regA,\{\regS_i\}_{i \in [\ell]},\ck,\crs,\hk$ until it outputs a message
    \[\left( \{\cm_i\}_{i \in [\ell]},\{v_i,x_i,h_i,r_i\}_{i \in T}, \{\widehat{\cm}_{i,0},\widehat{\cm}_{i,1},z_{i,0},z_{i,1}\}_{i \in \overline{T}}, \pi, \widetilde{m}_0,\widetilde{m}_1\right)\] and a state on register $\regA'$.
    \item Run Steps 1-3 of the Receiver's honest computation on the sender's message.
    \item Attempt to project registers $\{\regR_i\}_{i \in \overline{T}}$ onto \[\bigotimes_{i \in \overline{T}}\Pi[\ck,\widehat{\cm}_{i,0},z_{i,0},z_{i,1}]^{\regR_i},\] and abort if the projection fails.
    \item For each $i \in \overline{T},b \in \{0,1\}$, compute $t_{i,b} \gets \Ext(\ek,\widehat{\cm}_{i,b})$. Abort if any $t_{i,b} = \bot$ or if there does not exist $\Delta$ such that $t_{i,1} = \Delta \oplus t_{i,0}$ for all $i \in \overline{T}$.
    \item If there was an abort, \textcolor{red}{send $\bot$ to the ideal functionality, and output $\cA'$.} Otherwise, define \[m_0 \coloneqq \bigoplus_{i \in \overline{T}}t_{i,0}, ~~ m_1 \coloneqq \widetilde{m}_1 \oplus \Delta \oplus \bigoplus_{i \in \overline{T}}t_{i,0},\] \textcolor{red}{send $(m_0,m_1)$ to the ideal functionality, and output $\cA'$.} 
\end{itemize}

Observe that $\cH_9(\cA)$ describes the behavior of a simulator $\Sim$ that operates on input register $\cA$, and interacts with the ideal functionality $\cF_\ROT$. Thus, The following sequence of claims completes the proof.

\begin{claim}
$\cH_0 \approx_c \cH_1$.
\end{claim}

\begin{proof}
This follows directly from the extractability of the commitment (\cref{def:extractability}).
\end{proof}

\begin{claim}
$\cH_1 \equiv \cH_2$.
\end{claim}

\begin{proof}
The only difference is that we have applied the Step 4 check coherently before measuring in the standard basis. Since these measurements commute, these hybrids describe the same distribution. 
\end{proof}

\begin{claim}
$\cH_2 \approx_s \cH_3$.
\end{claim}

\begin{proof}
The newly introcued abort condition will only be triggered with negligible probability due to the soundness of the NIZK (\cref{def:ZK-soundness}) and the extractability of the commitment (\cref{def:extractability}).
\end{proof}

\begin{claim}
$\cH_3 \approx_s \cH_4$.
\end{claim}

\begin{proof}
We are now defining $m_0,m_1$ based on the strings extracted by $\Ext$ rather than the strings measured by the Receiver. Since the strings measured by the Receiver must be valid commitment openings, this only introduces a negligible difference due to the extractability of the commitment (\cref{def:extractability}).
\end{proof}


\begin{claim}\label{claim:4-to-5}
$\cH_4 \approx_s \cH_5$.
\end{claim}

\begin{proof}
By Gentle Measurement (\cref{lemma:gentle-measurement}), it suffices to argue that the projection introduced in $\cH_5$ will succeed with probability $1-\negl(\secp)$. So towards contradiction, assume that the projection fails with non-negligible probability. We will eventually use this assumption to break the correlation intractability of $H$. First, consider the following experiment.\\

\noindent\underline{$\Exp_1$}

\begin{itemize}
    \item Prepare $\ell$ collections of EPR pairs on registers $\{\regS_i,\regR_i\}_{i \in [\ell]}$. Sample $(\ck,\ek) \gets \ExtGen(1^\secp)$, $\crs \gets \{0,1\}^n$, and $\hk \gets \{0,1\}^k$.
    \item Run $\Adv_\secp$ on input $\regA,\{\regS_i\}_{i \in [\ell]},\ck,\crs,\hk$, and receive a message that includes $\{\cm_i\}_{i \in [\ell]},\allowbreak\{\widehat{\cm}_{i,0},\widehat{\cm}_{i,1}\}_{i \in \overline{T}},\allowbreak\pi$. 
    \item Compute $T = H_\secp(\hk,(\cm_1,\dots,\cm_\ell))$, check that for all $i \in T$, $\cm_i = \Com(\ck,(v_i,x_i,h_i);r_i)$, and that $\NIZK.\Ver\left(\crs,\left(\{\widehat{\cm}_{i,0},\widehat{\cm}_{i,1}\}_{i \in \overline{T}},\{\cm_i\}_{i \in [\ell]}\right),\pi\right) = \top$, and abort if not.
    \item For each $i \in [\ell]$, compute $(v_i,x_i,h_i) \gets \Ext(\ek,\cm_i)$, and abort if any are $\bot$.
    \item For each $i \in [\ell]$, measure registers $\regR_i$ in the basis $\{\dyad{\psi_{v_i,x_i,h_i}}{\psi_{v_i,x_i,h_i}}, \bbI - \dyad{\psi_{v_i,x_i,h_i}}{\psi_{v_i,x_i,h_i}}\}$ and define the bit $e_i = 0$ if the first outome is observed and $e_i = 1$ if the second outcome is observed.
    \item Output 1 if (i) there exists an $s \in \{0,1\}^\secp$ such that $(x_1,\dots,x_\ell) = \PRG(s)$,\footnote{Note that this step is not efficient to implement, but this will not be important for our arguments.} (ii) $e_i = 0$ for all $i \in T$, and (iii) $e_i = 1$ for at least $1/30$ fraction of $i : i \in \overline{T}$.
\end{itemize}

We claim that $\Pr[\Exp_1 \to 1] = \nonnegl(\secp)$. This nearly follows from the assumption that the measurement introduced in $\cH_5$ rejects with non-negligible probability, except for the following two differences. One difference from $\cH_3$ is that in $\Exp_1$, we are using $\{(v_i,x_i,h_i)\}_{i \in T}$ extracted from $\{\cm_i\}_{i \in T}$ to measure registers $\{\regR_i\}_{i \in T}$, rather than the strings sent by the adversary. However, this introduces a negligible difference due to the extractability of the commitment scheme. The other difference is that we require $(x_1,\dots,x_\ell)$, which are extracted from $\{\cm_i\}_{i \in [\ell]}$, to be in the image of $\PRG(\cdot)$. However, by extractability of the commitment scheme and soundness of the NIZK, the probability that the procedure does not abort and this fails to occur is negligible. Next, consider the following experiment. \\

\noindent\underline{$\Exp_2$}

\begin{itemize}
    \item Prepare $\ell$ collections of EPR pairs on registers $\{\regS_i,\regR_i\}_{i \in [\ell]}$. Sample $(\ck,\ek) \gets \ExtGen(1^\secp)$, $\crs \gets \{0,1\}^n$, and $\hk \gets \{0,1\}^k$.
    \item Sample $s^* \gets \{0,1\}^\secp$ and set $(x_1^*,\dots,x_\ell^*) = \PRG(s^*)$. For each $i \in [\ell]$ and $j \in [2\secp]$ such that $x_{i,j} = 1$, apply a CNOT gate from register $\regR_i^\ctl$ to $\regR_{i,j}^\msg$, then measure $\regR_i^\ctl$ in the Hadamard basis to obtain $h_i^*$ and measure $\regR_{i,1}^\msg,\dots,\regR_{i,2\secp}^\msg$ in the standard basis to obtain $v_i^*$.
    \item Run $\Adv_\secp$ on input $\regA,\{\regS_i\}_{i \in [\ell]},\ck,\crs,\hk$, and receive a message that includes $\{\cm_i\}_{i \in [\ell]}$.
    \item Compute $T = H_\secp(\hk,(\cm_1,\dots,\cm_\ell))$ and $(v_i,x_i,h_i) \gets \Ext(\ek,\cm_i)$ for each $i \in [\ell]$.
    \item Output 1 if (i) $(x_1,\dots,x_\ell) = (x_1^*,\dots,x_\ell^*)$, (ii) $(v_i,h_i) = (v_i^*,h_i^*)$ for all $i \in T$, and (iii) $(v_i,h_i) \neq (v_i^*,h_i^*)$ for at least $1/30$ fraction of $i : i \in \overline{T}$.
\end{itemize}

It follows that $\Pr[\Exp_2 \to 1] = \nonnegl(\secp)/2^\secp > 1/2^{\secp_\CI^\delta}$, since the guess of $s^*$ is uniformly random and independent of the adversary's view. Finally, we will show that $\Exp_2$ can be used to break the correlation intractability of $H$, but first we introduce some notation.

\begin{itemize}
    \item For each $(\ek,s^*,\{v_i^*,h_i^*\}_{i \in [\ell]})$, define the relation $R[\ek,s^*,\{v_i^*,h_i^*\}_{i \in [\ell]}]$ as follows. Recalling that $\ell = c \cdot t$, we will associate each $i \in [\ell]$ with a pair $(\iota,\kappa)$ for $\iota \in [t], \kappa \in [c]$. Also, for each set of strings $\{\cm_i\}_{i \in [\ell]}$, we fix $(v_i,x_i,h_i) \coloneqq \Ext(\ek,\cm_i)$ for each $i \in [\ell]$. Then the domain will consist of strings $\{\cm_i\}_{i \in [\ell]}$ such that (i) $(x_1,\dots,x_\ell) = \PRG(s^*)$, (ii) $|i : (v_i,h_i) = (v_i^*,h_i^*)| \leq (1-1/60)\ell$, and (iii) for each $\iota \in [t]$, $|\kappa : (v_{(\iota,\kappa)},h_{(\iota,\kappa)}) = (v_{(\iota,\kappa)}^*,h_{(\iota,\kappa)}^*)| \geq (1/2)c$.

    \item For each $\{\cm_i\}_{i \in [\ell]}$ in the domain of $R[\ek,s^*,\{v_i^*,h_i^*\}_{i \in [\ell]}]$, define the sets $\{S_{\iota,\{\cm_i\}_{i \in [\ell]}}\}_{\iota \in [t]}$ as follows. If $(1/2)c \leq |\kappa : (v_{(\iota,\kappa)},h_{(\iota,\kappa)}) = (v_{(\iota,\kappa)}^*,h_{(\iota,\kappa)}^*)| \leq (1-1/120)c$, let $S_{\iota,\{\cm_i\}_{i \in [\ell]}}$ consist of subsets $C \subset [c]$ of size $c/2$ such that for all $\kappa \in C$, $(v_{(\iota,\kappa)},h_{(\iota,\kappa)}) = (v_{(\iota,\kappa)}^*,h_{(\iota,\kappa)}^*)$. Otherwise, let $S_{\iota,\{\cm_i\}_{i \in [\ell]}} = \emptyset$. 
    \item Define the set $R[\ek,s^*,\{v_i^*,h_i^*\}_{i \in [\ell]}]_{\{\cm_i\}_{i \in [\ell]}}$ to consist of all $y = (C_1,\dots,C_t)$ such that $C_\iota \in S_{\iota,\{\cm\}_{i \in [\ell]}}$ for all $\iota$ such that $S_{\iota,\{\cm\}_{i \in [\ell]}} \neq \emptyset$. We claim that there are always at least $1/120$ fraction of $\iota \in [t]$ such that $S_{\iota,\{\cm_i\}_{i \in [\ell]}} \neq \emptyset$. To see this, note that $S_{\iota,\{\cm_i\}_{i \in [\ell]}} \neq \emptyset$ iff $|\kappa : (v_{(\iota,\kappa)},h_{(\iota,\kappa)}) \neq (v_{(\iota,\kappa)}^*,h_{(\iota,\kappa)}^*)| > (1/120)c$. However, if less 1/120 fraction of $\iota$ satisfies this condition, then the fraction of $i \in [\ell]$ such that $(v_i,h_i) \neq (v_i^*,h_i^*)$ is at most $(1/120) + (1/120)(1-1/120) < 1/60$, which would contradict the fact that $\{\cm_i\}_{i \in [\ell]}$ is in the domain of $R[\ek,s^*,\{v_i^*,h_i^*\}_{i \in [\ell]}]_{\{\cm_i\}_{i \in [\ell]}}$.

    Thus, $R[\ek,s^*,\{v_i^*,h_i^*\}_{i \in [\ell]}]$ is an $\alpha$-approximate efficiently verifiable product relation for $\alpha=1/120$ with sparsity $\rho = \binom{(1-\alpha)c}{(1/2)c}/2^c < \alpha$. 
    
\end{itemize}

Now, whenever $\Exp_2 = 1$, it must be the case that $\{\cm_i\}_{i \in [\ell]}$ is in the domain of $R[\ek,s^*,\{v_i^*,h_i^*\}_{i \in [\ell]}]$, and $T \in R[\ek,s^*,\{v_i^*,h_i^*\}_{i \in [\ell]}]_{\{\cm_i\}_{i \in [\ell]}}$. Thus, we can break correlation intractability as follows. Begin running $\Exp_2$, but don't sample $\hk$. Once $\ek,s^*$ are sampled and $\{v_i^*,h_i^*\}_{i \in [\ell]}$ are measured, declare the relation $R[\ek,s^*,\{v_i^*,h_i^*\}_{i \in [\ell]}]$. Then, receive $\hk$ from the correlation intractability challenger, continue running $\Exp_2$ until $\{\cm_i\}_{i \in [\ell]}$ is obtained, and output this to the challenger. The above analysis shows that this breaks correlation intractability for the relation $R[\ek,s^*,\{v_i^*,h_i^*\}_{i \in [\ell]}]$.

\end{proof}

\begin{claim}\label{claim:5-to-6}
$\cH_5 \approx_s \cH_6$.
\end{claim}

\begin{proof}
By Gentle Measurement (\cref{lemma:gentle-measurement}), it suffices to show that the projection introduced in $\cH_6$ will succeed with probability $1-\negl(\secp)$. To do so, we will rule out one bad case. For each $i \in \overline{T}$, define the bit $f_i = 0$ if and only if $\widehat{\cm}_{i,0} = \Com(\ck,t_{i,0};r_{i,0})$ and $\widehat{\cm}_{i,1} = \Com(\ck,t_{i,1};r_{i,1})$, where $(t_{i,0},r_{i,0}) = z_{i,0} \oplus v_i$, $(t_{i,1},r_{i,1}) = z_{i,1} \oplus v_i \oplus x_i$, and $(v_i,x_i,h_i) \coloneqq \Ext(\ek,\cm_i)$. Now we claim that if the fraction of $i \in \overline{T}$ such that $f_i = 1$ is $\geq 1/2-1/30$, then the attempted projection onto \[\bigotimes_{i \in \overline{T}}\Pi[\ck,\widehat{\cm}_{i,0},z_{i,0},z_{i,1}]^{\regR_i}\] performed during Step 4 of the receiver's computation would have failed with probabilty $1-\negl(\secp)$. To see this, consider any state $\ket{\psi}^{\{\regR_i\}_{i \in [\ell]},\regX}$ in the image of $\Pi\left[1/30,\{(v_i,x_i,h_i)\}_{i \in \overline{T}}\right]$, where $\regX$ is an arbitrary auxiliary register. Then, defining $\gamma = 1/30$, we can write $\ket{\psi}$ as 

\[\ket{\psi} \coloneqq \sum_{e \in \{0,1\}^{|\overline{T}|} : \hw(e) < \gamma|\overline{T}|}\left(\bigotimes_{i : e_i = 0}\ket{\psi_{v_i,x_i,h_i}}^{\regR_i}\right) \otimes \ket{\psi_e}^{\{\regR_i\}_{i : e_i = 1},\regX},\] where $\ket{\psi_e}$ is some unit vector that is orthogonal to $\ket{\psi_{v_i,x_i,h_i}}$ for all $i$ such that $e_i = 1$. Then,
\begin{align*}\bigg\|&\bigotimes_{i \in \overline{T}}\Pi[\ck,\widehat{\cm}_{i,0},z_{i,0},z_{i,1}]\ket{\psi}\bigg\|^2 \\
&\leq \bigg\| \sum_{e \in \{0,1\}^{|\overline{T}|}: \hw(e) < \gamma|\overline{T}|}\bigotimes_{i:e_i = 0}\Pi[\ck,\widehat{\cm}_i,z_{i,0},z_{i,1}]\ket{\psi_{v_i,x_i,h_i}}^{\regR_i}\bigg\|^2 \\
&\leq \binom{|\overline{T}|}{\gamma|\overline{T}|}\sum_{e \in \{0,1\}^{|\overline{T}|} : \hw(e) < \gamma|\overline{T}|} \bigg\|\bigotimes_{i : e_i = 0}\Pi[\ck,\widehat{\cm}_i,z_{i,0},z_{i,1}] \ket{\psi_{v_i,x_i,h_i}}^{\regR_i}\bigg\|^2 \\
&\leq \binom{|\overline{T}|}{\gamma |\overline{T}|}^2 \cdot 2^{-(1/2-2\gamma)|\overline{T}|} \\ 
&\leq (3/\gamma)^{2\gamma|\overline{T}|} \cdot 2^{-(1/2-2\gamma)|\overline{T}|}\\
&= 2^{|\overline{T}|(2\gamma\log(3/\gamma) - (1/2-2\gamma))}\\
&= \negl(\secp)\end{align*} where the second inequality is Cauchy-Schwartz, the third inequality follow from the fact that there are at least $1/2-2\gamma$ fraction of indices where $f_i = 1$ and $e_i = 0$, and the final equality follows because $\gamma=1/30$ is such that $2\gamma\log(3/\gamma) - (1/2-2\gamma) = O(1)$, and $|\overline{T}| = \ell/2 = \Omega(\secp)$, which means that the exponent is $\Omega(\secp)$.

Thus it suffices to consider the case where the fraction of $i \in \overline{T}$ such that $f_i = 1$ is $< 1/2-1/30$. So consider any state $\ket{\psi}^{\{\regR_i\}_{i \in [\ell]},\regX}$ in the image of $\Pi\left[1/30,\{(v_i,x_i,h_i)\}_{i \in \overline{T}}\right]$, which we can write as 

\[\ket{\psi} \coloneqq \sum_{e \in \{0,1\}^{|\overline{T}|} : \hw(e) < |\overline{T}|/30}\left(\bigotimes_{i : e_i = 0}\ket{\psi_{v_i,x_i,h_i}}^{\regR_i}\right) \otimes \ket{\psi_e}^{\{\regR_i\}_{i : e_i = 1},\regX}.\] Then,

\begin{align*}
    &\Pi[\ck,\widehat{\cm}_{i,0},z_{i,0},z_{i,1}]\ket{\psi}\\ &= \sum_{e \in \{0,1\}^{|\overline{T}|} : \hw(e) < |\overline{T}|/30} \left(\bigotimes_{i: e_i = 0 \wedge f_i = 0} \Pi[\ck,\widehat{\cm}_{i,0},z_{i,0},z_{i,1}]\ket{\psi_{v_i,x_i,h_i}}\right)\\ & \hspace{4cm} \otimes \left(\bigotimes_{i : e_i = 1 \vee f_i = 1}\Pi[\ck,\widehat{\cm}_{i,0},z_{i,0},z_{i,1}]\right)\ket{\psi_e} \\
    &= \sum_{e \in \{0,1\}^{|\overline{T}|} : \hw(e) < |\overline{T}|/30} \left(\bigotimes_{i: e_i = 0 \wedge f_i = 0}\ket{\psi_{v_i,x_i,h_i}}\right) \otimes \left(\bigotimes_{i : e_i = 1 \vee f_i = 1}\Pi[\ck,\widehat{\cm}_{i,0},z_{i,0},z_{i,1}]\right)\ket{\psi_e} \\
    &= \sum_{e' \in \{0,1\}^{|\overline{T}|} : \hw(e') < |\overline{T}|/2} \left(\bigotimes_{i : e'_i = 0}\ket{\psi_{v_i,x_i,h_i}}\right) \otimes \ket{\psi_{e'}} \\
    & \in \mathsf{Im}\left(\Pi\left[1/2,\{(v_i,x_i,h_i)\}_{i \in \overline{T}}\right]\right),
\end{align*}

where the $\ket{\psi_{e'}}$ are some set of unit vectors.



\end{proof}

\begin{claim}
$\cH_6 \equiv \cH_7$.
\end{claim}

\begin{proof}
It suffices to show that in $\cH_6$, the bit $b = \bigoplus_{i \in \overline{T}} b_i$ sampled by measuring registers $\{\regR_i^\ctl\}_{i \in \overline{T}}$ of $\ket{\psi}^{\{\regR_i\}_{i \in [\ell]},\regX}$ in the standard basis is uniformly random, even conditioned on the auxiliary register $\regX$ (which includes the view of the adversarial sender). This follows from \cref{thm:XOR-extractor} by applying a change of basis. In more detail, define the unitary $U_{v_i,x_i,h_i}$ to be applied to $\regR_i$ as follows: For each $j \in [2\secp]$ such that $x_{i,j} = 1$ apply a CNOT gate from $\regR_{i,j}^\ctl$ to  $\regR_{i,j}^\msg$, then apply a classically controlled phase flip $Z^{h_i}$ to $\regR_i^\ctl$, and finally apply a Hadamard gate to $\regR_i^\ctl$. In particular, \[U_{v_i,x_i,h_i}\ket{\psi_{v_i,x_i,h_i}} = \ket{0}\ket{v_i}.\] Thus, for any $\ket{\psi} \in \mathsf{Im}(\Pi\left[1/2,\{(v_i,x_i,h_i)\}_{i \in \overline{T}}\right])$, it holds that registers $\{\regR_i^\ctl\}_{i \in \overline{T}}$ of $(\bigotimes_{i \in \overline{T}}U_{v_i,x_i,h_i})\ket{\psi}$ are in a superposition of standard basis states with Hamming weight $< |\overline{T}|/2$. Since applying $U^\dagger_{v_i,x_i,h_i}$ to a standard basis measurement of $\regR_i^\ctl$ yields a Hadamard basis measurement of $\regR_i^\ctl$, \cref{thm:XOR-extractor} directly implies that the bit $b = \bigoplus_{i \in \overline{T}} b_i$ is uniformly random, even conditioned on the auxiliary register $\regX$. 

\end{proof}

\begin{claim}
$\cH_7 \approx_s \cH_8$.
\end{claim}

\begin{proof}
We are removing the two measurements introduced in hybrids $\cH_5$ and $\cH_6$, and indistinguishability follows from the same arguments used in the corresponding claims \cref{claim:4-to-5} and \cref{claim:5-to-6}.
\end{proof}

\begin{claim}
$\cH_8 \equiv \cH_9$.
\end{claim}

\begin{proof}
This is just a syntactic switch, routing information through the ideal functionality $\cF_\ROT$.
\end{proof}

\end{proof}

\begin{lemma}\label{lemma:ROT-receiver-security}
The protocol in \cref{fig:NIOT} is secure against a malicious receiver.
\end{lemma}

\begin{proof}
Let $\{\Adv_\secp\}_{\secp \in \bbN}$ be a QPT adversary corrupting the receiver, which takes as input register $\regA$ of $\{\ket{\psi_\secp}^{\regA,\regD}\}_{\secp \in \bbN}$. Instead of explicitly considering a register $\regX$ holding the honest sender's input, we write their input as classical strings $(m_0,m_1) \in \{0,1\}^\secp$. We define a simulator $\{\Sim_\secp\}_{\secp \in \bbN}$ as follows.\\

\noindent\underline{$\Sim_\secp(\regA)$}
\begin{itemize}
    \item Obtain $(b,m_b)$ from the ideal functionality.
    \item Sample $\ck \gets \{0,1\}^h$.
    \item Sample $T \subset [\ell]$ as a uniformly random sequence of $t$ subsets of $[c]$ of size $c/2$. 
    \item For $i \in T$, sample $v_i,x_i \gets \{0,1\}^{2\secp}, h_i \gets \{0,1\},r_i \gets \{0,1\}^\secp$, set $\cm_i \coloneqq \Com(\ck,(v_i,x_i,h_i);r_i)$, and initialize register $\regR_i$ to the state $\frac{1}{\sqrt{2}}\left({\ket{0,v_i} + (-1)^{h_i}\ket{1,v_i \oplus x_i}}\right)$.
    \item For $i \in \overline{T}$, sample $\cm_i \gets \Com(\ck,0)$.
    \item Sample random bits $\{b_i\}_{i \in \overline{T}}$ conditioned on $\bigoplus_{i \in \overline{T}}b_i = b$.
    \item For $i \in \overline{T}$, sample $v_i'\gets \{0,1\}^{2\secp},t_i',r_i' \gets \{0,1\}^\secp,z_{i,1-b} \gets \{0,1\}^{2\secp}$, set $\widehat{\cm}_{i,b_i} \coloneqq \Com(\ck,t_i';r_i')$, $\widehat{\cm}_{i,1-b_i} \gets \Com(\ck,0)$, $z_{i,b_i} \coloneqq (t_i',r_i') \oplus v_i'$, and initialize register $\regR_i$ to the state $\ket{b_i,v_i'}$.
    \item Compute $(\crs,\pi) \gets \NIZK.\Sim\left(1^\secp,\left(\{\widehat{\cm}_{i,0},\widehat{\cm}_{i,1}\}_{i \in \overline{T}},\{\cm_i\}_{i \in [\ell]}\right)\right)$.
    \item Sample $\hk \gets \Samp(1^\secp,(\cm_1,\dots,\cm_\ell),T)$.
    \item Set $\widetilde{m}_b \coloneqq m_b \oplus \bigoplus_{i \in \overline{T}}t_i'$ and sample $\widetilde{m}_{1-b} \gets \{0,1\}^\secp$.
    \item Run $\Adv_\secp$ on input $\regA, \{\regR_i\}_{i \in [\ell]}$, and \[\left(\{\cm_i\}_{i \in [\ell]}, \{v_i,x_i,h_i,r_i\}_{i \in T}, \{\widehat{\cm}_{i,0},\widehat{\cm}_{i,1},z_{i,0},z_{i,1}\}_{i \in \overline{T}}, \pi, \widetilde{m}_0,\widetilde{m}_1\right),\] and output their final state on register $\regA'$.
\end{itemize}


\noindent Now, we define a sequence of hybrids, and argue indistinguishability between each adjacent pair.

\begin{itemize}
    \item $\cH_0$: This is the real distribution $\Pi_{\cF_\ROT}[\Adv_\secp,\ket{\psi_\secp}]$.
    \item $\cH_1$: Same as $\cH_0$ except that $(\crs,\pi) \gets \NIZK.\Sim\left(1^\secp,\left(\{\widehat{\cm}_{i,0},\widehat{\cm}_{i,1}\}_{i \in \overline{T}},\{\cm_i\}_{i \in [\ell]}\right)\right)$. Computational indistinguishability from $\cH_0$ follows directly from zero-knowledge of $\NIZK$ (\cref{def:ZK-ZK}).
    \item $\cH_2$: Same as $\cH_1$ except that $(x_1,\dots,x_\ell) \gets \{0,1\}^{2\secp\ell}$ are sampled as uniformly random strings. Computational indistinguishability from $\cH_1$ follows directly from the pseudorandomness of $\PRG$.
    \item $\cH_3$: Same as $\cH_2$ except that $T \subset [\ell]$ is sampled as described in the simulator, and the hash key $\hk$ is sampled at the end of the Sender's computation by $\hk \gets \Samp(1^\secp,(\cm_1,\dots,\cm_\ell),T)$. This is the same distribution as $\cH_2$, which follows from the programmability of the correlation-intractable hash function (\cref{def:programmability}).
    \item $\cH_4$: Same as $\cH_3$ except that $\{\cm_i \gets \Com(\ck,0)\}_{i \in \overline{T}}$ are sampled as commitments to 0. Computational indistinguishability from $\cH_3$ follows directly from hiding of the commitment scheme (\cref{def:hiding}).
    \item $\cH_5$: Same as $\cH_4$ except the registers $\{\regS_i^\ctl\}_{i \in \overline{T}}$ are measured by the Sender in the standard basis rather than the Hadamard basis. We will let $\{b_i\}_{i \in \overline{T}}$ be the bits measured. This is the same distribution as $\cH_4$ since these registers are outside the adverary's view, and the results of measuring these registers are independent of the adversary's view.
    \item $\cH_6$: Same as $\cH_5$ except that $\{\widehat{\cm}_{i,1-b_i} \gets \Com(\ck,0)\}_{i \in \overline{T}}$ are sampled as commitments to 0. Computational indistinguishability from $\cH_5$ follows directly from hiding of the commitment scheme.
    \item $\cH_7$: Same as $\cH_6$ except that $\widetilde{m}_{1-b}$ is sampled as a random string, where $b \coloneqq \bigoplus_{i \in \overline{T}} b_i$. This is the same distribution as $\cH_6$ since $\Delta \gets \{0,1\}^\secp$ is independent of the adversary's view. Also observe that for any honest sender inputs $m_0,m_1 \in \{0,1\}^\secp$, this is the same as the simulated distribution $\widetilde{\Pi}_{\cF_\ROT}[\Sim_\secp,\ket{\psi}]$, which completes the proof. 
\end{itemize}

\end{proof}

\subsection{Applications}\label{subsec:applications}


First, we make use of a result from \cite{C:GIKOS15} (based on earlier work of \cite{C:IshPraSah08}) showing that any unidirectional randomized classical functionality can be securely computed in one message given access to a secure one-message protocol for random string OT. In our words, their theorem is the following.

\begin{importedtheorem}[\cite{C:GIKOS15}]
For any polynomial-size classical circuit $C(x,r) \to y$, there exists a one-shot secure computation protocol for $\cF_\CL[C]$ given polynomially many parallel queries to $\cF_\ROT$.
\end{importedtheorem}

Thus, we immediately obtain the following corollary of \cref{thm:one-shot-security}.

\begin{corollary}\label{cor:classical-functionalities}
For any polynomial-size classical circuit $C(x,r) \to y$, there exists a one-shot secure computation protocol for $\cF_\CL[C]$ in the shared EPR pair model, assuming the sub-exponential hardness of LWE.
\end{corollary}


\paragraph{Secure teleportation through quantum channels.} Next, we consider the secure computation of unidirectional \emph{quantum} functionalities. The setting here is that, given the public description of a quantum map $Q$, a sender with an arbitrary state on register $\cS$ can send a single message to the receiver who will recover the state $Q(\cS)$ on register $\cR$. The receiver is guaranteed to learn only $Q(\cS)$. Security also guarantees a strong notion of hiding against a malicious sender, whose view will have no more information than if they had just traced register $\cS$ out of their system. In particular, the sender won't learn anything about the resulting state $Q(\cS)$ or any other garbage registers that may have been created while computing $Q(\cS)$. In other words, this notion captures an ideal scenario where a sender places register $\cS$ into an honest channel implementing $Q$, and the receiver obtains the state at the other end.

In general, this notion corresponds to \emph{secure two-party quantum computation} where only one party has input and the other party has output. However, acheiving this task in only one message (which can be classical without loss of generality) in the shared EPR pair model is reminiscent of the setting of \emph{quantum teleportation}, except that there are additional strong integrity and privacy guarantees about the recevied state. This motivates the following definition.

\begin{definition}[Secure Teleportation through Quantum Channel]
A {secure teleportation protocol through a quantum channel $Q$} is a one-message protocol in the shared EPR pair model that can securely realize the ideal functionality $\cF_\QU[Q]$ for any efficient quantum map $Q$.
\end{definition}

We now use a result from \cite{BCKM} showing that any unidirectional quantum functionality can be securely computed in one message given access to a secure one-message protocol for arbitrary unidirectional classical functionalities. Since this theorem is not explicitly stated in \cite{BCKM}, we add some explanation for how it follows from their work.

\begin{importedtheorem}[\cite{BCKM}]
For any polynomial-size quantum circuit $Q(\regX) \to \regY$, there exists a one-shot secure computation protocol for $\cF_\QU[Q]$ given one query to $\cF_\CL[C_Q]$ for some polynomial-size classical circuit $C_Q$ that depends on $Q$.
\end{importedtheorem}

\begin{proof}
\cite[Section 5.1]{BCKM} presents and shows the security of a generic two-party quantum computation protocol that is three-round when both parties obtain output and two-round when only one party obtains output. Here, we are interested in the two-round case, where the receiver sends the first message and the sender sends the second message. In particular, we note that in the case that the receiver has no input, their first message only consists of the first message of an underlying \emph{classical} two-party computation protocol for a classical functionality in which the receiver has no input (this fact was also used in the proof of \cite[Theorem 6.2]{BCKM}). Since we are assuming that classical functionalities can be computed in a single message from sender to receiver, this entirely removes the need for the receiver to send its first message, completing the proof.
\end{proof}

\begin{remark} We describe at a very high level how the resulting one-message quantum protocol operates. This description will be somewhat imprecise and simplified by design.

It makes use of a type of \emph{quantum garbled circuit} first described in \cite{10.1145/3519935.3520073} and formalized by \cite{BCKM}. This quantum garbled circuit is based on the Clifford + measurement representation of quantum circuits, and supports a quantum input encoding procedure and classical garbling procedure. That is, given an input $\ket{\psi}$, the sender can sample a classical key $k$ and run a quantum encoding procedure $\widetilde{\ket{\psi}} \gets \mathsf{Encode}_k(\ket{\psi})$. Then, there is a \emph{classical} procedure $\widetilde{Q} \gets \mathsf{Garble}_k(Q)$ that given a key $k$ and description of quantum circuit $Q$, outputs a classical garbled version of $Q$. Finally, it holds that $Q(\ket{\psi}) = \widetilde{Q}(\widetilde{\ket{\psi}})$, but no other information about $\ket{\psi}$ is leaked by $\widetilde{\ket{\psi}}$ and $\widetilde{Q}$.

Now, given a circuit $Q$ describing a unidirectional quantum functionality, the sender samples $k$ and sends $\widetilde{\ket{\psi}} \gets \mathsf{Encode}_k(\ket{\psi})$ to the receiver. In addition, it inputs $k$ to a unidirectional randomized classical functionality that will sample $\widetilde{Q}$ and along with a description of some measurement $M$ that the receiver can apply to $\widetilde{\ket{\psi}}$ to check that it is well-formed (these are the zero and $T$ checks described in \cite{BCKM}). Thus, the receiver obtains $\widetilde{\ket{\psi}}$ from the sender, and $M$ and $\widetilde{Q}$ from the classical functionality, checks that $\widetilde{\ket{\psi}}$ is well-formed using $M$, and finally evaluates $\widetilde{Q}(\widetilde{\ket{\psi}})$ to obtain their output.
\end{remark}

Thus, we immediately obtain the following corollary of \cref{thm:one-shot-security} and \cref{cor:classical-functionalities}.

\begin{corollary}\label{cor:quantum-functionalities}
Assuming the sub-exponential hardness of LWE, there exists a secure teleportation protocol through any quantum channel.
\end{corollary}


Now, we explore the following applications as special cases of \cref{cor:quantum-functionalities}.

\paragraph{NIZK for QMA.} First we recall the definition of the complexity class QMA.

\begin{definition}[QMA]
A language $\cL = (\cL_{\text{yes}}, \cL_{\text{no}})$  in QMA is defined by a tuple $(V,p,\alpha,\beta)$, where $p$ is a polynomial, $V = \{V_\secp\}_{\secp \in \mathbb{N}}$ is a uniformly generated family of circuits such that for every $\secp$, $V_\secp$ takes as input a string $x \in \{0,1\}^\secp$ and a quantum state $\ket{\psi}$ on $p(\secp)$ qubits and returns a single bit, and $\alpha,\beta : \mathbb{N} \to [0,1]$ are such that $\alpha(\secp) - \beta(\secp) \geq 1/p(\secp)$. The language is then defined as follows.

\begin{itemize}
    \item For all $x \in \cL_{\text{yes}}$ of length $\secp$, there exists a quantum state $\ket{\psi}$ of size at most $p(\secp)$ such that the probability that $V_\secp$ accepts $(x, \ket{\psi})$ is at least $\alpha(\secp)$. We denote the (possibly infinite) set of quantum witnesses that make $V_\secp$ accept $x$ by $\cR_\cL(x)$.
    \item For all $x \in \cL_{\text{no}}$ of length $\secp$, and all quantum states $\ket{\psi}$ of size at most $p(\secp)$, it holds that $V_\secp$ accepts on input $(x,\ket{\psi})$ with probability at most $\beta(\secp)$.
\end{itemize}
\end{definition}

Next, we define the notion of NIZK for QMA (with setup). In the following definition, we assume the completeness-soundness gap of the QMA language is 1-$\negl(\secp)$ (which can be obtained without loss of generality).

\begin{definition}[NIZK for QMA]
A NIZK for QMA for a language $\cL \in\QMA$ with relation $\cR_\cL$ consists of the following efficient algorithms.
\begin{itemize}
    \item $\NIZK.\Setup(1^\secp) \to (\cP,\cV)$: On input the security parameter $1^\secp$, the setup outputs a bipartite state on registers $\cP,\cV$.
    \item $\NIZK.\Prove(\cP, \ket{\psi}, x) \to \pi$: On input the state on register $\cP$, a witness $\ket{\psi}$, and a statement $x$, the proving algorithm outputs a proof $\pi$.
    \item $\NIZK.\Verify(\cV, \pi, x) \to \{\top,\bot\}$: On input the state on register $\cV$, a proof $\pi$, and a statement $x$, the verification algorithm returns $\top$ or $\bot$.
\end{itemize}
They should satisfy the following properties.
\begin{itemize}
	\item \textbf{Completeness.} For all $x\in\cL_{\mathsf{yes}}$, and all $\ket{\psi}\in\cR_\cL(x)$ it holds that
\[
\Pr\left[\NIZK.\Verify(\cV, \NIZK.\Prove(\cP, \ket{\psi}, x), x) \to \top\right] = 1 -\negl(\secp),
\] where $(\cP,\cV) \gets \NIZK.\Setup(1^\secp)$.

	\item \textbf{Soundness.} For any non-uniform QPT malicious prover $\{P^*_\secp,\ket{\psi_\secp}\}_{\secp \in \bbN}$ and $x \in \cL_{\mathsf{no}}$, it holds that \[\Pr\left[\NIZK.\Verify(\cV,P^*_\secp(\cP,\ket{\psi_\secp}),x)\to \top\right] = \negl(\secp).\]
	
	\item \textbf{Zero-knowledge.} There exists a simulator $\{S_\secp\}_{\secp \in \bbN}$ such that for any non-uniform QPT malicious verifier $\{V^*_\secp,\ket{\psi_\secp}\}_{\secp \in \bbN}$, $x \in \cL_{\mathsf{yes}}$, and $\ket{\psi} \in \cR_\cL(x)$, it holds that \[\bigg| \Pr\left[V^*_\secp(\cV,\pi,x) \to 1 : \begin{array}{r}(\cP,\cV) \gets \NIZK.\Setup(1^\secp) \\ \pi \gets \NIZK.\Prove(\cP,\ket{\psi},x)\end{array}\right] - \Pr\left[V^*_\secp(\cV,\pi,x) \to 1 : \begin{array}{r}(\cV,\pi) \gets S_\secp(x)\end{array}\right]\bigg| = \negl(\secp).\]

\end{itemize}
\end{definition}

\begin{corollary}
Assuming sub-exponential LWE, there exists a NIZK for QMA in the shared EPR model, where $\NIZK.\Setup(1^\secp)$ is a deterministic algorithm that outputs polynomially many EPR pairs.
\end{corollary}

\begin{proof}
This follows as a special case of our secure teleportation protocol, where the map $Q$ is the QMA verification circuit that outputs a single classical bit. Indeed, soundness follow directly from simulation security against a malicious sender, and zero-knowledge follows directly from simulation security against a malicious receiver.
\end{proof}

Prior work on NIZK for QMA includes the following.
\begin{itemize}
	\item Constructions in a \emph{pre-processing} model, where $\NIZK.\Setup(1^\secp)$ is randomized and samples correlated \emph{private randomness} for the prover and verifier \cite{ACGH,CVZ,BG,Shm21,BCKM,MY22}. Some of these can also be interpreted as two-message protocols in the common random string (CRS) model.
	\item A construction in the shared EPR pair + quantum random oracle model \cite{MY22}.
	\item A construction in the common reference string + classical oracle model \cite{BM22} (alternatively, from iO and non-black-box use of a hash function modeled as a random oracle). 
\end{itemize}

Thus, we achieve the first construction of NIZK for QMA in the shared EPR model from (sub-exponential) standard assumptions.



\paragraph{Zero-knowledge State Synthesis.} State synthesis \cite{https://doi.org/10.48550/arxiv.1607.05256,Rosenthal2022InteractivePF,10.4230/LIPIcs.CCC.2022.5} refers to the problem of generating a quantum state $\rho$ given an implicit classical description of it. That is, given the description of a quantum circuit $Q$ acting on $n$ qubits, generate the state 
$$\rho = \Tr_{\cB}\left(Q \ket{0^n}\right),$$ where $\rho$ consists of $m$ qubits and $\cB$ is a register of $n-m$ qubits.  We note that prior work has typically formalized this problem only for pure states ($n=m$), but the problem extends naturally to mixed states.

Thus far, state synthesis has been studied in the setting where a weak verifier is interacting with a more-powerful oracle or prover in order to generate a highly complex state $\rho$. However, \cite{Rosenthal2022InteractivePF} raised the question of whether there is some meaningful notion of \emph{zero-knowledge} state synthesis. In this setting, the notion would make sense even for efficiently preparable states. 

We propose one notion of zero-knowledge state synthesis. Rather than considering a single circuit $Q$, we consider a \emph{family} of circuits $\{Q_w\}_w$ parameterized by a witness $w$ that may be known to the prover but not necessarily to the verifier. The goal would then be to have the prover help the verifier generate
\[\rho_w = \Tr_{\cB}\left(Q_w\ket{0^n}\right)\] without leaking the description of $w$. A formal definition follows.

\begin{definition}[Zero-knowledge State Synthesis]\label{def:ZKSS}
Let $n(\secp),m(\secp),k(\secp)$ be polynomials such that $m(\secp) \leq n(\secp)$. A zero-knowledge state synthesis protocol for a family of polynomial-size quantum circuits $\{\{Q_{\secp,w}\}_{w \in \{0,1\}^{k(\secp)}}\}_{\secp \in \bbN}$ is an interactive protocol between a QPT interactive machine $P$ and a QPT interactive machine $V$. $P$ takes as input a classical string $w \in \{0,1\}^{k(\secp)}$ and has no output, while $V$ has no input but outputs either a state on $m(\secp)$ qubits or a special abort symbol $\dyad{\bot}{\bot}$.\footnote{We will assume that the abort symbol is orthogonal to this $m(\secp)$-qubit Hilbert space.} For any $w$, we write $\rho_w = \Tr_\cB\left(Q_w\ket{0^{n(\secp)}}\right),$ where register $\cB$ holds the final $n(\secp)-m(\secp)$ qubits. The protocol should satisfy the following properties.

\begin{itemize}
    \item \textbf{Completeness.} For any $w \in \{0,1\}^{k(\secp)}$, it holds that \[\TD\left(\rho_w,\langle P(w),V \rangle\right) = \negl(\secp),\] where $\langle P(w),V \rangle$ denotes the output of the honest verifier after interacting in the protocol with $P(w)$.
    \item \textbf{Soundness.} For any QPT malicious prover $\{P^*_\secp\}_{\secp \in \bbN}$, there exists a state $\rho^*$ in the convex combination of $\{\rho_w\}_{w \in \{0,1\}^{k(\secp)}} \cup \{\dyad{\bot}{\bot}\}$ such that for any QPT distinguisher $\{D_\secp\}_{\secp \in \bbN}$,\footnote{One could ask for a stronger definition of soundness, where the state output by the verifier must be \emph{statistically} rather than \emph{computationally} close to a distribution over honest output states (and abort). However, our protocol will only achieve this computational notion. }
    
    \[\big|\Pr[D_\secp(\rho^*) \to 1]  - \Pr[D_\secp( \langle P^*_\secp,V\rangle) \to 1]\big| = \negl(\secp),\]

    
    where $\langle P^*_\secp,V\rangle$ denotes the output of the honest verifier after interacting in the protocol with $P^*_\secp$.
    \item \textbf{Zero-Knowledge.} For any QPT malicious verifier $\{V^*_\secp\}_{\secp \in \bbN}$, there exists a QPT simulator $\{S_\secp\}_{\secp \in \bbN}$ such that for any QPT distinguisher $\{D_\secp\}_{\secp \in \bbN}$ and $w \in \{0,1\}^{k(\secp)}$, \[\big|\Pr[D_\secp(\langle P(w),V^*_\secp\rangle) \to 1] - \Pr[D_\secp(S_\secp(\rho_w)) \to 1]\big| = \negl(\secp),\] where $\langle P(w),V^*_\secp\rangle$ denotes the output of the malicious verifier $V^*_\secp$ after interacting in the protocol with $P(w)$, which may be a state on an arbitrary (polynomial-size) register.
\end{itemize}
\end{definition}

Now we perform some sanity checks on the definition. We also stress that there may be other natural definitions to consider, along with applications, and we defer a more thorough exploration of this topic to future work.

\begin{itemize}
	\item The trivial protocol where the prover prepares $\rho_w$ and sends it the verifier will not satisfy soundness, because it may not be efficient (or even possible) in general for the verifier to check that there exists $w$ such that its received state $\rho = \Tr_\cB(Q_w\ket{0^n})$.
	\item The trivial protocol where the prover sends $w$ and the verifier prepares $\rho_w$ will not satisfy zero-knowledge, because the simulator will not in general be able to recover $w$ from $\rho_w$.
	\item This notion generalizes zero-knowledge for NP. Each instance $x$ of an NP language defines a family of circuits $\{Q_w\}_w$ where $Q_w$ runs the verification algorithm on $x$ and $w$ and outputs a bit. Then, soundness says that for any no instance, the verifier outputs a state that is computationally indistinguishable from a mixture over $\dyad{0}{0}$ and $\dyad{\bot}{\bot}$, which implies that it must output either $\ket{0}$ (and reject) or $\ket{\bot}$ (and reject) with overwhelming probability. Zero-knowledge says that for any yes instance, and simulator can simulate the verifier's view just given $\dyad{1}{1}$, which merely indicates that the instance is a yes instance.
\end{itemize}



We note that while prior work \cite{BCKM} has implicitly constructed a two-message zero-knowledge state synthesis protocol in the CRS model, our work is the first to obtain this primitive with one message in the shared EPR pair model.

\begin{corollary}
Assuming sub-exponential LWE, there exists zero-knowledge state synthesis for any efficient family of quantum circuits.
\end{corollary}

\begin{proof}
This follows as a special case of our secure teleportation protocol, where the map $Q$ takes a classical input $w$ and outputs $\rho_w$. 
Soundness follow directly from simulation security against a malicious sender, where the state $\rho^*$ is defined by the probability distribution over classical strings $w$ (and $\mathsf{abort}$) sent by the simulator to the ideal functionality. Zero-knowledge follows directly from simulation security against a malicious receiver.
\end{proof}

\ifsubmission\else\section{Two-Round MPC}\label{sec:MPC}
In this section, we give a construction of a two-round MPC protocol for computing any classical function $f$ in the shared EPR pairs model without making use of public-key cryptography.

\subsection{Two-round OT in the shared EPR pairs model}

In this section, we present syntax and definitions for (chosen-input) two-round oblivious transfer where the parties begin with shared EPR pairs. We assume that all communication is classical (which is without loss of generality, due to teleportation), and we additionally require that all \emph{computation} is classical after initial (input-independent) measurements of the shared EPR pairs. The syntax follows.

\begin{itemize}
    \item $(\regR,\regS) \gets \Setup(1^\secp)$: The setup algorithm prepares $\poly(\secp)$ EPR pairs with halves on register $\regR$ and other halves on register $\regS$.
    \item $\sigma_R \gets M_R(\regR)$: The receiver performs a measurement $M_R$ on its halves of EPR pairs, resulting in a classical string $\sigma_R$.
    \item $\sigma_S \gets M_S(\regS)$: The sender performs a measurement $M_S$ on its halves of EPR pairs, resulting in a classical string $\sigma_S$.
    \item $(\ots_1,\omega) \gets \OT_1(\sigma_R,b)$: The receiver, on input a string $\sigma_R$ and bit $b \in \{0,1\}$, applies the first OT algorithm to obtain a message $\ots_1$ and secret $\omega$.
    \item $\{\ots_2,\bot\} \gets \OT_2(\sigma_S,\ots_1,m_0,m_1)$: The sender, on input a string $\sigma_S$, message $\ots_1$, and strings $m_0,m_1$, applies the second OT algorithm to obtain either a message $\ots_2$ or an abort symbol $\bot$.
    \item $\{m,\bot\} \gets \OT_3(\ots_2,b,\omega)$: The receiver, on input a message $\ots_2$, bit $b$, and secret $\omega$, outputs either a message $m$ or an abort symbol $\bot$.
\end{itemize}

We now define what it means for such a protocol to be \emph{black-box friendly}.

\begin{definition}
A two-round OT protocol in the shared EPR pair model is \emph{black-box friendly} if:
\begin{itemize}
    \item The receiver's algorithm $\OT_1$ can be split into two parts $\OT_1^\C,\OT_1^\NC$, where $(\ots_1^\C,\omega) \gets \OT_1^\C(\sigma_R)$, $\ots_1^\NC \gets \OT_1^\NC(b,\omega)$, and $\ots_1 \coloneqq (\ots_1^\C, \ots_1^\NC)$.
    \item The sender's algorithm $\OT_2$ can be split into two parts $\OT_2^\C,\OT_2^\NC$, where $\{\top,\bot\} \gets \OT_2^\C(\sigma_S,\ots_1^\C)$, $\ots_2 \gets \OT_2^\NC(\sigma_S,\ots_1,m_0,m_1)$, and the output of $\OT_2(\sigma_S,\ots_1,m_0,m_1)$ is set to $\bot$ if the output of $\OT_2^\C$ was $\bot$ and is otherwise set to the output of $\OT_2^\NC$.
    \item The only algorithms in the entire protocol that involve cryptographic operations are $\OT_1^\C$ and $\OT_2^\C$.
\end{itemize}
\end{definition}

Next, we define the security properties that we will need.

\begin{definition}\label{def:two-round-sim-security}
A two-round OT protocol in the shared EPR pair model satisfies \emph{standard simulation-based security} if it satisfies \cref{def:secure-realization} for functionality $\cF_\OT$. 
\end{definition}

\begin{definition}\label{def:two-round-equiv-security}
A two-round OT protocol in the shared EPR pair model satisfies \emph{equivocal receiver's security} if there exists a QPT simulator $\Sim_\EQ(1^\secp)$ such that for any $b \in \{0,1\}$,
\[\left\{(\regS,\ots_1,\omega_b) : (\regS,\ots_1,\omega_0,\omega_1) \gets \Sim_\EQ(1^\secp)\right\} \approx_c \left\{(\regS,\ots_1,\omega) : \begin{array}{r}(\regR,\regS) \gets \Setup(1^\secp), \\ \sigma_R \gets M_R(\regR),\\(\ots_1,\omega) \gets \OT_1(\sigma_R,b)\end{array}\right\}.\]
\end{definition}

\begin{theorem}\label{thm:black-box-friendly}
Given a black-box friendly two-round OT protocol in the shared EPR pair model that satisfies standard simulation-based security and equivocal receiver's security, there exists a black-box construction of two-round MPC for classical functionalities in the shared EPR pair model.
\end{theorem}

\begin{proof}
This follows from a few tweaks to the construction from \cite{EC:GarSri18a}, and some passages of this proof will be taken verbatim from \cite{EC:GarSri18a}.  At a high level, the \cite{EC:GarSri18a} construction proceeds in two steps.
\begin{enumerate}
    \item Write any multi-round secure MPC protocol as a \emph{conforming protocol}, where in each round a single party broadcasts a single bit.
    \item Compile any conforming protocol into a two-round protocol using garbled circuits and two-round OT with equivocal receiver's security.
\end{enumerate}

To implement Step 1, we will start with any multi-round MPC protocol in the OT-hybrid model \cite{STOC:Kilian88,C:CreVanTap95,C:IshPraSah08} and replace each call to the OT ideal functionality with an implementation of our black-box friendly two-round OT. However, we want to make sure that all cryptographic operations happen ``outside'' of the part of the protocol that is compiled during Step 2 in a non-black-box manner. Thus, we require each party to run their initial measurements $M_R,M_S$ and cryptographic algorithms $\OT_1^\C,\OT_2^\C$ for each OT that they will participate in at the \emph{beginning} of the protocol. Note that this is possible because each of these algorithms are required to be \emph{input-independent}. The result will be a conforming protocol with the following structure.

\begin{definition}
A \emph{black-box friendly conforming protocol in the shared EPR pair model} is an MPC protocol $\Phi$ between parties $P_1,\dots,P_n$ with inputs $x_1,\dots,x_n$ respectively. For each $i \in [n]$, we let $x_i \in \{0,1\}^m$ denote the input of party $P_i$. The protocol is defined by functions $\pre,\chck,\post$, and computation steps or what we call \emph{actions} $\phi_1,\dots,\phi_T$. Each pair of parties begins with a polynomial number of shared EPR pairs, and we let $\cP_i$ denote party $P_i$'s register that contains its halves of EPR pairs with every other party. Then, the protocol proceeds in three stages: the pre-processing phase, the computation phase, and the output phase.
\begin{itemize}
    \item \textbf{Pre-processing phase:} For each $i \in [n]$, party $P_i$ computes \[(z_i,v_i,w_i) \gets \pre(1^\secp,i,x_i,\cP_i),\] where $\pre$ is a quantum operation that takes as input the security parameter $1^\secp$, party identifier $i$, input $x_i$, and halves of EPR pairs on register $\cP_i$, and outputs $z_i \in \{0,1\}^{\ell/n}$, $v_i \in \{0,1\}^\ell$ (where $\ell$ is a parameter of the protocol), and $w_i \in \{0,1\}^*$ . Next, $P_i$ broadcasts $z_i$ to every other party, while retaining $v_i$ and $w_i$ as secret information. We require that all $v_{i,j} = 0$ for all $j \in [\ell] \setminus \{(i-1)\ell/n + 1,\dots,i\ell/n\}$. Finally, for each $i \in [n]$, party $P_i$ computes \[\{\top,\bot\} \gets \chck(i,w_i,z_1,\dots,z_n),\] and aborts the protocol in the case of $\bot$.
    \item \textbf{Computation phase:} For each $i \in [n]$, party $P_i$ sets \[\state_i \coloneqq (z_1, \dots, z_n) \oplus v_i.\] Next, for each $t \in \{1,\dots,T\}$ parties proceed as follows:
    \begin{enumerate}
        \item Parse action $\phi_t$ as $(i,f,g,h)$ where $i \in [n]$ and $f,g,h \in [\ell]$.
        \item Party $P_i$ computes \emph{one} $\NAND$ gate as 
        \[\state_{i,h} = \NAND(\state_{i,f},\state_{i,g})\] and broadcasts $\state_{i,h} \oplus v_{i,h}$ to every other party.
        \item Every party $P_j$ for $j \neq i$ updates $\state_{j,h}$ to the bit value received from $P_i$.
    \end{enumerate}
    We require that for all $t,t' \in [T]$ such that $t \neq t'$, we have that if $\phi = (\cdot,\cdot,\cdot,h)$ and $\phi_{t'} = (\cdot,\cdot,\cdot,h')$, then $h \neq h'$. Also we denote $A_i \subset [T]$ to be the set of rounds in which party $P_i$ sends a bit. Namely, $A_i = \{t \in T : \phi_t = (i,\cdot,\cdot,\cdot)\}$.
    \item \textbf{Output phase:} For each $i \in [n]$, party $P_i$ outputs $\post(\state_i)$.
\end{itemize}
Finally, we require that the only algorithms in the entire protocol that involve cryptographic operations are $\pre$ and $\chck$.
\end{definition}

Now, the following two claims complete the proof of the theorem.

\begin{claim}
Assuming a black-box friendly two-round OT protocol in the shared EPR pair model that satisfies simulation-based security, there exists a black-box friendly conforming protocol in the shared EPR pair model.
\end{claim}

\begin{proof}
We start with a multi-round MPC protocol in the OT-hybrid model \cite{STOC:Kilian88,C:CreVanTap95} and replace each call to the OT ideal functionality with an implementation of a black-box friendly two-round OT in the shared EPR pair model. We assume without loss of generality that each pair of parties $(P_i,P_j)$ participates in $u$ many OT protocols with $P_i$ as the receiver and $P_j$ as the sender. We split the resulting protocol $\Pi$ into three stages, where all pairs of parties begin with sufficiently many shared EPR pairs.
\begin{enumerate}
    \item In the first stage, each party $P_i$ does the following. For each $j \neq i$, $P_i$ measures $M_R$ on $u$ disjoint sets of halves of EPR pairs that they share with party $P_j$ to obtain $\{\sigma_{R,i,j,o}\}_{o \in [u]}$, measures $M_S$ on $u$ other disjoint sets of halves of EPR pairs that they share with party $P_j$ to obtain  $\{\sigma_{S,i,j,o}\}_{o \in [u]}$, and computes $\{(\ots^\C_{1,i,j,o},\omega_{i,j,o}) \gets \OT_1^\C(\sigma_{R,i,j,o})\}_{o \in [u]}$. Then, party $P_i$ broadcasts $y_i \coloneqq \{\ots^\C_{1,i,j,o}\}_{j \neq i, o \in [u]}$. We assume without loss of generality that there exists $k$ such that for each $i$, $y_i \in \{0,1\}^k$.
    
    \item In the second stage, each party $P_i$ runs $\OT_2^\C(\sigma_{S,i,j,o},\ots_{1,i,j,o}^\C)$ for each $j \neq i, o \in [u]$ and outputs $\bot$ if any are $\bot$ (and otherwise continues). 
    \item In the third stage, each party $P_i$ defines \[x_i' \coloneqq \left(x_i, \{\omega_{i,j,o}\}_{j \neq i, o \in [u]}, \{\sigma_{S,j,i,o}\}_{j \neq i, o \in [u]}, \$_i \right),\] where $\$_i$ are any additional random coins that $P_i$ needs for computing the MPC protocol. Then, the parties engage in a deterministic MPC protocol with private inputs $x'_1,\dots,x'_n$. Note that each OT $o \in [u]$ between receiver $P_i$ and sender $P_j$ can be completed by 
    \begin{itemize}
        \item computing $\ots_{1,i,j,o}^\NC \gets \OT_1^\NC(b,\omega_{i,j,0})$,
        \item setting $\ots_{1,i,j,o} \coloneqq (\ots_{1,i,j,o}^\C, \ots_{1,i,j,o}^\NC)$,
        \item computing $\ots_{2,i,j,o} \gets \OT_2^\NC(\sigma_{S,i,j,o},\ots_{1,i,j,o},m_0,m_1)$,
        \item and computing $\OT_3(\ots_{2,i,j,o},b)$.
    \end{itemize}
    We assume without loss of generality that (i) $|x'_1| = \dots = |x'_n| = m'$, (ii) in each round $r \in [p]$, a single party broadcasts a single bit computed using circuit $C_r$, (iii) there exists $q$ such that for each $r \in [p]$, $|C_r| = q$, (iv) each $C_r$ is composed on just $\NAND$ gates with fan-in 2, and (v) each party sends an equal number of bits in the execution of $\Pi$. Note that there are no crytographic operations in this third stage.
\end{enumerate}

Now we are ready to describe the transformed conforming protocol $\Phi$. The protocol $\Phi$ will include $T = pq$ actions, and we let $\ell = (m'+k)n + pq$ and $\ell' = pq/n$.

\begin{itemize}
    \item $\pre(1^\secp,i,x_i,\cP_i)$: Run party $P_i$'s first stage algorithm to obtain strings $x'_i \in \{0,1\}^{m'}$ and $y_i \in \{0,1\}^k$. Sample $r_i \gets \{0,1\}^{m'}$ and $s_i \gets (\{0,1\}^{q-1} \| 0)^{p/n}$. (Observe that $s_i$ is an $\ell'$ bit random string such that its $q^{th}$, $2q^{th}$,$\dots$ locations are set to 0). Output \[z_i \coloneqq (x'_i \oplus r_i, y_i, 0^{\ell'}), ~~ v_i \coloneqq (0^{\ell/n},\dots,r_i,0^k,s_i,\dots,0^{\ell/n}), ~~ w_i \coloneqq \{\sigma_{S,i,j,o}\}_{j \neq i, o \in [u]}.\]
    \item $\chck(i,w_i,z_1,\dots,z_n)$: For each $j \neq i$, obtain $y_i$ from $z_i$ and use $y_i$ and $w_i$ to run $P_i$'s second stage algorithm.
    \item We are now ready to desribe the actions $\phi_1,\dots,\phi_T$. For each $r \in [p]$, round $r$ in $\Pi$ is expanded into $q$ actions $\{\phi_j\}_{j \in \{(r-1)q + 1,\dots,rq\}}$. Let $P_i$ be the party that computes the circuit $C_r$ and broadcasts the output bit in round $r$. For each $j$, we set $\phi_j = (i,f,g,h)$ where $f$ and $g$ are the locations in $\state_i$ that the $j^{th}$ gate of $C_r$ is computed on (recall that initially $\state_i$ is set to $(z_1,\dots,z_n) \oplus v_i$). Moreover, we set $h$ to be the first location in $\state_i$ among the locations $(i-1)\ell/n + m + k + 1$ to $i\ell/n$ that has previously not been assigned to an action. Recall from before that on the execution of $\phi_j$, party $P_i$ sets $\state_{i,h} \coloneqq \NAND(\state_{i,f},\state_{i,g})$ and broadcasts $\state_{i,h} \oplus v_{i,h}$ to all parties.
    \item $\post(i,\state_i)$: Gather the local state of $P_i$ and the messages sent by the other parties in $\Pi$ from $\state_i$ and output the output of $\Pi$.
\end{itemize}

Now we need to argue that $\Phi$ preserves the correctness and security properties of $\Pi$. Observe that $\Phi$ is essentially the same protocol as $\Pi$ except that in $\Phi$, some additional bits are sent. Specifically, in addition to the messages that were sent in $\Pi$, in $\Phi$ parties send $z_i$ and $q-1$ additional bits per every bit sent in $\Pi$. Note that these additional bits sent are not used in the computation of $\Phi$. Thus these bits do not affect the functionality of $\Pi$ if dropped. This ensures that $\Phi$ inherits the correctness properties of $\Pi$. Next note that each of these extra bits is masked by a uniformly random and independent bit. This ensures that $\Phi$ achieves the same security properties as $\Pi$. Also note that by construction, for all $t,t' \in [T]$ such that $t \neq t'$, we have that if $\phi_t = (\cdot,\cdot,\cdot,h)$ and $\phi_{t'} = (\cdot,\cdot,\cdot,h')$ then $h \neq h'$, as required. Finally, note that the only algorithms that involve cryptographic operations are $\pre$ and $\chck$, as required.

\end{proof}

\begin{claim}
Assuming a black-box friendly conforming protocol in the shared EPR pair model and a two-round OT protocol in the shared EPR pair model that satisfies simulation-based security and equivocal receiver's security, there exists a black-box two-round MPC for classical functionalities in the shared EPR pair model.
\end{claim}

\begin{proof}
This follows from the compiler presented in \cite[Section 6]{EC:GarSri18a} that takes a conforming protocol and a two-round OT protocol and produces a two-round MPC protocol, by addressing the following slight differenes.

\begin{enumerate}
    \item We start with a conforming protocol where $\pre$ additionally operates on some shared quantum registers, and there is an extra algorithm $\chck$ in the pre-processing phases.
    \item We start with a conforming protocol where there are no cryptographic operations in the computation phase.
    \item We use a two-round OT protocol in the shared EPR model rather than the CRS model.
    \item We prove security against \emph{quantum} adversaries rather than classical adversaries.
\end{enumerate}

First, we note that (1) makes no difference in their construction or proof, since the entire pre-processing phase in the compiled protocol is run exactly as it is run in the conforming protocol. That is, it is run ``outside'' of the garbled protocol used to round-collapse the $T$ actions of the computation phase. Next, because of condition (2) we obtain a resulting two-round MPC protocol that is \emph{black-box}, since only the computation phase is used in a non-black-box way by the compiler. We also note that (3) makes no difference in the construction or proof, since the two-round OT is used in a black-box manner. Finally, all of the simulators and reductions in \cite{EC:GarSri18a} are straight-line black-box and do not use rewinding (which is typical in the CRS model), and as such they carry over immediately to the quantum setting.
\end{proof}

\end{proof}

\subsection{Construction}

In this section, we revisit an OT protocol from \cite{ABKK} and show that it is a black-box friendly two-round OT in the shared EPR pair model. 

\paragraph{Ingredients}
\begin{itemize}
    \item Non-interactive extractable commitment $(\Com,\ExtGen,\Ext)$ in the common random string model (\cref{subsec:commitments}).
    \item A programmable hash function family $\{H_\secp\}_{\secp \in \bbN}$ that is correlation intractable for efficiently verifiable approximate product relations with constant sparsity (\cref{subsec:CI}).
    \item A universal hash function family $\{F_\secp\}_{\secp \in \bbN}$.
\end{itemize}

\paragraph{Parameters}
\begin{itemize}
    \item Security parameter $\secp$.
    \item Size of commitment key $h = h(\secp)$.
    \item Size of correlation intractable hash key $k_1 = k_1(\secp)$.
    \item Size of universal hash key $k_2 = k_2(\secp)$.
    \item Approximation parameter $\alpha = 1/120$.
    \item Number of repetitions in each group $c = 480$.
    \item Sparsity $\rho = \frac{\binom{(1-\alpha)c}{(1/2)c}}{2^c} < \alpha$.
    \item Product parameter $t = t(\secp) = 180^3 \secp \geq \secp / (\alpha-\rho)^3$.
    \item Total number of repetitions $\ell = \ell(\secp) = c \cdot t = O(\secp)$.
    \item CI hash range $\cY^t$, where $\cY$ is the set of subsets of $[c]$ of size $c/2$. We will also parse $T \in \cY^t$ as a subset of $[\ell]$ of size $\ell/2$.
\end{itemize}

The protocol is given in \cref{fig:two-round}. We remark that shared uniformly random strings can be obtained by measuring shared EPR pairs in the same basis, and thus the entire Setup can be obtained with just shared EPR pairs.

\protocol{Two-round OT in the shared EPR pair model}{A black-box friendly two-round OT protocol that can be used to construct two-round MPC in the shared EPR pair model.}{fig:two-round}{
\paragraph{Setup}
\begin{itemize}
    \item $2\ell$ EPR pairs on registers $\{(\regR_{i,b},\regS_{i,b})\}_{i \in [\ell], b \in \{0,1\}}$. Let $\regR \coloneqq \{\regR_{i,b}\}_{i \in [\ell],b \in \{0,1\}}$ and $\regS \coloneqq \{\regS_{i,b}\}_{i \in [\ell],b \in \{0,1\}}$.
    \item Commitment key $\ck \gets \{0,1\}^h$.
    \item Correlation intractable hash key $\hk \gets \{0,1\}^{k_1}$.
\end{itemize}

\paragraph{Protocol}

\begin{itemize}
    \item $M_R(\regR)$: Sample $\theta^R \gets \{0,1\}^\ell$ and for each $i \in [\ell]$, measure registers $\regR_{i,0}$ and $\regR_{i,1}$ in the $\theta^R_i$ basis (where 0 indicates standard basis and 1 indicates Hadamard basis) to obtain bits $v^R_{i,0},v^R_{i,1}$. Output $\sigma_R \coloneqq \{\theta_i^R,v_{i,0}^R,v_{i,1}^R\}_{i \in [\ell]}$.
    \item $M_S(\regS)$: Sample $U$ as a uniformly random subset of $[\ell]$. For $i \in [\ell]$:
    \begin{itemize}
        \item If $i \in U$ sample $\theta^S_i \gets \{0,1\}$ and measure registers $\regS_{i,0}$ and $\regS_{i,1}$ in the $\theta^S_i$ basis to obtain bits $v^S_{i,0},v^S_{i,1}$.
        \item If $i \notin U$, measure register $\regS_{i,0}$ in the standard basis and $\regS_{i,1}$ in the Hadamard basis to obtain bits $v^S_{i,0},v^S_{i,1}$.
    \end{itemize}
    Output $\sigma_S \coloneqq (U,\{\theta_i^S\}_{i \in U},\{v^S_{i,0},v^S_{i,1}\}_{i \in [\ell]})$.
    \item $\OT_1^\C(\sigma_R)$: For each $i \in [\ell]$, compute $\cm_i \coloneqq \Com(\ck,(\theta^R_i,v^R_{i,0},v^R_{i,1});r_i)$, where $r_i \gets \{0,1\}^\secp$ are the random coins used for commitment. Compute $T = H_\secp(\hk,(\cm_1,\dots,\cm_\ell))$. Output $\ots_1^\C \coloneqq (\{\cm_i\}_{i \in [\ell]},T,\{\theta^R_i,v^R_{i,0},v^R_{i,1},r_i\}_{i \in T})$ and $\omega \coloneqq \{\theta^R_i,v^R_{i,\theta^R_i}\}_{i \in \overline{T}}$.
    \item $\OT_1^\NC(b,\omega)$: For $i \in \overline{T}$, compute $d_i \coloneqq b \oplus \theta^R_i$, and output $\ots_1^\NC \coloneqq \{d_i\}_{i \in \overline{T}}$.
    \item $\OT_2^\C(\sigma_S,\ots_1^\C)$: Output $\top$ if:
    \begin{itemize}
        \item $T = H_\secp(\hk,(\cm_1,\dots,\cm_\ell))$,
        \item for all $i \in T$, $\cm_i = \Com(\ck,(\theta^R_i,v^R_{i,0},v^R_{i,1});r_i)$ and,
        \item for all $i \in T \intersect U$ such that $\theta^S_i = \theta^R_i$, $(v^S_{i,0},v^S_{i,1})  = (v^R_{i,0},v^R_{i,1})$,
    \end{itemize}
    and otherwise output $\bot$.
    \item $\OT_2^\NC(\sigma_S,\ots_1,m_0,m_1)$: Sample $s \gets \{0,1\}^{k_2}$, let $V_0$ be the concatenation of bits $\{v^S_{i,d_i}\}_{i \in \overline{T} \setminus U}$, let $V_1$ be the concatenation of bits $\{v^S_{i,d_i \oplus 1}\}_{i \in \overline{T} \setminus U}$, compute $\widetilde{m}_0 \coloneqq m_0 \oplus F_\secp(s,V_0)$, $\widetilde{m}_1 \coloneqq m_1 \oplus F_\secp(s,V_1)$, and output $\ots_2 \coloneqq (s,U,\widetilde{m}_0,\widetilde{m}_1)$.
    \item $\OT_3(\ots_2,b,\omega)$: Output $m_b \coloneqq \widetilde{m}_b \oplus F_\secp(s,V)$, where $V$ is the concatenation of bits $\{v^R_{i,\theta^R_i}\}_{i \in \overline{T} \setminus U}$.
\end{itemize}
}

\subsection{Security}

\begin{importedtheorem}[\cite{ABKK}]
If $\Com$ is replaced by an extractable commitment in the QROM (\cite[Section 5]{ABKK}) and $H_\secp$ is modeled as quantum random oracle, the protocol in \cref{fig:two-round} satisfies standard simulation-based security and equivocal receiver's security in the QROM.
\end{importedtheorem}

\begin{proof}
A three-round version of this protocol was shown to be simulation-secure in the QROM in \cite[Section 7]{ABKK}. In the three-round version, the first round consists of random BB84 states sampled by the sender according to some distribution over basis choices. This is equivalent to imagining that the sender and receiver share EPR pairs and having the sender measure their halves using the bases sampled according to their distribution. This is exactly what is accomplished by the measurement $M_S$ in \cref{fig:two-round}, and thus the simulation-security of the two-round protocol in the shared EPR pair model has an identical proof. Equivocal sender's security follows by defining $\Sim_\EQ$ in the same manner as in the proof of \cref{thm:two-round} below.
\end{proof}

\begin{theorem}\label{thm:two-round}
The protocol in \cref{fig:two-round} satisfies standard simulation-based (\cref{def:two-round-sim-security}) security and equivocal receiver's security (\cref{def:two-round-equiv-security}).
\end{theorem}

By \cref{thm:black-box-friendly}, we immediately have the following corollaries.

\begin{corollary}
There exists a two-round MPC protocol in the shared EPR pair model that is unconditionally-secure in the quantum random oracle model. There exists a black-box two-round MPC protocol in the shared EPR pair model assuming non-interactive extratable commitments and correlation intractability for efficiently searchable relations.
\end{corollary}

\begin{proof}(Of \cref{thm:two-round})
First, we establish standard simulation-based security via the following two lemmas.

\begin{lemma}
The protocol in \cref{fig:two-round} satisfies simulation-based security against a malicious receiver.
\end{lemma}

\begin{proof}

Let $\{\Adv_\secp\}_{\secp \in \bbN}$ be a QPT adversary corrupting the receiver, which takes as input register $\regA$ of $\{\ket{\psi_\secp}^{\regA,\regD}\}_{\secp \in \bbN}$. Instead of explicitly considering a register $\regX$ holding the honest sender's input, we write their input as classical strings $(m_0,m_1) \in \{0,1\}^\secp$. We define a simulator $\{\Sim_\secp\}_{\secp \in \bbN}$ as follows.\\

\noindent\underline{$\Sim_\secp(\regA)$}
\begin{itemize}
    \item Prepare $2\ell$ EPR pairs on registers $(\regR,\regS)$, sample $(\ck,\ek) \gets \ExtGen(1^\secp)$, and sample $\hk \gets \{0,1\}^{k_1}$.
    \item Run $\sigma_S \gets M_S(\regS)$ as the honest sender.
    \item Run $(\ots_1,\regA') \gets \Adv_\secp(\regA,\regR,\ck,\hk)$.
    \item Run $\OT_2^\C(\sigma_S,\ots_1^\C)$ and if the result is $\bot$ then send $\abort$ to $\cF_\OT$ and output $\regA'$.
    \item Parse $\ots_1$ as $(\{\cm_i\}_{i \in [\ell]},T,\{\theta_i^R,v_{i,0}^R,v_{i,1}^R,r_i\}_{i \in T},\{d_i\}_{i \in \overline{T}})$. For each $i \in \overline{T} \setminus U$, compute $(\theta_i^R,v_{i,0}^R,v_{i,1}^R) \gets \Ext(\ek,\cm_i)$. Then, set $b \coloneqq \maj\{\theta_i^R \oplus d_i\}_{i \in \overline{T} \setminus U}$, query $\cF_\OT$ with $b$ to obtain $m_b$, and set $m_{1-b} \coloneqq 0^\secp$.
    \item Run $\ots_2 \gets \OT_2^\NC(\sigma_S,\ots_1,m_0,m_1)$, run $\regA'' \gets \Adv_\secp(\ots_2,\regA')$, and output $\regA''$.
\end{itemize}

\noindent Now, we define a sequence of hybrids.

\begin{itemize}
    \item $\cH_0$: This is the real distribution $\Pi_{\cF_\OT}[\Adv_\secp,\ket{\psi_\secp}]$.
    \item $\cH_1$: Same as $\cH_0$ except that $\ck$ is sampled as $(\ck,\ek) \gets \ExtGen(1^\secp)$.
    \item $\cH_2$: Same as $\cH_1$ except that we delay the measurement of registers $\{\regS_{i,b}\}_{i \in [\ell],b \in \{0,1\}}$ until the results of measurement are required for the protocol. That is, $\{\regS_{i,0},\regS_{i,1}\}_{i \in T \cap U}$ are measured during $\OT_2^\C$ and $\{\regS_{i,0},\regS_{i,1}\}_{i \in \overline{T}\setminus U}$ are measured during $\OT_2^\NC$ (and the rest of the registers are left unmeasured).
    \item $\cH_3$: Same as $\cH_2$ except that we insert a measurement on the registers $\{\regS_{i,0},\regS_{i,1}\}_{i \in \overline{T}\setminus U}$ that is performed during $\OT_2^\NC$ right before $\{\regS_{i,0}\}_{i \in \overline{T} \setminus U}$ are measured in the standard basis and $\{\regS_{i,0}\}_{i \in \overline{T} \setminus U}$ are measured in the Hadamard basis to obtain $\{v_{i,0}^S,v_{i,1}^S\}_{i \in \overline{T} \setminus U}$. Before specifying this measurement, we introduce some notation.
    \begin{itemize}
        \item For any $i$, let $\regS_i \coloneqq (\regS_{i,0},\regS_{i,1})$, and for $\{\theta_i,v_{i,0},v_{i,1}\}_{i \in \overline{T} \setminus U}$ and $e \in \{0,1\}^{|\overline{T} \setminus U|}$, define
        \begin{align*}&\Pi[e,\{\theta_i,v_{i,0},v_{i,1}\}_{i \in \overline{T} \setminus U}]^{\{\regS_i\}_{i \in \overline{T} \setminus U}}\\ &\coloneqq \bigotimes_{i: e_i = 0}H^{\theta_i}\dyad{v_{i,0},v_{i,1}}{v_{i,0},v_{i,1}}H^{\theta_i} \otimes \bigotimes_{i : e_i =1}\bbI - H^{\theta_i}\dyad{v_{i,0},v_{i,1}}{v_{i,0},v_{i,1}}H^{\theta_i}.\end{align*}
        \item For $\{\theta_i,v_{i,0},v_{i,1}\}_{i \in \overline{T} \setminus U}$ and a constant $\gamma \in [0,1]$, define 
        \[\Pi[\gamma,\{\theta_i,v_{i,0},v_{i,1}\}_{i \in \overline{T} \setminus U}]^{\{\regS_i\}_{i \in \overline{T} \setminus U}} \coloneqq \sum_{e: \{0,1\}^{|\overline{T} \setminus U|} : \hw(e) < \gamma |\overline{T} \setminus U|} \Pi[e,\{\theta_i,v_{i,0},v_{i,1}\}_{i \in \overline{T} \setminus U}]^{\{\regS_i\}_{i \in \overline{T} \setminus U}}.\]
    \end{itemize}
    
    Now, before the $\{\regS_i\}_{i \in \overline{T}\setminus U}$ are measured according to the protocol, compute $\{(\theta_i^R,v_{i,0}^R,v_{i,1}^R) \gets \Ext(\ek,\cm_i)\}_{i \in \overline{T} \setminus U}$ and attempt to project $\{\regS_i\}_{i \in \overline{T}\setminus U}$ onto
    
    \[\Pi[1/6,\{\theta_i^R,v_{i,0}^R,v_{i,1}^R\}_{i \in \overline{T} \setminus U}],\] and abort and output $\bot$ if this projection fails.
    \item $\cH_4$: Same as $\cH_3$ except that $m_{1-b}$ is set to $0^\secp$, where $b \coloneqq \maj\{\theta_i^R \oplus d_i\}_{i \in \overline{T} \setminus U}$.
    \item $\cH_5$: Same as $\cH_4$ except that the measurement introduced in $\cH_3$ is removed, and the measurement of registers $\{\regS_{i,b}\}_{i \in [\ell],b \in \{0,1\}}$ are performed during $M_S$ as in the honest protocol. For any honest sender inputs $m_0,m_1 \in \{0,1\}$, this is the simulated distribution $\widetilde{\Pi}_{\cF_\OT}[\Sim_\secp,\ket{\psi_\secp}]$.
\end{itemize}

The following sequence of claims completes the proof of the lemma.

\begin{claim}
$\cH_0 \approx_c \cH_1$
\end{claim}

\begin{proof}
This follows directly from the extractability of the commitment scheme (\cref{def:extractability}).
\end{proof}

\begin{claim}
$\cH_1 \equiv \cH_2$
\end{claim}

\begin{proof}
Delaying the honest party's measurements has no effect on the distribution seen by the adversary.
\end{proof}

\begin{claim}
$\cH_2 \approx_s \cH_3$
\end{claim}

\begin{proof}
By Gentle Measurement, it suffices to show that the projection introduced in $\cH_2$ succeeds with probability $1-\negl(\secp)$. So, towards contradiction, assume that the projection fails with non-negligible probability. We will eventually use this assumption to break the correlation intractability of $H$. First, consider the following experiment.\\

\noindent\underline{$\Exp_1$}

\begin{itemize}
    \item Prepare $2\ell$ EPR pairs on registers $(\regR,\regS)$, sample $(\ck,\ek) \gets \ExtGen(1^\secp)$, and sample $\hk \gets \{0,1\}^{k_1}$.
    \item Sample $U$ as a uniformly random subset of $[\ell]$ and sample $\theta_i^S \gets \{0,1\}$ for each $i \in [\ell]$.
    \item Run $(\ots_1,\regA') \gets \Adv_\secp(\regA,\regR,\ck,\hk)$, and let $\{\cm_i\}_{i \in [\ell]}$ be the commitments that are part of $\ots_1$.
    \item For each $i \in [\ell]$, compute $(\theta_i^R,v_{i,0}^R,v_{i,1}^R) \gets \Ext(\ek,\cm_i)$. Say that index $i$ is ``consistent'' if $i \in U$ and $\theta_i^S = \theta_i^R$.
    \item For each $i \in [\ell]$, measure $(\regS_{i,0},\regS_{i,1})$ in basis $\theta_i^R$ to obtain bits $(v_{i,0}^S,v_{i,1}^S)$. Say that index $i$ is ``correct'' if $(v_{i,0}^S,v_{i,1}^S) = (v_{i,0}^R,v_{i,1}^R)$.
    \item Output 1 if (i) all consistent $i \in T$ are correct, and (ii) at least $1/6$ fraction of $i \in \overline{T} \cap U$ are incorrect. 
\end{itemize}

We claim that $\Pr[\Exp_1 \to 1] = \nonnegl(\secp)$. This nearly follows by assumption that the measurement introduced in $\cH_3$ rejects with non-negligible probability, except for the following differences. One difference from $\cH_3$ is that to determine condition (i) in $\Exp_1$, we are using extracted $\{\theta_i^R\}_{i \in T \cap U}$ bases rather than the bases opened by the adversary. However, this introduces a negligible difference due to the extractability of the commitment scheme. Another difference is that to determine condition (ii), we are measuring registers $i \in \overline{T} \cap U$ rather registers $i \in \overline{T} \setminus U$. However, since each $i$ is chosen to be in $U$ with probability 1/2 indepedent of the adversary's view, this produces the same distribution. Next, consider the following procedure.\\

\noindent \underline{$\Exp_2$}

\begin{itemize}
    \item Prepare $2\ell$ EPR pairs on registers $(\regR,\regS)$, sample $(\ck,\ek) \gets \ExtGen(1^\secp)$, and sample $\hk \gets \{0,1\}^{k_1}$.
    \item Sample $U$ as a uniformly random subset of $[\ell]$ and sample $\theta_i^S \gets \{0,1\}$ for each $i \in [\ell]$.
    \item For $i \in [\ell]$, measure $(\regS_{i,0},\regS_{i,1})$ in the $\theta_i^S$ basis to obtain bits $(v_{i,0}^S,v_{i,1}^S)$.
    \item Run $(\ots_1,\regA') \gets \Adv_\secp(\regA,\regR,\ck,\hk)$, and let $\{\cm_i\}_{i \in [\ell]}$ be the commitments that are part of $\ots_1$.
    \item For each $i \in [\ell]$, compute $(\theta_i^R,v_{i,0}^R,v_{i,1}^R) \gets \Ext(\ek,\cm_i)$.
    \item Output 1 if (i) all consistent $i \in T$ are correct, (ii) at least $1/6$ fraction of consistent $i \in \overline{T}$ are incorrect, and (iii) at least 1/5 fraction of $i \in \overline{T}$ are consistent.
\end{itemize}

We claim that $\Pr[\Exp_2 \to 1] = \nonnegl(\secp)$. To establish (ii), note that conditioned on at least $1/6$ fraction of $i \in \overline{T} \cap U$ being incorrect, we have that the expected fraction of \emph{consistent} $i \in \overline{T}$ being incorrect is at least $1/6$, and (iii) follows from a straightforward application of Hoeffding's inequality. Also note that we are measuring the $\regS$ registers in the $\theta_i^S$ bases rather than the $\theta_i^R$ bases, but this make no difference since our claims are only about consistent indices $i$ (where $\theta_i^S = \theta_i^R$). Now, we will show that this experiment can be used to break the correlation intractability of $H$, but first we introduce some notation.

\begin{itemize}
    \item For each $(\ek,\{\theta_i^S,v_{i,0}^S,v_{i,1}^S\}_{i \in [\ell]})$, define the relation $R[\ek,\{\theta_i^S,v_{i,0}^S,v_{i,1}^S\}_{i \in [\ell]}]$ as follows. Recalling that $\ell = c \cdot t$, we will associate each $i \in [\ell]$ with a pair $(\iota,\kappa)$ for $\iota \in [t], \kappa \in [c]$. Also, for each set of strings $\{\cm_i\}_{i \in [\ell]}$, we fix $(\theta_i^R,v_{i,0}^R,v_{i,1}^R) \coloneqq  \Ext(\ek,\cm_i)$ for each $i \in [\ell]$. Then the domain will consist of strings $\{\cm_i\}_{i \in [\ell]}$ such that (i) $|i : (\theta_i^R,v_{i,0}^R,v_{i,1}^R) = (\theta_i^S,v_{i,0}^S,v_{i,1}^S)| \leq (1-1/60)\ell$ and (ii) for each $\iota \in [t], |\kappa : (\theta_{(\iota,\kappa)}^S,v_{(\iota,\kappa),0}^S,v_{(\iota,\kappa),1}^S) = (\theta_{(\iota,\kappa)}^R,v_{(\iota,\kappa),0}^R,v_{(\iota,\kappa),1}^R)| \geq (1/2)c$.
    \item For each $\{\cm_i\}_{i \in [\ell]}$ in the domain of $R[\ek,\{\theta_i^S,v_{i,0}^S,v_{i,1}^S\}_{i \in [\ell]}]$, define the sets $\{S_{\iota,\{\cm_i\}_{i \in [\ell]}}\}_{\iota \in [t]}$ as follows. If $(1/2)c \leq |\kappa : (\theta_{(\iota,\kappa)}^S,v_{(\iota,\kappa),0}^S,v_{(\iota,\kappa),1}^S) = (\theta_{(\iota,\kappa)}^R,v_{(\iota,\kappa),0}^R,v_{(\iota,\kappa),1}^R)| \leq (1-1/120)c$, let $S_{\iota,\{\cm_i\}_{i \in [\ell]}}$ consist of subsets $C \subset [c]$ such that for all $\kappa \in C, (\theta_{(\iota,\kappa)}^S,v_{(\iota,\kappa),0}^S,v_{(\iota,\kappa),1}^S) = (\theta_{(\iota,\kappa)}^R,v_{(\iota,\kappa),0}^R,v_{(\iota,\kappa),1}^R)$. Otherwise, let $S_{\iota,\{\cm\}_{i \in [\ell]}} = \emptyset$. 
    \item Define the set $R[\ek,\{\theta_i^S,v_{i,0}^S,v_{i,1}^S\}_{i \in [\ell]}]_{\{\cm_i\}_{i \in [\ell]}}$ to consist of all $y = (C_1,\dots,C_t)$ such that $C_\iota \in S_{\iota,\{\cm_i\}_{i \in [\ell]}}$ for all $\iota$ such that $S_{\iota,\{\cm_i\}_{i \in [\ell]}} \neq \emptyset$. Noting that there are always at least $\alpha = 1/120 $ fraction of $\iota \in [t]$ such that $S_{\iota,\{\cm_i\}_{i \in [\ell]}} \neq \emptyset$, we see that $R[\ek,\{\theta_i^S,v_{i,0}^S,v_{i,1}^S\}_{i \in [\ell]}]$ is an $\alpha$-approximate efficiently verifiable product relation with sparsity $\rho = \binom{(1-\alpha)c}{(1/2)c}/2^c < \alpha$.
\end{itemize}

Now, whenever $\Exp_2 = 1$, it must be the case that $\{\cm_i\}_{i \in [\ell]}$ is in the domain of $R[\ek,\{\theta_i^S,v_{i,0}^S,v_{i,1}^S\}_{i \in [\ell]}]$, and $T \in R[\ek,\{\theta_i^S,v_{i,0}^S,v_{i,1}^S\}_{i \in [\ell]}]_{\{\cm_i\}_{i \in [\ell]}}$. Thus, we can break correlation intractability as follows. Begin running $\Exp_2$ until right before $\Adv_\secp$ is initialized, except that we don't sample $\hk$. Instead, declare the relation $R[\ek,\{\theta_i^S,v_{i,0}^S,v_{i,1}^S\}_{i \in [\ell]}]$, and receive $\hk$ from the correlation intractability challenger. Then, continue running $\Exp_2$ until $\{\cm_i\}_{i \in [\ell]}$ is obtained, and output this to the challenger. The above analysis shows that this breaks correlation intractability for $R[\ek,\{\theta_i^S,v_{i,0}^S,v_{i,1}^S\}_{i \in [\ell]}]$.

\end{proof}

\begin{claim}
$\cH_3 \approx_s \cH_4$
\end{claim}

\begin{proof}

First, we note that by a standard Hoeffding inequality, it holds that $\Pr[|\overline{T} \setminus U| \geq \ell/5] = 1-\negl(\secp)$, so we will assume that this is the case. Now, in $\cH_3$ we are guaranteed that the state on registers $\{\regS_{i,0},\regS_{i,1}\}_{i \in \overline{T} \setminus U}$ is in the image of \[\Pi[1/6,\{\theta_i^R,v_{i,0}^R,v_{i,1}^R\}_{i \in \overline{T} \setminus U}]\] right before the $\{\regS_{i,0}\}_{i \in \overline{T} \setminus U}$ are measured in the standard basis and $\{\regS_{i,1}\}_{i \in \overline{T} \setminus U}$ are measured in the Hadamard basis to obtain $\{v_{i,0}^S,v_{i,1}^S\}_{i \in \overline{T} \setminus U}$. Thus, setting $b \coloneqq \maj\{\theta_i^R \oplus d_i\}_{i \in \overline{T} \setminus U}$, we see that registers $\{\regS_{i,d_i \oplus b \oplus 1}\}_{i \in \overline{T} \setminus U}$ are in a superposition of at most $2^{(1/6)(\ell/5)} = 2^{\ell/30}$ states in the $\{\theta_i^R\}_{i \in \overline{T} \setminus U}$ basis. Moreover, at least $1/2-1/6 = 1/3$ fraction of these registers will be measured in the conjugate basis to obtain the bits $\{v_{i,d_i \oplus b \oplus 1}^S\}_{i \in \overline{T} \setminus U}$ that comprise $V_{1-b}$. Thus, by \cref{impthm:small-superposition} and \cref{impthm:privacy-amplification}, the string $F_\secp(s,V_{1-b})$ has statistical distance at most $2^{-\frac{1}{2}(\ell/15-\ell/30 - \secp)} < 2^{-\secp}$ conditioned on the adversary's view. Thus, replacing $m_{1-b}$ with $0^\secp$ results in a statistically close distribution.
\end{proof}

\begin{claim}
$\cH_4 \approx_s \cH_5$
\end{claim}

\begin{proof}
Same argument as $\cH_2 \approx_s \cH_3$.
\end{proof}

\end{proof}

\begin{lemma}
The protocol in \cref{fig:two-round} satisfies simulation-based security and equivocal receiver's security against a malicious sender.
\end{lemma}

\begin{proof}
Let $\{\Adv_\secp\}_{\secp \in \bbN}$ be a QPT adversary corrupting the sender, which takes as input register $\regA$ of $\{\ket{\psi_\secp}^{\regA,\regD}\}_{\secp \in \bbN}$.  Instead of explicitly considering a register $\regX$ holding the honest receiver's input, we write their input as a bit $b \in \{0,1\}$. We define a simulator $\{\Sim_\secp\}_{\secp \in \bbN}$ as follows.\\

\noindent\underline{$\Sim_\secp(\regA)$}
\begin{itemize}
    \item Sample $\ck \gets \{0,1\}^h$.
    \item Sample $T \subset [\ell]$ as a uniformly random sequence of $t$ subsets of $[c]$ of size $c/2$. 
    \item For $i \in [\ell]$, sample $\theta_i^R,v_{i,0}^R,v_{i,1}^R \gets \{0,1\}^3$.
    \item For $i \in T$, sample $r_i \gets \{0,1\}^\secp$ and set $\cm_i \coloneqq \Com(\ck,(\theta_i^R,v_{i,0}^R,v_{i,1}^R);r_i)$ and for $i \in \overline{T}$, set $\cm_i \gets \Com(\ck,0)$.
    \item For $i \in T$, initialize register $\regS_{i,0}$ to $H^{\theta_i^R}\ket{v_{i,0}^R}$ and $\regS_{i,1}$ to $H^{\theta_i^R}\ket{v_{i,1}^R}$, and for $i \in \overline{T}$, initialize $\regS_{i,0}$ to $\ket{v_{i,0}^R}$ and $\regS_{i,1}$ to $H\ket{v_{i,1}^R}$.
    \item Sample $\hk \gets \Samp(1^\secp,(\cm_1,\dots,\cm_\ell),T)$.
    \item Set $\ots_1^\C \coloneqq (\{\cm_i\}_{i \in [\ell]},T,\{\theta_i^R,v_{i,0}^R,v_{i,1}^R\}_{i \in T})$ and $\ots_1^\NC \coloneqq \{\theta_i^R\}_{i \in \overline{T}}$.
    \item Run $(\ots_2,\regA') \gets \Adv_\secp(\regA,\{\regS_{i,0},\regS_{i,1}\}_{i \in [\ell]},\ck,\hk,\ots_1^\C,\ots_1^\NC)$.
    \item Let $V_0$ be the concatenation of bits $\{v_{i,\theta_i^R}^R\}_{i \in \overline{T} \setminus U}$ and $V_1$ be the concatenation of bits $\{v_{i,\theta_i^R\oplus 1}^R \}_{i \in \overline{T} \setminus U}$, and compute $m_0 \coloneqq \widetilde{m}_0 \oplus F(s,V_0)$ and $m_1 \coloneqq \widetilde{m}_1 \oplus F(s,V_1)$. Send $(m_0,m_1)$ to the ideal functionality $\cF_\OT$ and output $\regA'$.
\end{itemize}

\noindent Now, we define a sequence of hybrids, and argue indistinguishability between each adjacent pair.

\begin{itemize}
    \item $\cH_0$: This is the real distribution $\Pi_{\cF_\OT}[\Adv_\secp,\ket{\psi_\secp}]$.
    \item $\cH_1$: Same as $\cH_0$ except that $T \subset [\ell]$ is sampled as described in the simulator, and the hash key $\hk$ is sampled at the end of the Sender's computation by $\Samp(1^\secp,(\cm_1,\dots,\cm_\ell),T)$. This is the same distribution as $\cH_0$, which follows from programmability of the correlation-intractable hash function (\cref{def:programmability}).
    \item $\cH_2$: Same as $\cH_1$ except that $\{\cm_i \gets \Com(\ck,0)\}_{i \in \overline{T}}$ are sampled as commitments to 0. Computational indistinguishability from $\cH_1$ follows directly from hiding of the comitment scheme (\cref{def:hiding}).
    \item $\cH_3$: Same as $\cH_2$ except that for $i \in \overline{T}$, register $\regR_{i,0}$ is measured in the standard basis to obtain $v_{i,0}^R$ and register $\regR_{i,1}$ is measured in the Hadamard basis to obtain $v_{i,1}^R$. This is the same distribution as $\cH_2$ since we are only changing the way we sample the bits $\{v_{i,\theta_i^R \oplus 1}\}_{i \in \overline{T}}$, which are outside the adversary's view.
    \item $\cH_4$: Same as $\cH_3$ except that $d_i \coloneqq \theta_i^R$ for $i \in \overline{T}$, and $V$ is set to the concatentation of bits $\{v_{i,\theta_i^R\oplus b}^R\}_{i \in \overline{T} \setminus U}$. This is the same distribution as $\cH_3$ since in $\cH_3$, the bits  $\{\theta_i^R\}_{i \in \overline{T}}$ are sampled uniformly at random and are only used to define the $\{d_i\}_{i \in \overline{T}}$. Also note that for any honest receiver input $b \in \{0,1\}$, this is the same as the simulated distribution $\widetilde{\Pi}_{\cF_\OT}[\Sim_\secp,\ket{\psi_\secp}]$, which completes the proof.
\end{itemize}

\end{proof}

Finally, we define $\Sim_\EQ$ to run the first part of the simulator above, obtaining $\regS \coloneqq (\{\regS_{i,0},\regS_{i,1}\}_{i \in [\ell]},\ck,\hk)$ and $\ots \coloneqq (\ots_1^\C,\ots_1^\NC)$, and setting $\omega_0 \coloneqq \{\theta_i^R,v_{i,\theta_i^R}^R\}_{i \in \overline{T}}, \omega_1 \coloneqq \{\theta_i^R,v_{i,\theta_i^R \oplus 1}^R\}_{i \in \overline{T}}$. The arguments above imply that this definition of $\Sim_\EQ$ establishes \emph{equivocal receiver's security} of the protocol.

\end{proof}

\section*{Acknowledgments}
D. Khurana was supported in part by by NSF CAREER CNS-2238718, DARPA SIEVE
and a gift from Visa Research. This material is based upon work supported by
the Defense Advanced Research Projects Agency through Award HR00112020024. A. Srinivasan is supported in part by a SERB startup grant and Google India Research Award.

\fi

\ifsubmission
\bibliographystyle{plain}
\else
\bibliographystyle{alpha}
\fi

\bibliography{abbrev3,crypto,main,snarg-refs}

\ifsubmission
\newpage
\begin{center}
\textbf{\Huge Supplementary Material}
\end{center}
\vspace{2em}
\else
\fi

\ifsubmission
\appendix

\else\fi
\end{document}